\documentclass{LMCS}

\usepackage{graphicx,amssymb,url,amsmath,amsthm,proof,hyperref,gastex}
\newdimen\proofrulebreadth \proofrulebreadth=.05em
\newdimen\proofdotseparation \proofdotseparation=1.25ex
\newdimen\proofrulebaseline \proofrulebaseline=2ex
\newcount\proofdotnumber \proofdotnumber=3
\let\then\relax
\def\hfi{\hskip0pt plus.0001fil}
\mathchardef\squigto="3A3B
\newif\ifinsideprooftree\insideprooftreefalse
\newif\ifonleftofproofrule\onleftofproofrulefalse
\newif\ifproofdots\proofdotsfalse
\newif\ifdoubleproof\doubleprooffalse
\let\wereinproofbit\relax
\newdimen\shortenproofleft
\newdimen\shortenproofright
\newdimen\proofbelowshift
\newbox\proofabove
\newbox\proofbelow
\newbox\proofrulename
\def\shiftproofbelow{\let\next\relax\afterassignment\setshiftproofbelow\dimen0 }
\def\shiftproofbelowneg{\def\next{\multiply\dimen0 by-1 }%
\afterassignment\setshiftproofbelow\dimen0 }
\def\setshiftproofbelow{\next\proofbelowshift=\dimen0 }
\def\setproofrulebreadth{\proofrulebreadth}

\def\prooftree{% NESTED ZERO (\ifonleftofproofrule)
\ifnum  \lastpenalty=1
\then   \unpenalty
\else   \onleftofproofrulefalse
\fi
\ifonleftofproofrule
\else   \ifinsideprooftree
        \then   \hskip.5em plus1fil
        \fi
\fi
\bgroup% NESTED ONE (\proofbelow, \proofrulename, \proofabove,
\setbox\proofbelow=\hbox{}\setbox\proofrulename=\hbox{}%
\let\justifies\proofover\let\leadsto\proofoverdots\let\Justifies\proofoverdbl
\let\using\proofusing\let\[\prooftree
\ifinsideprooftree\let\]\endprooftree\fi
\proofdotsfalse\doubleprooffalse
\let\thickness\setproofrulebreadth
\let\shiftright\shiftproofbelow \let\shift\shiftproofbelow
\let\shiftleft\shiftproofbelowneg
\let\ifwasinsideprooftree\ifinsideprooftree
\insideprooftreetrue
\setbox\proofabove=\hbox\bgroup$\displaystyle % NESTED TWO
\let\wereinproofbit\prooftree
\shortenproofleft=0pt \shortenproofright=0pt \proofbelowshift=0pt
\onleftofproofruletrue\penalty1
}

\def\eproofbit{% NESTED TWO
\ifx    \wereinproofbit\prooftree
\then   \ifcase \lastpenalty
        \then   \shortenproofright=0pt  % 0: some other object, no indentation
        \or     \unpenalty\hfil         % 1: empty hypotheses, just glue
        \or     \unpenalty\unskip       % 2: just had a tree, remove glue
        \else   \shortenproofright=0pt  % eh?
        \fi
\fi
\global\dimen0=\shortenproofleft
\global\dimen1=\shortenproofright
\global\dimen2=\proofrulebreadth
\global\dimen3=\proofbelowshift
\global\dimen4=\proofdotseparation
\global\count255=\proofdotnumber
$\egroup  % NESTED ONE
\shortenproofleft=\dimen0
\shortenproofright=\dimen1
\proofrulebreadth=\dimen2
\proofbelowshift=\dimen3
\proofdotseparation=\dimen4
\proofdotnumber=\count255
}

\def\proofover{% NESTED TWO
\eproofbit % NESTED ONE
\setbox\proofbelow=\hbox\bgroup % NESTED TWO
\let\wereinproofbit\proofover
$\displaystyle
}%
\def\proofoverdbl{% NESTED TWO
\eproofbit % NESTED ONE
\doubleprooftrue
\setbox\proofbelow=\hbox\bgroup % NESTED TWO
\let\wereinproofbit\proofoverdbl
$\displaystyle
}%
\def\proofoverdots{% NESTED TWO
\eproofbit % NESTED ONE
\proofdotstrue
\setbox\proofbelow=\hbox\bgroup % NESTED TWO
\let\wereinproofbit\proofoverdots
$\displaystyle
}%
\def\proofusing{% NESTED TWO
\eproofbit % NESTED ONE
\setbox\proofrulename=\hbox\bgroup % NESTED TWO
\let\wereinproofbit\proofusing
\kern0.3em$
}

\def\endprooftree{% NESTED TWO
\eproofbit % NESTED ONE
  \dimen5 =0pt% spread of hypotheses
\dimen0=\wd\proofabove \advance\dimen0-\shortenproofleft
\advance\dimen0-\shortenproofright
\dimen1=.5\dimen0 \advance\dimen1-.5\wd\proofbelow
\dimen4=\dimen1
\advance\dimen1\proofbelowshift \advance\dimen4-\proofbelowshift
\ifdim  \dimen1<0pt
\then   \advance\shortenproofleft\dimen1
        \advance\dimen0-\dimen1
        \dimen1=0pt
        \ifdim  \shortenproofleft<0pt
        \then   \setbox\proofabove=\hbox{%
                        \kern-\shortenproofleft\unhbox\proofabove}%
                \shortenproofleft=0pt
        \fi
\fi
\ifdim  \dimen4<0pt
\then   \advance\shortenproofright\dimen4
        \advance\dimen0-\dimen4
        \dimen4=0pt
\fi
\ifdim  \shortenproofright<\wd\proofrulename
\then   \shortenproofright=\wd\proofrulename
\fi
\dimen2=\shortenproofleft \advance\dimen2 by\dimen1
\dimen3=\shortenproofright\advance\dimen3 by\dimen4
\ifproofdots
\then
        \dimen6=\shortenproofleft \advance\dimen6 .5\dimen0
        \setbox1=\vbox to\proofdotseparation{\vss\hbox{$\cdot$}\vss}%
        \setbox0=\hbox{%
                \advance\dimen6-.5\wd1
                \kern\dimen6
                $\vcenter to\proofdotnumber\proofdotseparation
                        {\leaders\box1\vfill}$%
                \unhbox\proofrulename}%
\else   \dimen6=\fontdimen22\the\textfont2 % height of maths axis
        \dimen7=\dimen6
        \advance\dimen6by.5\proofrulebreadth
        \advance\dimen7by-.5\proofrulebreadth
        \setbox0=\hbox{%
                \kern\shortenproofleft
                \ifdoubleproof
                \then   \hbox to\dimen0{%
                        $\mathsurround0pt\mathord=\mkern-6mu%
                        \cleaders\hbox{$\mkern-2mu=\mkern-2mu$}\hfill
                        \mkern-6mu\mathord=$}%
                \else   \vrule height\dimen6 depth-\dimen7 width\dimen0
                \fi
                \unhbox\proofrulename}%
        \ht0=\dimen6 \dp0=-\dimen7
\fi
\let\doll\relax
\ifwasinsideprooftree
\then   \let\VBOX\vbox
\else   \ifmmode\else$\let\doll=$\fi
        \let\VBOX\vcenter
\fi
\VBOX   {\baselineskip\proofrulebaseline \lineskip.2ex
        \expandafter\lineskiplimit\ifproofdots0ex\else-0.6ex\fi
        \hbox   spread\dimen5   {\hfi\unhbox\proofabove\hfi}%
        \hbox{\box0}%
        \hbox   {\kern\dimen2 \box\proofbelow}}\doll%
\global\dimen2=\dimen2
\global\dimen3=\dimen3
\egroup % NESTED ZERO
\ifonleftofproofrule
\then   \shortenproofleft=\dimen2
\fi
\shortenproofright=\dimen3
\onleftofproofrulefalse
\ifinsideprooftree
\then   \hskip.5em plus 1fil \penalty2
\fi
}

\usepackage{enumerate}

\newtheorem{theorem}{Theorem}[section]
\newtheorem{proposition}[theorem]{Proposition}
\newtheorem{lemma}[theorem]{Lemma}
\newtheorem{corollary}[theorem]{Corollary}

\theoremstyle{definition}

\newtheorem{defn}[theorem]{Definition}
\newtheorem{example}[theorem]{Example}

\newcommand{\pow}[1]{\mathcal{P}(#1)}
\newcommand{\defeq}{=_{\mbox{\scriptsize{def}}}}

\newcommand{\nat}{\mathbb{N}}

\newcommand{\integers}{\mathbb{Z}}

\newcommand{\partialfn}{\rightharpoonup}
\renewcommand{\implies}{\rightarrow}

\newcommand{\vars}{\mathcal{V}}
\newcommand{\monoid}{\langle R,\circ,e \rangle}
\newcommand{\cbimodel}{\langle R,\circ,e,-,\infty \rangle}
\newcommand{\mbimodel}{\langle R,\circ,\myerightarrow,e,-,\infty \rangle}
\newcommand{\abgroup}{\langle R,\circ,e,-\rangle}
\newcommand{\BI}{\mathrm{BI}}
\newcommand{\BBI}{\mathrm{BBI}}

\newcommand{\CBI}{\mathrm{CBI}}

\newcommand{\ML}{\mathrm{ML}}
\newcommand{\CLL}{\mathrm{CLL}}
\newcommand{\ILL}{\mathrm{ILL}}
\newcommand{\langs}[1]{\mathcal{L}({#1})}

\newcommand{\LAX}{\mathrm{LAX}}
\newcommand{\AX}{\mathrm{AX}}
\newcommand{\displayCBI}{\mathrm{DL}_{\mathrm{CBI}}}

\newcommand{\aemp}{\emptyset}
\newcommand{\memp}{\varnothing}
\newcommand{\ainv}{\sharp}
\newcommand{\minv}{\flat}

\newcommand{\true}{\top}
\newcommand{\false}{\bot}
\newcommand{\mtrue}{\true^*}
\newcommand{\mfalse}{\false^{\!*}}
\newcommand{\mneg}{\mathord{\sim}}
\newcommand{\mor}{\mathrel{\stackrel{\vcenter{\offinterlineskip
  \hbox{\scriptsize{$\ast$}}\vskip-1.5ex}}{\vee}}}
\newcommand{\wand}{\mathrel{\hbox{---}\llap{$\ast$}}}

\newcommand{\inv}{\mathord{-}}
\newcommand{\dlantform}[1]{\Psi_{#1}}
\newcommand{\dlconform}[1]{\Upsilon_{#1}}

\newcommand{\embed}[1]{\ulcorner{#1}\urcorner}
\newbox\gnBoxA
\newdimen\gnCornerHgt
\setbox\gnBoxA=\hbox{$\llcorner$}
\global\gnCornerHgt=\ht\gnBoxA
\newdimen\gnArgHgt

\def\revembed #1{%
    \setbox\gnBoxA=\hbox{$#1$}%
    \gnArgHgt=\ht\gnBoxA%
    \ifnum \gnArgHgt<\gnCornerHgt
        \gnArgHgt=0pt%
    \else
        \advance \gnArgHgt by -\gnCornerHgt%
    \fi
    \lower\gnArgHgt\hbox{$\llcorner\!\!$} \box\gnBoxA %
        \lower\gnArgHgt\hbox{$\!\!\lrcorner$}}

\newcommand{\seq}[2]{#1 \vdash #2}

\newcommand{\drvequiv}[2]{#1 \dashv\vdash #2}
\newcommand{\sat}[2]{#1 \models_\rho #2}
\newcommand{\notsat}[2]{#1 \not\models_\rho #2}

\newcommand{\modusponens}{\mbox{(MP)}}
\newcommand{\subst}{\mbox{(Subst)}}
\newcommand{\invfalse}{\mbox{($\inv\false$)}}
\newcommand{\circfalse}{\mbox{($\circ\false$)}}
\newcommand{\lollyfalse}{\mbox{($\myerightarrow\false$)}}
\newcommand{\invdisj}{\mbox{($\inv\vee$)}}
\newcommand{\circdisj}{\mbox{($\circ\vee$)}}
\newcommand{\lollydisjl}{\mbox{($\myerightarrow\vee$L)}}
\newcommand{\lollydisjr}{\mbox{($\myerightarrow\vee$R)}}

\newcommand{\ax}{\mbox{(Id)}}

\newcommand{\cut}{\mbox{(Cut)}}
\newcommand{\truel}{\mbox{($\true$L)}}
\newcommand{\truer}{\mbox{($\true$R)}}
\newcommand{\falsel}{\mbox{($\false$L)}}
\newcommand{\falser}{\mbox{($\false$R)}}
\newcommand{\mtruel}{\mbox{($\mtrue$L)}}
\newcommand{\mtruer}{\mbox{($\mtrue$R)}}
\newcommand{\mfalsel}{\mbox{($\mfalse$L)}}
\newcommand{\mfalser}{\mbox{($\mfalse$R)}}
\newcommand{\negr}{\mbox{($\neg$R)}}
\newcommand{\negl}{\mbox{($\neg$L)}}
\newcommand{\andr}{\mbox{($\wedge$R)}}
\newcommand{\andl}{\mbox{($\wedge$L)}}
\newcommand{\orr}{\mbox{($\vee$R)}}

\newcommand{\orl}{\mbox{($\vee$L)}}
\newcommand{\impr}{\mbox{($\rightarrow$R)}}
\newcommand{\impl}{\mbox{($\rightarrow$L)}}
\newcommand{\starl}{\mbox{($*$L)}}
\newcommand{\starr}{\mbox{($*$R)}}
\newcommand{\wandl}{\mbox{($\wand$L)}}
\newcommand{\wandr}{\mbox{($\wand$R)}}
\newcommand{\morl}{\mbox{($\,\mor$L)}}
\newcommand{\morr}{\mbox{($\mor$R)}}
\newcommand{\mnegl}{\mbox{($\mneg$L)}}
\newcommand{\mnegr}{\mbox{($\mneg$R)}}

\newcommand{\myerightarrow}{\mathbin{- \! {\bullet}}}
\newcommand{\linimp}{\mathbin{- \! {\circ}}}

\newcommand{\adl}[2]{\mbox{(AD{#1}{#2})}}
\newcommand{\mdl}[2]{\mbox{(MD{#1}{#2})}}
\newcommand{\display}{\mbox{($\displayeq$)}}
\newcommand{\displayeq}{\equiv_{D}}

\newcommand{\weakr}{\mbox{(WkR)}}
\newcommand{\weakl}{\mbox{(WkL)}}

\newcommand{\contrr}{\mbox{(CtrR)}}
\newcommand{\contrl}{\mbox{(CtrL)}}
\newcommand{\aassocl}{\mbox{(AAL)}}
\newcommand{\aassocr}{\mbox{(AAR)}}
\newcommand{\massocl}{\mbox{(MAL)}}
\newcommand{\massocr}{\mbox{(MAR)}}

\newcommand{\aunitl}{\mbox{($\aemp$L)}}
\newcommand{\aunitr}{\mbox{($\aemp$R)}}
\newcommand{\munitl}{\mbox{($\memp$L)}}
\newcommand{\munitr}{\mbox{($\memp$R)}}

\def\doi{6 (3:3) 2010}
\lmcsheading%
{\doi}
{1--42}
{}
{}
{Aug.~\phantom01, 2009}
{Jul.~20, 2010}
{}   

\begin{document}

\title[Classical $\BI$]{Classical $\BI$: Its Semantics and Proof Theory}

\author[J.~Brotherston]{James Brotherston\rsuper a}
\address{Dept.\ of Computing, Imperial College London, UK}
\email{J.Brotherston@imperial.ac.uk, ccris@doc.ic.ac.uk}
\thanks{{\lsuper a}Research supported by an EPSRC Postdoctoral Fellowship.}

\author[C.~Calcagno]{Cristiano Calcagno\rsuper b}
\address{\vskip-6 pt}
\thanks{{\lsuper b}Research supported by an EPSRC Advanced Fellowship.}


\begin{abstract}
We present \emph{Classical $\BI$} ($\CBI$), a new addition to the
family of \emph{bunched logics} which originates in O'Hearn and Pym's
logic of bunched implications $\BI$. $\CBI$ differs from
existing bunched logics in that its multiplicative connectives
behave classically rather than intuitionistically (including in
particular a multiplicative version of classical negation). At
the semantic level, $\CBI$-formulas have the normal bunched
logic reading as declarative statements about resources, but
its resource models necessarily feature more structure than
those for other bunched logics; principally, they satisfy the
requirement that every resource has a unique dual.  At the
proof-theoretic level, a very natural formalism for $\CBI$ is
provided by a display calculus \emph{\`a la} Belnap, which can
be seen as a generalisation of the bunched sequent calculus
for $\BI$. In this paper we formulate the aforementioned
model theory and proof theory for $\CBI$, and prove some
fundamental results about the logic, most notably completeness
of the proof theory with respect to the semantics.
\end{abstract}

\maketitle

\section{Introduction}
\label{sec:introduction} Substructural logics, whose best-known
varieties include linear logic, relevant logic and the Lambek
calculus, are characterised by their restriction of the use of
the so-called \emph{structural} proof principles of classical
logic~\cite{Restall:00}. These may be roughly characterised as
those principles that are insensitive to the syntactic form of
formulas, chiefly weakening (which permits the introduction of
redundant premises into an argument) and contraction (which
allows premises to be arbitrarily duplicated). For example, in
linear logic, only formulas prefixed with a special
``exponential'' modality are subject to weakening and
contraction, while in relevant logic it is usual for
contraction but not weakening to be permitted.

\emph{Bunched logic} is a relatively new area of substructural logic, but one that has
been receiving increasing attention amongst the logical and
computer science research communities in recent years.
In bunched logic, the restriction on the use of structural proof principles is achieved by allowing
the connectives of a standard ``additive'' propositional logic,
which admits weakening and contraction, to be freely combined
with those of a second ``multiplicative'' propositional logic,
which does not. In contrast to linear logic, whose restricted
treatment of additive connectives yields a natural
constructive reading of proofs as
computations~\cite{Abramsky:93}, the inclusion of unrestricted
additives in bunched logics gives rise to a simple Kripke-style
truth interpretation according to which formulas can be
understood as declarative statements about
\emph{resource}~\cite{Pym-OHearn-Yang:04}. This resource
reading of bunched logic has found substantial application in
computer science, most notably in the shape of \emph{separation
logic}, which is a Hoare logic for program verification based upon various bunched logic
models of heap memory~\cite{Reynolds:02}.  The proof theory of bunched logic also differs markedly from the proof theory of linear logic, which is typically formulated in terms of sequent calculi whose sequents have the usual flat context structure based on lists or (multi)sets.  However, since bunched logics contain both an (unrestricted) additive logic and a multiplicative one, proof systems for bunched logic employ both additive and multiplicative structural connectives for  forming contexts (akin to the comma in standard sequent calculus).  This gives rise to proof judgements whose contexts are \emph{trees} --- originally termed ``bunches'' --- built from structural connectives and formulas.

Although the main ideas necessary to develop bunched logic can
retrospectively be seen to have been present in earlier work on
relevant logics, it first emerged fairly recently with the
introduction of $\BI$, O'Hearn and Pym's \emph{logic of bunched
implications}~\cite{OHearn-Pym:99}. Semantically, $\BI$ can be
seen to arise by considering the structure of cartesian doubly
closed categories --- i.e.\ categories with one cartesian closed
structure and one symmetric monoidal closed
structure~\cite{Pym:02}. Concretely, such categories
correspond to a combination of standard intuitionistic logic
with multiplicative intuitionistic linear logic\footnote{We
refer here to linear logic without the exponentials.}
($\mathrm{MILL}$), and thus one has the following propositional
connectives\footnote{$\mtrue$, which is the unit of $*$, is
often elsewhere written $I$.} for $\BI$:

\[\begin{array}{l@{\hspace{0.5cm}}c@{\hspace{0.5cm}}c@{\hspace{0.5cm}}c@{\hspace{0.5cm}}c@{\hspace{0.5cm}}c@{\hspace{0.5cm}}c}
\mbox{Additive:} & \true & \false & \neg & \wedge & \vee & \rightarrow \\
\mbox{Multiplicative:} & \;\mtrue & & & * & & \wand
\end{array}\]
(where $\neg$ is the intuitionistic negation defined by $\neg F
= F \implies \false$). As well as the semantics based on the
aforementioned categories, $\BI$ can be given an algebraic
semantics: one simply requires that the algebraic structure for
$\BI$ has both the Heyting algebra structure required to
interpret intuitionistic logic, and the residuated commutative
monoid structure required to interpret $\mathrm{MILL}$.  By
requiring a Boolean algebra instead of the Heyting algebra, one
obtains the variant logic Boolean $\BI$ ($\BBI$), which can be
seen as a combination of classical logic and
$\mathrm{MILL}$~\cite{Pym-OHearn-Yang:04,Pym:02}.  Most of the
computer science applications of bunched logic are in fact
based on $\BBI$ rather than $\BI$; for example, the heap model
used in separation logic is a model of
$\BBI$~\cite{Ishtiaq-OHearn:01}.

A natural question from a logician's standpoint is whether
bunched logics exist in which the multiplicative connectives
behave classically, rather than intuitionistically (and do not
simply collapse into their additive equivalents).  A computer
scientist might also enquire whether such a logic could, like
its siblings, be understood semantically in terms of resource.
In this paper, we address these questions by presenting a new
addition to the bunched logic family, which we call
\emph{Classical $\BI$} ($\CBI$), and whose additives and
multiplicatives both behave classically.  In particular, $\CBI$
features multiplicative analogues of the additive falsity,
negation, and disjunction, which are absent in the other
bunched logics. Thus $\CBI$ can be seen as a combination of
classical logic and multiplicative classical linear logic
($\mathrm{MLL}$). We examine $\CBI$ both from the
model-theoretic and the proof-theoretic perspective, each of
which we describe below.

\paragraph{\em Model-theoretic perspective:}  From the point of view of
computer science, the main interest of bunched logic stems from
its Kripke-style frame semantics based on relational
commutative monoids, which can be understood as an abstract
representation of
resource~\cite{Galmiche-Mery-Pym:05,Galmiche-Larchey-Wendling:06}.
In such models, formulas of bunched logic have a natural
declarative reading as statements about resources (i.e.\ monoid
elements). Thus the multiplicative unit $\mtrue$ denotes the
empty resource (i.e.\ the monoid identity element) and a
multiplicative conjunction $F * G$ of two formulas denotes
those resources which divide, via the monoid operation, into
two component resources satisfying respectively $F$ and $G$.
The multiplicative implication $\wand$ then comes along
naturally as the right-adjoint of the multiplicative
conjunction $*$, so that $F \wand G$ denotes those resources with the property that, when they are extended with a resource satisfying $F$, this extension satisfies $G$.

The difference between intuitionistic and classical logics can
be seen as a matter of the differing strengths of their
respective negations~\cite{Prawitz:65}. From this viewpoint the
main obstacle to formulating a bunched logic like $\CBI$ is in
giving a convincing account of classical multiplicative
negation; multiplicative falsity can then be obtained as the
negation of $\mtrue$ and multiplicative disjunction as the de
Morgan dual of $*$.  We show that multiplicative negation can
be given a declarative resource reading just as for the usual
bunched logic connectives, provided that we enrich the
relational commutative monoid structure of $\BBI$-models with
an involutive operator (which interacts with the binary monoid
operation in a suitable fashion). Thus every resource in a
$\CBI$-model is required to have a unique dual. In particular, every
Abelian group can be seen as a $\CBI$-model by taking the dual of an element to be its group inverse. Our interpretation
of multiplicative negation $\mneg$ is then in the tradition of
Routley's interpretation of negation in relevant
logic~\cite{Routley-Routley:72,Dunn:93}: a resource satisfies
$\mneg F$ iff its dual fails to satisfy $F$. This
interpretation, which at first sight may seem unusual, is
justified by the desired semantic equivalences between
formulas. For example, under our interpretation $F \wand G$ is
semantically equivalent to $\mneg F \mor G$, where $\mor$
denotes the multiplicative disjunction.

In Section 2 we state the additional conditions on $\BBI$-models
qualifying them as $\CBI$-models and examine some fundamental
properties of these models.  We then give the forcing semantics
for $\CBI$-formulas with respect to our models, and compare the
resulting notion of validity with that for $\BBI$. Our most
notable result about validity is that $\CBI$ is a non-conservative extension
of $\BBI$, which indicates that $\CBI$ is intrinsically
different in character to its bunched logic siblings, and
justifies independent consideration.

\paragraph{\em Proof-theoretic perspective:} The proof theory of $\BI$ (cf.\
\cite{Pym:02,OHearn-Pym:99}) can be motivated by the
observation that the presence of two implications $\rightarrow$
and $\wand$ should give rise to two context-forming operations,
which correspond to the conjunctions $\wedge$ and $*$ at the
meta-level.  This situation is illustrated by the following
(intuitionistic) sequent calculus right-introduction rules for
the implications:
\[\begin{array}{c@{\hspace{1.0cm}}c}
\begin{prooftree}
  \seq{\Gamma; F_1}{F_2}
  \justifies
  \seq{\Gamma}{F_1 \rightarrow F_2}
  \using \impr
\end{prooftree}
&
\begin{prooftree}
  \seq{\Gamma,F_1}{F_2}
  \justifies
  \seq{\Gamma}{F_1 \wand F_2}
  \using \wandr
\end{prooftree}
\end{array}\]
For similar reasons, there should also be two different ``empty contexts'' or structural units,
which are the structural equivalents of $\true$ and $\mtrue$ respectively.
Accordingly, the contexts $\Gamma$ on the left-hand side of the
sequents in the rules above are not sets or sequences, as in
standard sequent calculi, but rather \emph{bunches}: trees
whose leaves are formulas or structural units and whose internal nodes are either
semicolons or commas. The crucial
difference between the latter two operations is that weakening and
contraction are possible for the additive semicolon but not for
the multiplicative comma.  Since $\BI$ is intuitionistic
in both its additive and multiplicative components,
bunches arise only on the left-hand side of sequents, with a
single formula on the right.  In order to take into account the
bunched contexts in $\BI$ sequents, the left-introduction rules
for logical connectives are then formulated so as to apply at
arbitrary positions within a bunch\footnote{In this respect,
the $\BI$ sequent calculus resembles calculi for \emph{deep
inference}~\cite{Brunnler:06}. However, deep inference calculi
differ substantially from sequent calculi in that they
abandon the distinction between logical and structural connectives, and thus technically they are more akin to term rewriting systems.}. E.g., the
left-introduction rules for the two implications can be
formulated as:
\[\begin{array}{c@{\hspace{1cm}}c}
\begin{prooftree}
\seq{\Delta}{F_1}
\phantom{w}
\seq{\Gamma(F_2)}{F}
\justifies
\seq{\Gamma(\Delta ; F_1 \implies F_2)}{F}
\using \impl
\end{prooftree}
&
\begin{prooftree}
\seq{\Delta}{F_1}
\phantom{w}
\seq{\Gamma(F_2)}{F}
\justifies
\seq{\Gamma(\Delta , F_1 \wand F_2)}{F}
\using \wandl
\end{prooftree}
\end{array}\]
where $\Gamma(\Delta)$ denotes a bunch $\Gamma$ with a
distinguished sub-bunch occurrence $\Delta$.  In contrast, the
right-introduction rules need take into account only the top
level of bunches, as in the right-introduction rules above for
the implications.

For a classical bunched logic like $\CBI$, it would appear
natural from a proof-theoretic perspective to consider a full
two-sided sequent calculus, in which semicolon and comma in
bunches on the right of sequents correspond to the additive and multiplicative
disjunctions.  Unfortunately, it is far from
clear whether there exists such a sequent calculus admitting
cut-elimination, or a similar natural deduction system
satisfying normalisation
(see~\cite{Brotherston:10,Pym:02} for some discussion
of the difficulties).

In Section~\ref{sec:display_BI}, we address this rather
unsatisfactory situation by formulating a \emph{display
calculus} proof system for $\CBI$ that satisfies
cut-elimination, with an attendant subformula property for
cut-free proofs.  Display calculi were first introduced in the setting of
Belnap's \emph{display logic}~\cite{Belnap:82}, which is a generalised
framework that can be instantiated to give consecution calculi
\emph{\`a la} Gentzen for a wide class
of logics. Display calculi are characterised by the fact that any proof judgement may always be rearranged so that a chosen structure occurrence appears alone on one side of the proof turnstile. Remarkably, Belnap also showed that cut-elimination is guaranteed for any display calculus whose proof rules satisfy 8 simple syntactic conditions.  It is a straightforward matter to instantiate Belnap's display logic so as to obtain a display calculus for $\CBI$, and to show that it meets the conditions for cut-elimination. Moreover, our display calculus
is sound and complete with respect to validity in our class of
$\CBI$-models. Soundness follows by showing directly that each
of the proof rules preserves $\CBI$-validity. The proof of completeness, which is presented in
Section~\ref{sec:completeness}, is by reduction to a
completeness result for modal logic due to Sahlqvist.

\paragraph{\em Applications:} Bunched logic (especially $\BBI$) and its
resource semantics has found application in several areas of
computer science, including polymorphic
abstraction~\cite{Collinson-Pym-Robinson:08}, type systems for
reference update and disposal~\cite{Berdine-OHearn:06}, context
logic for tree update~\cite{Calcagno-Gardner-Zarfaty:07} and,
most ubiquitously, separation logic~\cite{Reynolds:02} which
forms the basis of many contemporary approaches to reasoning
about pointer programs (recent examples include
\cite{Parkinson-Bierman:08,Chin-etal:08,Chang-Rival:08}).

Unfortunately, the fact that $\CBI$ is a non-conservative
extension of $\BBI$ appears to rule out the naive use of $\CBI$
for reasoning directly about some $\BBI$-models such as the
separation logic heap model, which is not a $\CBI$-model. On
the other hand, non-conservativity indicates that $\CBI$ can
reasonably be expected to have different applications to those
of $\BI$ and $\BBI$. In Section~\ref{sec:examples} we
consider a range of example $\CBI$-models drawn from quite
disparate areas of mathematics and computer science, including
bit arithmetic, regular languages, money, generalised heaps
 and fractional permissions. In Section~\ref{sec:conclusion} we suggest some
directions for future applications of $\CBI$, and discuss some
related work.

This paper is a revised and expanded version
of~\cite{Brotherston-Calcagno:09}, including several new
results. We have endeavoured to include detailed proofs where
space permits.

\section{Frame semantics and validity for $\CBI$}
\label{sec:models}

In this section we define $\CBI$, a fully classical bunched
logic featuring additive and multiplicative versions of all the
usual propositional connectives (cf.~\cite{Pym:02}), via a
class of Kripke-style frame models.  We also compare the resulting notion of $\CBI$-validity with validity in $\BBI$.

Our $\CBI$-models are based on the relational commutative
monoids used to model
$\BBI$~\cite{Galmiche-Larchey-Wendling:06,Calcagno-Gardner-Zarfaty:07}.
In fact, they are special cases of these monoids, containing
extra structure: an involution operation `$\inv$' on elements
and a distinguished element $\infty$ that characterises the
result of combining an element with its involutive dual. We
point the reader to Section~\ref{sec:examples} for a range of
examples of such models.

In the following, we first recall the usual frame models of $\BBI$, and then give
the additional conditions required for such models to be $\CBI$-models.
Note that we write $\pow{X}$ for the powerset of a
set $X$.

\begin{defn}[$\BBI$-model]
\label{defn:BBI_model} A \emph{$\BBI$-model} is a relational
commutative monoid, i.e.\ a tuple $\monoid$, where $e \in R$
and $\circ: R \times R \rightarrow \pow{R}$ are such that
$\circ$ is commutative and associative, with $r \circ e =
\{r\}$ for all $r \in R$.   Associativity of $\circ$ is
understood with respect to its pointwise extension to $\pow{R}
\times \pow{R} \rightarrow \pow{R}$, given by $X \circ Y \defeq
\bigcup_{x \in X,y \in Y} x \circ y$.
\end{defn}

Note that we could equally well represent the operation $\circ$
in a $\BBI$-model $\monoid$ as a ternary relation, i.e.\ $\circ
\subseteq R \times R \times R$, as is typical for the frame
models used for modal logic~\cite{Blackburn-deRijke-Venema:01}
and relevant logic~\cite{Restall:00}. We view $\circ$ as a
binary function with type $R \times R \rightarrow \pow{R}$
because $\BBI$-models are typically understood as abstract
models of \emph{resource}, in which $\circ$ is understood as a
(possibly non-deterministic) way of combining resources from
the set $R$.

\begin{defn}[$\CBI$-model]
\label{defn:CBI_model} A \emph{$\CBI$-model} is given by a
tuple $\cbimodel$, where $\monoid$ is a $\BBI$-model and $\inv
: R \rightarrow R$ and $\infty \in R$ are such that, for each
$x \in R$, $\inv x$ is the unique element of $R$ satisfying
$\infty \in x \circ \inv x$. We extend `$\inv$' pointwise to
$\pow{R} \rightarrow \pow{R}$ by $\inv X
\defeq \{\inv x \mid x \in X\}$.
\end{defn}

We remark that, in our original definition of
$\CBI$-models~\cite{Brotherston-Calcagno:09}, both $\infty$ and
$\inv{x}$ for $x \in R$ were defined as subsets of $R$, rather
than elements of $R$. However, under such circumstances both
$\inv{x}$ and $\infty$ are forced to be singleton sets by the
other conditions on $\CBI$-models\footnote{In fact, $\infty$ is
forced to be a singleton set because our models employ a single
unit $e$ and we have $\infty = \inv e$ (see
Prop~\ref{prop:cbimodel_properties}).  It is, however, possible
to generalise our $\BBI$-models to multi-unit models employing
a set of units $E \subseteq R$ such that $x \circ E = \{x\}$
(cf.~\cite{Dockins-etal:09,Brotherston-Kanovich:10}).
Then, in the corresponding definition of $\CBI$-model, we have
$\infty \subseteq R$ is not a singleton in general and $\inv x$
is required to be the unique element in $R$ with $\infty \cap
(x \circ \inv x) \neq \emptyset$. However, as we shall show in
Section~\ref{sec:completeness}, $\CBI$ is already complete with
respect to the class of single-unit models provided by our
Definition~\ref{defn:CBI_model}.
}.
Thus there is no loss of generality in requiring $\inv{x}$ and
$\infty$ to be elements of $R$.

\begin{proposition}[Properties of $\CBI$-models]
\label{prop:cbimodel_properties} If $\cbimodel$ is a
$\CBI$-model then:
\begin{enumerate}
\item\label{propitem:doubleinv} $\forall x \in R.\;
    \inv\inv x = x$;

\item\label{propitem:infty} $\inv e=\infty$;

\item\label{propitem:cancel} $\forall x,y,z \in R.\; z \in
    x \circ y$ iff $\inv x \in y \circ \inv z$ iff $\inv y
    \in x \circ \inv z$.
\end{enumerate}
\end{proposition}

\begin{proof}\noindent
\begin{enumerate}
\item By definition of $\CBI$-models, and using
    commutativity of $\circ$, we have $\infty \in \inv x
    \circ x$. However, again by definition, $\inv\inv x$ is
    the unique $y \in R$ such that $\infty \in \inv x \circ
    y$. Thus we must have $\inv\inv x = x$.

\item We have that $\inv e$ is the unique $y \in R$ such
    that $\infty \in e \circ y$.  Since $\infty \in
    \{\infty\} = e \circ \infty$ by definition, we have
    $\inv e = \infty$.

\item We prove that the two bi-implications hold by showing
    three implications.  Suppose first that $z \in x \circ
    y$. Using associativity of $\circ$, we have:
    \[\infty \in z \circ \inv z \subseteq (x \circ y)
    \circ \inv z = x \circ (y \circ \inv z)\] Since $\inv
    x$ is the unique $w \in R$ such that $\infty \in x
    \circ w$, we must have $\inv x \in y \circ \inv z$.

    For the second implication, suppose that $\inv x \in y
    \circ \inv z$.  By the first implication and
    part~\ref{propitem:doubleinv} above and commutativity
    of $\circ$, we then have as required:
    \[\inv y \in \inv z \circ \inv\inv x = \inv\inv x \circ
    \inv z = x \circ \inv z\]

    Finally, for the third implication, suppose that $\inv
    y \in x \circ \inv z$. Using the first and second
    implications together we obtain $\inv\inv z \in y \circ
    \inv\inv x$, i.e.\ $z \in x \circ y$ as required. This
    completes the proof.
\end{enumerate}
\end{proof}

We note that for any $\CBI$-model $\cbimodel$ based on a fixed
underlying $\BBI$-model $\monoid$, part~\ref{propitem:infty} of
Proposition~\ref{prop:cbimodel_properties} implies that the
element $\infty$ is determined by the choice of `$\inv$', while
the $\CBI$-model axiom in Definition~\ref{defn:CBI_model}
ensures that, conversely, `$\inv$' is determined by the choice of
$\infty$.  We include both `$\inv$' and $\infty$ in our model
definition only for convenience.

We now define the syntax of formulas of $\CBI$, and their
interpretation inside our $\CBI$-models.  We assume a fixed, countably infinite set
$\vars$ of propositional variables.

\begin{defn}[$\CBI$-formula]
\label{defn:CBI_formula} \emph{Formulas} of $\CBI$ are given by
the following grammar:
\[\begin{array}{r@{\hspace{0.2cm}}l}
F ::= & P \mid \true \mid \false \mid \neg F \mid F \wedge F \mid F \vee F \mid
F \rightarrow F \mid \\
& \mtrue \mid \mfalse \mid \mneg F \mid F * F \mid F \mor F \mid F \wand F\;
\end{array}\]
where $P$ ranges over $\vars$.  We treat the negations $\neg$
and $\mneg$ as having greater precedence than the other
connectives, and use parentheses to disambiguate where
necessary. As usual, we write $F \leftrightarrow G$ as an
abbreviation for $(F \implies G) \wedge (G \implies F)$.
\end{defn}

We remark that the connectives of $\CBI$-formulas are the
standard connectives of $\BBI$-formulas, plus a multiplicative
falsity $\mfalse$, negation $\mneg$ and disjunction $\mor$. In
order to define the interpretation of our formulas in a given
model, we need as usual environments which interpret the
propositional variables, and a satisfaction or ``forcing''
relation which interprets formulas as true or false relative to
model elements in a given environment.

\begin{defn}[Environment]
\label{defn:environment} An \emph{environment} for either a
$\CBI$-model $\cbimodel$ or a $\BBI$-model $\monoid$ is a
function $\rho: \vars \rightarrow \pow{R}$ interpreting
propositional variables as subsets of $R$.  An environment for
a model $M$ will sometimes be called an \emph{$M$-environment}.
\end{defn}

\begin{defn}[$\CBI$ satisfaction relation]
\label{defn:CBI_satisfaction} Let $M=\cbimodel$ be a
$\CBI$-model.  \emph{Satisfaction} of a $\CBI$-formula $F$ by
an $M$-environment $\rho$ and an element $r \in R$ is denoted
$\sat{r}{F}$ and defined by structural induction on $F$ as
follows:
\[\begin{array}{r@{\hspace{0.5cm}}c@{\hspace{0.5cm}}l}
\sat{r}{P} & \Leftrightarrow & r \in \rho(P) \\
\sat{r}{\true} & \Leftrightarrow & \mbox{always} \\
\sat{r}{\false} & \Leftrightarrow & \mbox{never} \\
\sat{r}{\neg F} & \Leftrightarrow & \notsat{r}{F} \\
\sat{r}{F_1 \wedge F_2} & \Leftrightarrow & \sat{r}{F_1} \mbox{ and } \sat{r}{F_2} \\
\sat{r}{F_1 \vee F_2} & \Leftrightarrow & \sat{r}{F_1} \mbox{ or } \sat{r}{F_2} \\
\sat{r}{F_1 \rightarrow F_2} & \Leftrightarrow & \sat{r}{F_1} \mbox{ implies } \sat{r}{F_2} \\
\sat{r}{\mtrue} & \Leftrightarrow & r = e \\
\sat{r}{\mfalse} & \Leftrightarrow & r \neq \infty \\
\sat{r}{\mneg F} & \Leftrightarrow & \notsat{\inv r}{F} \\
\sat{r}{F_1 * F_2} & \Leftrightarrow & \exists r_1,r_2 \in R.\ r \in r_1 \circ r_2 \mbox{ and } \sat{r_1}{F_1} \mbox{ and } \sat{r_2}{F_2} \\
\sat{r}{F_1 \mor F_2} & \Leftrightarrow & \forall r_1,r_2 \in R.\ \inv r \in r_1 \circ r_2 \mbox{ implies } \sat{{\inv r}_1}{F_1} \mbox{ or } \sat{{\inv r}_2}{F_2} \\
\sat{r}{F_1 \wand F_2} & \Leftrightarrow & \forall r',r'' \in R.\ r''\in r \circ r' \mbox{ and } \sat{r'}{F_1} \mbox{ implies } \sat{r''}{F_2}
\end{array}\]
\end{defn}

We remark that the above satisfaction relation for $\CBI$ is
just an extension of the standard satisfaction relation for
$\BBI$ with the clauses for $\mfalse$, $\mneg$ and $\mor$.  The
interpretations of $\mfalse$ and $\mor$, however, may be
regarded as being determined by the interpretation of the
multiplicative negation $\mneg$ since, as we expect the
classical relationships between multiplicative connectives to
hold, we may simply define $\mfalse$ to be $\mneg\mtrue$ and $F
\mor G$ to be $\mneg(\mneg F * \mneg G)$.  The interpretation
of $\mneg$ itself will not surprise readers familiar with
relevant logics, since negation there is usually semantically
defined by the clause:
\[
x \models \mneg A \; \Leftrightarrow \;x^* \not\models A
\]
where $x$ and $x^*$ are points in a model related by the
somewhat notorious ``Routley star'', the philosophical
interpretation of which has been the source of some angst for
relevant logicians (see e.g.~\cite{Restall:99} for a
discussion).  In the setting of $\CBI$, the involution
operation `$\inv$' in a $\CBI$-model plays the role of the
Routley star.  A more prosaic reason for our interpretation of
$\mneg$ is that it yields the expected semantic equivalences
between formulas.  Other definitions such as, e.g., the superficially
appealing $\sat{r}{\mneg F} \Leftrightarrow
\sat{\inv r}{F}$ do not work, because the model operation `$\inv$' does not itself behave like a negation (it is not antitonic with respect to entailment, for instance). For example, in analogy to
ordinary classical logic, we would expect that $\sat{r}{F \wand
G}$ iff $\sat{r}{\mneg(F
* \mneg G)}$.  However, satisfaction of $\wand$ involves
universal quantification while satisfaction of $*$ involves
existential quantification, strongly suggesting that the
incorporation of a Boolean negation into $\mneg$ is necessary
to ensure such an outcome.  One can also observe that the
following is true in any $\CBI$-model:
\[\begin{array}{rcl}
\sat{\inv r}{F} & \Leftrightarrow & \infty \in r \circ \inv r \mbox{ and } \sat{\inv r}{F}  \\
& \Leftrightarrow & \exists r',r''.\ r'' \in r \circ r' \mbox{ and } \sat{r'}{F} \mbox{ and } r'' = \infty \\
%& \Leftrightarrow & \mbox{not}(\forall r',r''.\ r'' \in r \circ r' \mbox{ and } \sat{r'}{F} \mbox{ implies } r'' \neq \infty) \\
\mbox{i.e. }\notsat{\inv r}{F} & \Leftrightarrow & \forall r',r''.\ r'' \in r \circ r' \mbox{ and } \sat{r'}{F} \mbox{ implies } r'' \neq \infty
\end{array}\]
By interpreting $\mfalse$ and $\mneg$
as we do in Definition~\ref{defn:CBI_satisfaction}, we
immediately obtain $\sat{r}{\mneg F}$ iff $\sat{r}{F \wand
\mfalse}$, another desired equivalence.

\begin{defn}[Formula validity]
\label{defn:formula_validity} We say that a $\CBI$-formula $F$
is \emph{true} in a $\CBI$-model $M=\cbimodel$ iff $\sat{r}{F}$
for any $M$-environment $\rho$ and $r \in R$. $F$ is said to be
($\CBI$)-\emph{valid} if it is true in all $\CBI$-models.

Truth of $\BBI$-formulas in $\BBI$-models, and $\BBI$-validity
of formulas, is defined similarly.
\end{defn}

\begin{lemma}[$\CBI$ equivalences]\label{lem:cbieq}
\label{lem:CBI_equivalences} The following formulas are all
$\CBI$-valid:
\[\begin{array}{r@{\hspace{0.2cm}}c@{\hspace{0.2cm}}l@{\hspace{1cm}}r@{\hspace{0.2cm}}c@{\hspace{0.2cm}}l}
\mneg\true & \leftrightarrow & \false & F \mor G & \leftrightarrow & \mneg(\mneg F * \mneg G) \\
\mneg\mtrue & \leftrightarrow & \mfalse & (F \wand G) & \leftrightarrow & \mneg F \mor G \\
\mneg\mneg F & \leftrightarrow & F & (F \wand G) & \leftrightarrow & (\mneg G \wand \mneg F) \\
\neg\mneg F & \leftrightarrow & \mneg\neg F & (F \wand G) & \leftrightarrow & \mneg(F * \mneg G) \\
\mneg F & \leftrightarrow & (F \wand \false) & F \mor \mfalse & \leftrightarrow & F
\end{array}\]
\end{lemma}

\begin{proof}
We fix an arbitrary $\CBI$-model $M$ and $M$-environment
$\rho$. For each of the equivalences $F \leftrightarrow G$ we
require to show $\sat{r}{F} \Leftrightarrow \sat{r}{G}$. These
follow directly from the definition of satisfaction, plus the
properties of $\CBI$-models given by
Proposition~\ref{prop:cbimodel_properties}.  We show three
of the cases in detail. \\

\noindent{\em Case $(F \wand G) \leftrightarrow \mneg F \mor
G$:}
\[\begin{array}{rcl}
\sat{r}{\mneg F \mor G} & \Leftrightarrow & \forall r_1,r_2 \in R.\ \inv r \in r_1 \circ r_2
\mbox{ implies } \sat{\inv r_1}{\mneg F} \mbox{ or } \sat{\inv r_2}{G} \\
\mbox{(by Prop~\ref{prop:cbimodel_properties}, pt.~\ref{propitem:doubleinv})}
& \Leftrightarrow & \forall r_1,r_2 \in R.\ \inv r \in r_1
\circ r_2
\mbox{ implies } \notsat{r_1}{F} \mbox{ or } \sat{\inv r_2}{G} \\
& \Leftrightarrow & \forall r_1,r_2 \in R.\ \inv r \in r_1
\circ r_2
\mbox{ and } \sat{r_1}{F} \mbox{ implies } \sat{\inv r_2}{G} \\
\mbox{(by Prop~\ref{prop:cbimodel_properties}, pt.~\ref{propitem:doubleinv})}
& \Leftrightarrow & \forall r_1,r_2 \in R.\ \inv r \in r_1 \circ \inv r_2
\mbox{ and } \sat{r_1}{F} \mbox{ implies } \sat{r_2}{G} \\
\mbox{(by Prop~\ref{prop:cbimodel_properties}, pt.~\ref{propitem:cancel})}
& \Leftrightarrow & \forall r_1,r_2 \in R.\ r_2 \in r \circ r_1
\mbox{ and } \sat{r_1}{F} \mbox{ implies } \sat{r_2}{G} \\
& \Leftrightarrow & \sat{r}{F \wand G}
\end{array}\]

\paragraph{\em Case $(F \wand G) \leftrightarrow (\mneg G \wand \mneg F)$:}
\[\begin{array}{rcl}
\sat{r}{\mneg G \wand \mneg F} & \Leftrightarrow & \forall r',r'' \in R.\ r'' \in r \circ r'
\mbox{ and } \sat{r'}{\mneg G} \mbox{ implies } \sat{r''}{\mneg F} \\
& \Leftrightarrow & \forall r',r'' \in R.\ r'' \in r \circ r'
\mbox{ and } \notsat{\inv r'}{G} \mbox{ implies } \notsat{\inv r''}{F} \\
\mbox{(by Prop~\ref{prop:cbimodel_properties}, pt.~\ref{propitem:doubleinv})}
& \Leftrightarrow & \forall r',r'' \in R.\ \inv r'' \in r \circ \inv r'
\mbox{ and } \notsat{r'}{G} \mbox{ implies } \notsat{r''}{F} \\
& \Leftrightarrow & \forall r',r'' \in R.\ \inv r'' \in r \circ \inv r'
\mbox{ and } \sat{r''}{F} \mbox{ implies } \sat{r'}{G}   \\
\mbox{(by Prop~\ref{prop:cbimodel_properties}, pt.~\ref{propitem:cancel})}
& \Leftrightarrow & \forall r',r'' \in R.\ r' \in r \circ r''
\mbox{ and } \sat{r''}{F} \mbox{ implies } \sat{r'}{G}   \\
& \Leftrightarrow & \sat{r}{F \wand G}
\end{array}\] \\

\noindent{\em Case $F \mor \mfalse \leftrightarrow F$:}
\[\begin{array}{rcl}
\sat{r}{F \mor \mfalse} & \Leftrightarrow & \forall r_1,r_2 \in R.\ \inv r \in r_1 \circ r_2
\mbox{ implies } \sat{\inv r_1}{F} \mbox{ or } \sat{\inv r_2}{\mfalse} \\
 & \Leftrightarrow & \forall r_1,r_2 \in R.\ \inv r \in r_1 \circ r_2
\mbox{ implies } \sat{\inv r_1}{F} \mbox{ or } \inv r_2 \neq \infty \\
\mbox{(by Prop~\ref{prop:cbimodel_properties}, pt.~\ref{propitem:infty})} & \Leftrightarrow
& \forall r_1,r_2 \in R.\ \inv r \in r_1 \circ r_2
\mbox{ implies } \sat{\inv r_1}{F} \mbox{ or } r_2 \neq e \\
& \Leftrightarrow & \forall r_1 \in R.\ \inv r \in r_1 \circ e
\mbox{ implies } \sat{\inv r_1}{F} \\
& \Leftrightarrow & \forall r_1 \in R.\ \inv r = r_1
\mbox{ implies } \sat{\inv r_1}{F} \\
& \Leftrightarrow & \sat{r}{F}
\end{array}\]
\end{proof}

We remark that there is nevertheless at least one important classical
equivalence whose multiplicative analogue does \emph{not} hold
in $\CBI$ in the strong sense of
Lemma~\ref{lem:CBI_equivalences}: the law of excluded middle,
$\mtrue \leftrightarrow F \mor \mneg F$, which (using the
lemma) is equivalent to
 the law of contradiction, $\mfalse
\leftrightarrow F * \mneg F$. This equivalence certainly holds
in one direction, since if $\sat{r}{F * \mneg F}$ then $r \in
r_1 \circ r_2$, $\sat{r_1}{F}$ and $\notsat{\inv r_2}{F}$, so
$r_1 \neq \inv r_2$ and thus $r \neq \infty$ by the
$\CBI$-model axiom, i.e.\ $\sat{r}{\mfalse}$.  The converse
implication does not hold as, given $\sat{r}{\mfalse}$ and some
formula $F$, it clearly is not the case in general that
$\sat{r}{F * \mneg F}$ (e.g., take $F = \false$).  However, the
law does hold in the weak sense that $\mfalse$ is true in a
model $M$ iff $F * \mneg F$ is true in $M$.  One direction of
the implication follows by the argument above, and the other
from the fact that $\mfalse$ is never true in $M$
(because $\notsat{\infty}{\mfalse}$ for any $\rho$).

One might be tempted to think that, since $\CBI$-models are $\BBI$-models and the definition of
satisfaction for $\CBI$ coincides with that of $\BBI$ when
restricted to $\BBI$-formulas, $\CBI$ and $\BBI$ might well be
indistinguishable under such a restriction. Our next result
establishes that this is by no means the case.

\begin{proposition}[Non-conservative extensionality]
\label{prop:nonconservative} $\CBI$ is a non-conservative
extension of $\BBI$.  That is, every $\BBI$-valid formula is
also $\CBI$-valid, but there is a $\BBI$-formula that is
$\CBI$-valid but not $\BBI$-valid.
\end{proposition}

\begin{proof}
To see that $\BBI$-valid formulas are also $\CBI$-valid, let
$M=\cbimodel$ be a $\CBI$-model, whence $M'=\monoid$ is a $\BBI$-model.
For any $\BBI$-valid formula $F$ we have that $F$ is true in $M'$, and thus $F$ is also true in $M$ (because the definition of satisfaction of $F$ coincides in $\CBI$ and $\BBI$ for
$\BBI$-formulas).  Since $M$ was arbitrarily chosen, $F$ is $\CBI$-valid as required.

Now let $P$ be a propositional variable and let $I$ and $J$ be
abbreviations for $\BBI$-formulas defined as follows:
\[\begin{array}{rcl}
I & \defeq & \neg\mtrue \wand \false\\
J & \defeq & \true * (\mtrue \wedge \neg(P \wand \neg I))
\end{array}\]

In a $\BBI$-model $\monoid$, the formula $I$ denotes ``nonextensible'' elements of $R$, i.e.\ those elements
$r \in R$ such that $r \circ r' = \emptyset$ for all $r' \neq
e$:
\[\begin{array}{rcl}
\sat{r}{I} & \Leftrightarrow & \forall r',r'' \in R.\; r'' \in r \circ r' \mbox{ and }
\sat{r'}{\neg\mtrue} \mbox{ implies } \sat{r''}{\false} \\
& \Leftrightarrow & \forall r',r'' \in R.\; r'' \in r \circ r' \mbox{ implies }
\notsat{r'}{\neg\mtrue} \\
& \Leftrightarrow & \forall r',r'' \in R.\; r'' \in r \circ r' \mbox{ implies }
r' = e \\
& \Leftrightarrow & \forall r' \in R.\; r' \neq e \mbox{ implies } r \circ r' = \emptyset
\end{array}\]

The formula $J$ is satisfied by an arbitrary element of $R$ iff
there exists \emph{some} element of $R$ that satisfies the
proposition $P$ and is nonextensible:
\[\begin{array}{rcl}
\sat{r}{J} & \Leftrightarrow & \exists r_1,r_2 \in R.\; r \in r_1 \circ r_2 \mbox{ and }
\sat{r_1}{\true} \mbox{ and } \sat{r_2}{\mtrue \wedge \neg(P \wand \neg I)} \\
& \Leftrightarrow & \exists r_1,r_2 \in R.\; r \in r_1 \circ r_2 \mbox{ and }
\sat{r_2}{\mtrue} \mbox{ and } \notsat{r_2}{P \wand \neg I} \\
& \Leftrightarrow & \notsat{e}{P \wand \neg I} \\
& \Leftrightarrow & \exists r',r'' \in R.\; r'' \in e \circ r' \mbox{ and } \sat{r'}{P}
\mbox{ but } \notsat{r''}{\neg I} \\
& \Leftrightarrow & \exists r' \in R.\; r' \in \rho(P)
\mbox{ and } \sat{r'}{I} \\
\end{array}\]
Note that in any $\CBI$-model $\cbimodel$, for any \mbox{$r \in R$} we
have ${r \circ \inv r} \neq \emptyset$ since $\infty \in {r
\circ \inv r}$ by definition.  Since $\infty$ is the unique
element \mbox{$x \in R$} such that $\inv x = e$ by
Proposition~\ref{prop:cbimodel_properties}, it follows that if
$\sat{r}{I}$ then $r = \infty$. Thus, in $\CBI$-models, if
$\sat{r}{I}$ and $\sat{r}{J}$ then $r = \infty \in \rho(P)$, so
the $\BBI$-formula $I \wedge J \rightarrow P$ is $\CBI$-valid.

To see that $I \wedge J \rightarrow P$ is not $\BBI$-valid,
consider the three-element model $\langle \{e,a,b\} , \circ ,
e\rangle$, where $\circ$ is defined by: $e \circ x = x \circ e
= \{x\}$ for all $x \in \{e,a,b\}$, and $x \circ y = \emptyset$
for all other $x,y \in \{e,a,b\}$.  It is easy to verify that
$\circ$ is both commutative and associative and that $e$ is a
unit for $\circ$, so $\langle \{e,a,b\} , \circ , e\rangle$ is
indeed a $\BBI$-model.  Now define an environment $\rho$ for
this model by $\rho(P) = \{a\}$.  We have both $\sat{a}{I}$ and $\sat{b}{I}$ because $a$ and $b$
are both nonextensible in the model, and $\sat{b}{J}$ because $\sat{a}{I}$ and $a \in \rho(P)$. Then we have $\sat{b}{I
\wedge J}$ but $\notsat{b}{P}$, so $I \wedge J \rightarrow P$
is false in this model and hence not $\BBI$-valid.
\end{proof}

If $\cbimodel$ is a $\CBI$-model and the cardinality of $x
\circ y$ is $\leq 1$ for all $x,y \in R$, then we understand
$\circ$ as a partial function $R \times R \rightharpoonup R$ in
the obvious way. The following proposition shows that, if we
were to restrict our class of $\CBI$-models to those in which
the binary operation is a partial function rather than a
relation, we would obtain a different notion of validity.  In
other words, $\CBI$ is sufficiently expressive to distinguish
between partial functional and relational $\CBI$-models.

\begin{proposition}[Distinction of partial functional and
relational $\CBI$-models] \label{prop:partial_fn_validity}
$\CBI$-validity does not coincide with validity in the class of
partial functional $\CBI$-models.  That is, there is a
$\CBI$-formula that is not generally valid, but is true in
every $\CBI$-model $\cbimodel$ in which $\circ$ is a partial
function.
\end{proposition}

\begin{proof}
Let $K$ and
$L$ be abbreviations for $\CBI$-formulas defined as follows:
\[\begin{array}{rcl}
K & \defeq & \neg(\neg\mfalse \wand \neg\mtrue) \\
L & \defeq & \neg\mfalse \wand \mtrue
\end{array}\]
In a $\CBI$-model $\cbimodel$, the formula $K$ is satisfied by
those model elements that can be extended by $\infty$ to obtain
$e$:
\[\begin{array}{rcl}
\sat{r}{K} & \Leftrightarrow & \exists r',r'' \in R.\; r'' \in r \circ r' \mbox{ and }
\sat{r'}{\neg\mfalse} \mbox{ but } \notsat{r''}{\neg\mtrue} \\
& \Leftrightarrow & \exists r',r'' \in R.\; r'' \in r \circ r' \mbox{ and }
r' = \infty \mbox{ and } r'' = e \\
& \Leftrightarrow & e \in r \circ \infty
\end{array}\]
Similarly, the formula $L$ is satisfied by those elements that,
\emph{whenever} they are extended by $\infty$, always yield
$e$:
\[\begin{array}{rcl}
\sat{r}{L} & \Leftrightarrow & \forall r',r'' \in R.\; r'' \in r \circ r' \mbox{ and }
\sat{r'}{\neg\mfalse} \mbox{ implies } \sat{r''}{\mtrue} \\
& \Leftrightarrow & \forall r',r'' \in R.\; r'' \in r \circ r' \mbox{ and }
r' = \infty \mbox{ implies } r'' = e \\
& \Leftrightarrow & r \circ \infty \subseteq \{e\} \\
\end{array}\]
Let $M=\cbimodel$ be a $\CBI$-model in which $\circ$ is a
partial function, let $\rho$ be an $M$-environment and let $r
\in R$.  Suppose that $\sat{r}{K}$, so that $e \in r \circ
\infty$ by the above.  Since $\circ$ is a partial function, the
cardinality of $r \circ \infty$ is at most 1, so we must have
$r \circ \infty = \{e\}$, i.e., $\sat{r}{L}$. Thus the formula
$K \implies L$ is true in $M$, and so valid with respect to
partial functional $\CBI$-models.

To see that $K \implies L$ is not generally valid, we must
provide a $\CBI$-model in which it is false.  Consider the
three-element model $\langle \{e,a,\infty\}, \circ, e, \inv,
\infty \rangle$, where $\inv$ is defined by $\inv e = \infty,
\inv a = a, \inv\infty = e$ and $\circ$ is defined as follows:
\[\begin{array}{l}
e \circ x = x \circ e = \{x\} \mbox{  for all $x \in \{e,a,\infty\}$} \\
a \circ a = \{e,\infty\} \\
a \circ \infty = \infty \circ a = \infty \circ \infty = \{e,a\}
\end{array}\]
In this model $e$ is a unit for $\circ$ and $\circ$ is
commutative by construction.  It can also easily be verified
that $\circ$ is associative (e.g., $a \circ (a \circ \infty) =
\{e,a,\infty\} = (a \circ a) \circ \infty$) and that $\inv x$
is the unique element such that $\infty \in x \circ \inv x$ for
all $x \in \{e,a,\infty\}$. Thus $\langle \{e,a,\infty\},
\circ, e, \inv, \infty \rangle$ is indeed a $\CBI$-model (and
we note that $\circ$ is not a partial function).  Now for any
environment $\rho$ we have $\sat{a}{K}$ since $e \in a \circ
\infty$, but $\notsat{a}{L}$ since $a \in a \circ \infty$. Thus
$K \implies L$ is false in this model, and hence invalid.
\end{proof}

Our
proof of Proposition~\ref{prop:partial_fn_validity} does not
transfer straightforwardly to $\BBI$ because it crucially
relies upon the fact that, in $\CBI$, we can write down a
formula ($\neg\mfalse$) that is satisfied by exactly one model
element ($\infty$), which is not the unit $e$ in general.
Subsequent to submission of this paper, however, it has been shown by Larchey-Wendling and Galmiche that $\BBI$ is indeed incomplete with respect to partial
functional models~\cite{Larchey-Wendling-Galmiche:10}.

\section{$\displayCBI$: a display calculus proof system for $\CBI$}
\label{sec:display_BI}

In this section, we present $\displayCBI$, a display calculus
for $\CBI$ based on Belnap's general \emph{display logic}~\cite{Belnap:82}, which
provides a generic framework for obtaining formal Gentzen-style consecution calculi for a
large class of logics.  Display calculi
are akin to sequent calculi in that logical connectives are
specified by a pair of introduction rules introducing the
connective on the left and right of proof judgements
respectively.  However, the proof judgements of display calculi
have a richer structure than an ordinary sequent, and thus we require a
corresponding set of meta-level rules (called \emph{display
postulates}) for manipulating this structure.  This ensures the
characteristic, and very useful \emph{display property} of
display calculi: any proof judgement may be rearranged so that
any given part of the judgement appears alone on one side of
the turnstile (without loss of information). In addition to its
conceptual elegance, this property ensures that cut-elimination
holds for any display calculus whose structural rules obey a
few easily verified conditions (cf.~\cite{Belnap:82}).  Our display calculus
$\displayCBI$ indeed satisfies these cut-elimination conditions. Furthermore, it is sound and complete with
respect to our $\CBI$-models.

Belnap's original formulation of display logic treats an
arbitrary number of ``families'' of propositional connectives.
The necessary structural connectives, display postulates and
logical introduction rules are then ascribed automatically to
each family, with only the structural rules governing the
family chosen freely.  For $\CBI$, it is obvious that there are
two complete families of propositional connectives, one
additive and one multiplicative.  Thus the formulation of
$\displayCBI$ can be viewed as arising more or less directly
from Belnap's general schema.

The proof judgements of $\displayCBI$, called consecutions, are
built from structures which generalise the bunches used in
existing proof systems for $\BI$ (cf.~\cite{Pym:02}).

\begin{defn}[Structure / consecution]
\label{defn:consecution} A \emph{$\displayCBI$-structure} $X$
is constructed according to the following grammar:
\[
X ::= F \mid \aemp \mid \ainv X \mid X ; X \mid \memp \mid \minv X \mid X , X
\]
where $F$ ranges over $\CBI$-formulas. If $X$ and $Y$ are
structures then $\seq{X}{Y}$ is said to be a
\emph{consecution}.
\end{defn}

\begin{figure}

\[\begin{array}{cccc}
\mbox{\em Connective} & \mbox{\em Arity} & \mbox{\em Antecedent meaning} & \mbox{\em Consequent meaning} \\
\aemp & 0 & \true & \false \\
\memp & 0 & \mtrue & \mfalse \\
\ainv & 1 & \neg & \neg \\
\minv & 1 & \mneg & \mneg \\
; & 2 & \wedge & \vee \\
, & 2 & * & \mor
\end{array}\]
\caption{The structural connectives of $\displayCBI$.
\label{fig:structural_connectives}}
\end{figure}

Figure~\ref{fig:structural_connectives} gives a summary of the
structural connectives of our display calculus and their
semantic reading as antecedents (or premises) and consequents
(or conclusions) in a consecution.  However, the presence of
the meta-level negations $\ainv$ and $\minv$ in our structures
leads to a subtler notion of antecedent and consequent parts of
consecutions than the simple left-right division of sequent
calculus.  Informally, moving inside a meta-level negation
flips the interpretation of its immediate substructure.  For
example, if $\ainv X$ or $\minv X$ is an antecedent part then
the substructure $X$ should be interpreted as a consequent part,
and vice versa.  This notion is made formal by the following
definition.

\begin{defn}[Antecedent part / consequent part]
\label{defn:parts} A structure occurrence $W$ is said to be a \emph{part}
of another structure $Z$ if $W$ occurs as a substructure of $Z$ (in
the obvious sense).  $W$ is said to be a \emph{positive part}
of $Z$ if $W$ occurs inside an even number of occurrences of
$\ainv$ and $\minv$ in $Z$, and a \emph{negative part} of $Z$
otherwise.

A structure occurrence $W$ is said to be an \emph{antecedent part} of a
consecution $\seq{X}{Y}$ if it is a positive part of $X$ or a
negative part of $Y$.  $W$ is said to be a \emph{consequent
part} of $\seq{X}{Y}$ if it is a negative part of $X$ or a
positive part of $Y$.
\end{defn}

To give the formal interpretation of our consecutions in the following definition,
we employ a pair of mutually recursive functions to capture the
dependency between antecedent and consequent interpretations.

\begin{defn}[Consecution validity]
\label{defn:DLBI_validity} For any structure $X$ we mutually
define two formulas $\dlantform{X}$ and $\dlconform{X}$ by
induction on the structure of $X$ as follows:
\[\begin{array}{rcl@{\hspace{1.0cm}}rcl}
\dlantform{F} & = & F & \dlconform{F} & = & F \\
\dlantform{\aemp} & = & \true & \dlconform{\aemp} & = & \false \\
\dlantform{\ainv X} & = & \neg \dlconform{X} & \dlconform{\ainv X} & = & \neg \dlantform{X} \\
\dlantform{X_1 ; X_2} & = & \dlantform{X_1} \wedge \dlantform{X_2} & \dlconform{X_1 ; X_2} & = & \dlconform{X_1} \vee \dlconform{X_2}\\
\dlantform{\memp} & = & \mtrue & \dlconform{\memp} & = & \mfalse \\
\dlantform{\minv X} & = & \mneg\dlconform{X} & \dlconform{\minv X} & = & \mneg\dlantform{X} \\
\dlantform{X_1 , X_2} & = & \dlantform{X_1} * \dlantform{X_2} & \dlconform{X_1 , X_2} & = & \dlconform{X_1} \mor \dlconform{X_2}
\end{array}\]
A consecution $\seq{X}{Y}$ is then \emph{valid} if
$\dlantform{X} \rightarrow \dlconform{Y}$ is a valid formula
(cf.~Defn.~\ref{defn:formula_validity}).
\end{defn}

We write a proof rule with a double line between premise and
conclusion to indicate that it is \emph{bidirectional}, i.e., that
the roles of premise and conclusion may be reversed. A figure
with three consecutions separated by two double lines is used
to abbreviate two bidirectional rules in the obvious way.

\begin{defn}[Display-equivalence]
\label{defn:displayeq} Two consecutions $\seq{X}{Y}$ and
$\seq{X'}{Y'}$ are said to be \emph{display-equivalent},
written $\seq{X}{Y} \displayeq \seq{X'}{Y'}$, if there is a
derivation of one from the other using only the \emph{display
postulates} given in Figure~\ref{fig:display_postulates}.
\end{defn}

\begin{figure}
{\small\[\begin{array}{c@{\hspace{1cm}}c@{\hspace{1cm}}c}
\begin{prooftree}
\[\seq{X ; Y}{Z}
\Justifies
\seq{X}{\ainv Y ; Z} \using \adl{1}{a} \]
\Justifies
\seq{Y ; X}{Z} \using \adl{1}{b}
\end{prooftree}
&
\begin{prooftree}
\[\seq{X}{Y ; Z}
\Justifies
\seq{X ; \ainv Y}{Z} \using \adl{2}{a} \]
\Justifies
\seq{X}{Z ; Y} \using \adl{2}{b}
\end{prooftree}
&
\begin{prooftree}
\[\seq{X}{Y}
\Justifies
\seq{\ainv Y}{\ainv X} \using \adl{3}{a} \]
\Justifies
\seq{\ainv \ainv X}{Y} \using \adl{3}{b}
\end{prooftree}
\\ & \\
\begin{prooftree}
\[\seq{X , Y}{Z}
\Justifies
\seq{X}{\minv Y , Z} \using \mdl{1}{a} \]
\Justifies
\seq{Y , X}{Z} \using \mdl{1}{b}
\end{prooftree}
&
\begin{prooftree}
\[\seq{X}{Y , Z}
\Justifies
\seq{X , \minv Y}{Z} \using \mdl{2}{a} \]
\Justifies
\seq{X}{Z , Y} \using \mdl{2}{b}
\end{prooftree}
&
\begin{prooftree}
\[\seq{X}{Y}
\Justifies
\seq{\minv Y}{\minv X} \using \mdl{3}{a} \]
\Justifies
\seq{\minv \minv X}{Y} \using \mdl{3}{b}
\end{prooftree} \\
\end{array}\]}
\caption{The display postulates for $\displayCBI$.
\label{fig:display_postulates}}
\end{figure}

The display postulates for $\displayCBI$ are essentially
Belnap's original display postulates, instantiated (twice) to
the additive and multiplicative connective families of $\CBI$.
The only difference is that our postulates build commutativity
of the comma and semicolon into the notion of
display-equivalence, since in $\CBI$ both the conjunctions and
both the disjunctions are commutative.

The fundamental characteristic of
display calculi is their ability to ``display'' structures occurring
in a consecution by rearranging it using the display
postulates.

\begin{theorem}[Display theorem (Belnap~\cite{Belnap:82})]
\label{thm:display} For any antecedent part $W$ of a
consecution $\seq{X}{Y}$ there exists a structure $Z$ such that
$\seq{W}{Z} \displayeq \seq{X}{Y}$.  Similarly, for any
consequent part $W$ of $\seq{X}{Y}$ there exists a structure
$Z$ such that $\seq{Z}{W} \displayeq \seq{X}{Y}$.
\end{theorem}

\begin{proof}
Essentially, one uses the display postulates to move any
structure surrounding $W$ to the opposite side of the
consecution, or to eliminate any preceding occurrences of
$\ainv$ and $\minv$ (note that for each possible position of
$W$ in $\seq{X}{Y}$ there are display postulates allowing the
topmost level of structure above $W$ to be moved away or
eliminated). Moreover, each of the display postulates preserves
antecedent and consequent parts of consecutions, so that $W$
must end up on the correct side of the consecution at the end
of this process.  The details are straightforward.
\end{proof}

\begin{example}
The antecedent part $Y$ of the consecution $\seq{\minv(X ,
\ainv Y)}{Z ; \minv W}$ can be displayed as follows:

\[\begin{prooftree}
\[\[\[\[\[\[\seq{\minv(X , \ainv Y)}{Z ; \minv W}
\justifies
\seq{\minv(Z ; \minv W)}{\minv\minv(X , \ainv Y)} \using \mdl{3}{a} \]
\justifies
\seq{\minv\minv\minv(Z ; \minv W)}{\minv\minv(X , \ainv Y)} \using \mdl{3}{a,b}\]
\justifies
\seq{\minv(X , \ainv Y)}{\minv\minv(Z ; \minv W)} \using \mdl{3}{a} \]
\justifies
\seq{\minv(Z ; \minv W)}{X , \ainv Y} \using \mdl{3}{a} \]
\justifies
\seq{\minv(Z ; \minv W) , \minv X}{\ainv Y} \using \mdl{2}{b} \]
\justifies
\seq{\ainv\ainv Y}{\ainv (\minv(Z ; \minv W) , \minv X)} \using \adl{3}{a}\]
\justifies
\seq{Y}{\ainv (\minv(Z ; \minv W) , \minv X)} \using \adl{3}{a,b}
\end{prooftree}\]
\end{example}

The proof rules of $\displayCBI$ are given in
Figure~\ref{fig:logical_rules}. The identity rules consist of the usual
identity axiom for propositional variables, a cut rule and a
rule for display equivalence.  The logical rules follow the
division between left and right introduction rules familiar
from sequent calculus.  Note that, since we can appeal to Theorem~\ref{thm:display}, the
formula introduced by a logical rule is always displayed in its conclusion.
Both the identity rules and the logical
rules are the standard ones for display logic, instantiated to
%the additive and multiplicative connective families of
$\CBI$.
The structural rules of $\displayCBI$ implement suitable
associativity and unitary laws on both sides of consecutions,
plus weakening and contraction for the (additive) semicolon.

\begin{figure}[t]
%\vspace{2cm}
\flushleft{\bf Identity rules:} \\
{\small\[\begin{array}{c@{\hspace{1cm}}c@{\hspace{1cm}}c}
\begin{prooftree}
\phantom{X}
\justifies
\seq{P}{P}
\using \ax
\end{prooftree}
&
\begin{prooftree}
\seq{X}{F}
\phantom{w}
\seq{F}{Y}
\justifies
\seq{X}{Y} \using \cut
\end{prooftree}
&
\begin{prooftree}
\seq{X'}{Y'}
\justifies
\seq{X}{Y}
\using \;\;\seq{X}{Y} \displayeq \seq{X'}{Y'}\;\;\display
\end{prooftree}
\end{array}\]} \\ \vspace{0.8cm}

\mbox{\bf Logical rules:} \\
{\small
\[\begin{array}{c@{\hspace{0.7cm}}c@{\hspace{0.7cm}}c@{\hspace{0.7cm}}c}
\begin{prooftree}
\seq{\aemp}{X}
\justifies
\seq{\true}{X}
\using \truel
\end{prooftree}
&
\begin{prooftree}
\phantom{w}
\justifies
\seq{\aemp}{\true}
\using \truer
\end{prooftree}
&
\begin{prooftree}
\seq{\memp}{X}
\justifies
\seq{\mtrue}{X}
\using \mtruel
\end{prooftree}
&
\begin{prooftree}
\phantom{w}
\justifies
\seq{\memp}{\mtrue}
\using \mtruer
\end{prooftree}
\\ & \\
\begin{prooftree}
\phantom{w}
\justifies
\seq{\false}{\aemp}
\using \falsel
\end{prooftree}
&
\begin{prooftree}
\seq{X}{\aemp}
\justifies
\seq{X}{\false}
\using \falser
\end{prooftree}
&
\begin{prooftree}
\phantom{w}
\justifies
\seq{\mfalse}{\memp}
\using \mfalsel
\end{prooftree}
&
\begin{prooftree}
\seq{X}{\memp}
\justifies
\seq{X}{\mfalse}
\using \mfalser
\end{prooftree}
\\ & \\
\begin{prooftree}
\seq{\ainv F}{X}
\justifies
\seq{\neg F}{X}
\using \negl
\end{prooftree}
&
\begin{prooftree}
\seq{X}{\ainv F}
\justifies
\seq{X}{\neg F}
\using \negr
\end{prooftree}
&
\begin{prooftree}
\seq{\minv F}{X}
\justifies
\seq{\mneg F}{X}
\using \mnegl
\end{prooftree}
&
\begin{prooftree}
\seq{X}{\minv F}
\justifies
\seq{X}{\mneg F}
\using \mnegr
\end{prooftree}
\\ & \\
\begin{prooftree}
\seq{F ; G}{X}
\justifies
\seq{F \wedge G}{X}
\using \andl
\end{prooftree}
&
\begin{prooftree}
\seq{X}{F}
\phantom{w}
\seq{Y}{G}
\justifies \seq{X ; Y}{F \wedge G}
\using \andr
\end{prooftree}
&
\begin{prooftree}
\seq{F , G}{X}
\justifies
\seq{F * G}{X}
\using \starl
\end{prooftree}
&
\begin{prooftree}
\seq{X}{F}
\phantom{w}
\seq{Y}{G}
\justifies \seq{X , Y}{F * G} \using \starr
\end{prooftree}
\\ & \\
\begin{prooftree}
\seq{F}{X}
\phantom{w}
\seq{G}{Y}
\justifies
\seq{F \vee G}{X ; Y} \using \orl
\end{prooftree}
&
\begin{prooftree}
\seq{X}{F ; G}
\justifies
\seq{X}{F \vee G} \using \orr
\end{prooftree}
&
\begin{prooftree}
\seq{F}{X}
\phantom{w}
\seq{G}{Y}
\justifies
\seq{F \mor G}{X , Y} \using \morl
\end{prooftree}
&
\begin{prooftree}
\seq{X}{F , G}
\justifies
\seq{X}{F \mor G} \using \morr
\end{prooftree}
\\ & \\
\begin{prooftree}
\seq{X}{F} \phantom{w} \seq{G}{Y}
\justifies
\seq{F \rightarrow G}{\ainv X ; Y} \using \impl
\end{prooftree}
&
\begin{prooftree}
\seq{X ; F}{G}
\justifies
\seq{X}{F \rightarrow G}\using \impr
\end{prooftree}
&
\begin{prooftree}
\seq{X}{F}
\phantom{w}
\seq{G}{Y}
\justifies
\seq{F \wand G}{\minv X , Y} \using \wandl
\end{prooftree}
&
\begin{prooftree}
\seq{X , F}{G}
\justifies
\seq{X}{F \wand G}\using \wandr
\end{prooftree}
\end{array}\]} \\ \vspace{0.8cm}

{\bf Structural rules:} \\
{\small
\[\begin{array}{c@{\hspace{0.7cm}}c@{\hspace{0.7cm}}c@{\hspace{0.7cm}}c}
\begin{prooftree}
\seq{W ; (X ; Y)}{Z}
\Justifies
\seq{(W ; X) ; Y}{Z}
\using \aassocl\!\!\!
\end{prooftree}
&
\begin{prooftree}
\seq{W}{(X ; Y) ; Z}
\Justifies
\seq{W}{X ; (Y ; Z)}
\using \aassocr\!\!\!
\end{prooftree}
&
\begin{prooftree}
\seq{W , (X , Y)}{Z}
\Justifies
\seq{(W , X) , Y}{Z}
\using \massocl\!\!\!
\end{prooftree}
&
\begin{prooftree}
\seq{W}{(X , Y) , Z}
\Justifies
\seq{W}{X , (Y , Z)}
\using \massocr\!\!\!
\end{prooftree}
\\ & \\
\begin{prooftree}
\seq{\aemp ; X}{Y}
\Justifies
\seq{X}{Y}
\using \aunitl
\end{prooftree}
&
\begin{prooftree}
\seq{X}{Y ; \aemp}
\Justifies
\seq{X}{Y}
\using \aunitr
\end{prooftree}
&
\begin{prooftree}
\seq{\memp , X}{Y}
\Justifies
\seq{X}{Y}
\using \munitl
\end{prooftree}
&
\begin{prooftree}
\seq{X}{Y , \memp}
\Justifies
\seq{X}{Y}
\using \munitr
\end{prooftree}
\\ & \\
\begin{prooftree}
\seq{X}{Z}
\justifies
\seq{X ; Y}{Z}
\using \weakl
\end{prooftree}
&
\begin{prooftree}
\seq{X}{Z}
\justifies
\seq{X}{Y ; Z}
\using \weakr
\end{prooftree}
&
\begin{prooftree}
\seq{X ; X}{Z}
\justifies
\seq{X}{Z}
\using \contrl
\end{prooftree}
&
\begin{prooftree}
\seq{X}{Z ; Z}
\justifies
\seq{X}{Z}
\using \contrr
\end{prooftree}
\end{array}\]} %\vspace{0.3cm}

\caption{The proof rules of $\displayCBI$.  $W,X,Y,Z$ range
over structures, $F,G$ range over $\CBI$-formulas and $P$
ranges over $\vars$.\label{fig:logical_rules}}
\end{figure}

The identity axiom of $\displayCBI$ is postulated only for
propositional variables\footnote{This is standard in display logic, and slightly simplifies the
proof of cut-elimination.}, but can be
recovered for arbitrary formulas.  We say a consecution is
\emph{cut-free provable} if it has a $\displayCBI$ proof
containing no instances of $\cut$.

\begin{proposition}
\label{prop:DLBI_identity} $\seq{F}{F}$ is cut-free provable in
$\displayCBI$ for any formula $F$.
\end{proposition}

\begin{proof}
By structural induction on $F$.
\end{proof}

\begin{theorem}[Cut-elimination]
\label{thm:DLBI_cut_elim} If a consecution $\seq{X}{Y}$ is
provable in $\displayCBI$ then it is also cut-free provable.
\end{theorem}

\begin{proof}
The $\displayCBI$ proof rules satisfy the conditions
shown by Belnap in~\cite{Belnap:82} to be sufficient for
cut-elimination to hold. We state these conditions and indicate how
they are verified in Appendix~\ref{app:DLBI_cut_elim}.
\end{proof}

The following corollary of Theorem~\ref{thm:DLBI_cut_elim} uses
the notion of a \emph{subformula} of a $\CBI$-formula, defined
in the usual way.

\begin{corollary}[Subformula property]
\label{cor:subformula} If $\seq{X}{Y}$ is
$\displayCBI$-provable then there is a $\displayCBI$ proof of
$\seq{X}{Y}$ in which every formula occurrence is a subformula
of a formula occurring in $\seq{X}{Y}$.
\end{corollary}

\begin{proof} If $\seq{X}{Y}$ is provable then it has a cut-free proof
by Theorem~\ref{thm:DLBI_cut_elim}. By inspection of the
$\displayCBI$ rules, no rule instance in this proof can have in
its premises any formula that is not a subformula of a formula
occurring in its conclusion. Thus a cut-free proof of
$\seq{X}{Y}$ cannot contain any formulas which are not
subformulas of formulas in $\seq{X}{Y}$.
\end{proof}

\begin{corollary}[Consistency]
\label{cor:consistency} Neither $\seq{\memp}{\memp}$ nor
$\seq{\aemp}{\aemp}$ is provable in $\displayCBI$.
\end{corollary}

\begin{proof}
If $\seq{\memp}{\memp}$ were $\displayCBI$-provable then, by
the subformula property (Corollary~\ref{cor:subformula}) there
is a proof of $\seq{\memp}{\memp}$ containing no formula
occurrences anywhere.  But every axiom of $\displayCBI$
contains a formula occurrence, so this is impossible. Then $\seq{\aemp}{\aemp}$ cannot be provable either, otherwise
$\seq{\aemp ; \memp}{\memp ; \aemp}$ is provable by applying
$\weakl$ and $\weakr$, whence $\seq{\memp}{\memp}$ is provable
by applying $\aunitl$ and $\aunitr$, which is a contradiction.
\end{proof}

\begin{figure}
\[\begin{prooftree}
\[\[\[\[\[\[\[\[\[\[\[\[\[\[\[
\mbox{(Proposition~\ref{prop:DLBI_identity})}
\leadsto
\seq{F}{F} \]
\justifies
\seq{\ainv F}{\ainv F} \using \display \]
\justifies
\seq{\ainv F}{\neg F} \using \negr \]
\justifies
\seq{\minv\neg F}{\minv\ainv F} \using \display \]
\justifies
\seq{\mneg\neg F}{\minv\ainv F} \using \mnegl \]
\justifies
\seq{\mneg\neg F ; \mneg F}{\minv\ainv F} \using \weakl \]
\justifies
\seq{\minv F}{\minv\ainv\minv(\mneg\neg F ; \mneg F)} \using \display \]
\justifies
\seq{\mneg F}{\minv\ainv\minv(\mneg\neg F ; \mneg F)} \using \mnegl \]
\justifies
\seq{\mneg\neg F ; \mneg F}{\minv\ainv\minv(\mneg\neg F ; \mneg F)} \using \weakl \]
\justifies
\seq{\ainv\minv(\mneg\neg F ; \mneg F)}{\minv(\mneg\neg F ; \mneg F)} \using \display \]
\justifies
\seq{\minv\aemp ; \ainv\minv(\mneg\neg F ; \mneg F)}{\minv(\mneg\neg F ; \mneg F)} \using \weakl \]
\justifies
\seq{\minv\aemp}{\minv(\mneg\neg F ; \mneg F) ; \minv(\mneg\neg F ; \mneg F)} \using \display \]
\justifies
\seq{\minv\aemp}{\minv(\mneg\neg F ; \mneg F)} \using \contrr \]
\justifies
\seq{\mneg\neg F}{\ainv\mneg F ; \aemp} \using \display \]
\justifies
\seq{\mneg\neg F}{\ainv\mneg F} \using \aunitr \]
\justifies
\seq{\mneg\neg F}{\neg\mneg F} \using \negr
\end{prooftree}\]
\caption{A cut-free $\displayCBI$ proof of $\seq{\mneg\neg F}{\neg\mneg
F}$. \label{fig:displayBI_proof}}
\end{figure}

Our main technical results concerning $\displayCBI$ are the
following.

\begin{proposition}[Soundness]
\label{prop:DLBI_sound} If $\seq{X}{Y}$ is
$\displayCBI$-provable then it is valid.
\end{proposition}

\begin{proof}
It suffices to show that each proof rule of $\displayCBI$ is
\emph{locally sound} in that validity of the conclusion follows
from the validity of the premises.  In the particular case of
the display rule $\display$, local soundness follows by
establishing that each display postulate (see
Figure~\ref{fig:display_postulates}) is locally sound. We show
how to deal with some sample rule cases.

\paragraph{\em Case \wandl.}  Let $M=\cbimodel$ be a
$\CBI$-model, let $r \in R$ and suppose $\sat{r}{F \wand G}$,
whence we require to show $\sat{r}{\mneg \dlantform{X} \mor
\dlconform{Y}}$.  Using Lemma~\ref{lem:CBI_equivalences}, it
suffices to show that $\sat{r}{\dlantform{X} \wand
\dlconform{Y}}$.  So, let $r',r'' \in R$ be such that $r'' \in
r \circ r'$ and $\sat{r'}{\dlantform{X}}$, whence we require to
show $\sat{r''}{\dlconform{Y}}$.  Since the premise
$\seq{X}{F}$ is valid and $\sat{r'}{\dlantform{X}}$ by
assumption, we have $\sat{r'}{F}$.   Then, since $\sat{r}{F
\wand G}$ and $r'' \in r \circ r'$, we have $\sat{r''}{G}$.
Finally, since the premise $\seq{G}{Y}$ is valid by assumption,
we have $\sat{r''}{\dlconform{Y}}$ as required.

\paragraph{\em Case \morl.}  Let $M=\cbimodel$ be a $\CBI$-model,
let $r \in R$ and suppose $\sat{r}{F \mor G}$, whence we
require to show $\sat{r}{\dlconform{X} \mor \dlconform{Y}}$.
So, let $r_1,r_2 \in R$ be such that $\inv r \in r_1 \circ
r_2$, whence we require to show either $\sat{\inv
r_1}{\dlconform{X}}$ or $\sat{\inv r_2}{\dlconform{Y}}$.  Since
$\inv r \in r_1 \circ r_2$ and $\sat{r}{F \mor G}$, we have
either $\sat{r_1}{F}$ or $\sat{r_2}{G}$.  Then, since the
premises $\seq{F}{X}$ and $\seq{G}{Y}$ are assumed valid, we
have the required conclusion in either case.

\paragraph{\em Case $\massocr$.} Both directions of the rule 
follow by establishing that for any $\CBI$-model $M=\cbimodel$ and $r
\in R$ we have $\sat{r}{\dlconform{X} \mor (\dlconform{Y} \mor
  \dlconform{Z})}$ iff $\sat{r}{(\dlconform{X} \mor \dlconform{Y})
  \mor \dlconform{Z}}$.  Using the equivalences $F \mor G
\leftrightarrow \mneg(\mneg F * \mneg G)$ and $\mneg\mneg F
\leftrightarrow F$ given by Lemma~\ref{lem:CBI_equivalences}, it
suffices to show that $\sat{r}{\mneg(\mneg\dlconform{X} *
  (\mneg\dlconform{Y} * \mneg\dlconform{Z}))}$ iff
$\sat{r}{\mneg((\mneg\dlconform{X} * \mneg\dlconform{Y}) *
  \mneg\dlconform{Z})}$.  This follows straightforwardly from the
definition of satisfaction and the associativity of $\circ$.

\paragraph{\em Case $\mdl{1}{a}$.}  We show how to treat one
direction of this display postulate; the reverse direction is
symmetric.  Let $M=\cbimodel$ be a $\CBI$-model, let $r \in R$
and suppose that $\sat{r}{\dlantform{X}}$, whence we require to
show $\sat{r}{\mneg\dlantform{Y} \mor \dlconform{Z}}$.  By
Lemma~\ref{lem:CBI_equivalences}, it suffices to show
$\sat{r}{\dlantform{Y} \wand \dlconform{Z}}$.  So let $r',r''
\in R$ be such that $r'' \in r \circ r'$ and
$\sat{r'}{\dlantform{Y}}$, whence we require to show
$\sat{r''}{\dlconform{Z}}$.  Since $\sat{r}{\dlantform{X}}$ we
have $\sat{r''}{\dlantform{X} * \dlantform{Y}}$, whence we have
$\sat{r''}{\dlconform{Z}}$ as required because the premise
$\seq{X , Y}{Z}$ is assumed valid.
\end{proof}

\begin{theorem}[Completeness of $\displayCBI$]
\label{thm:DLBI_complete} If $\seq{X}{Y}$ is valid then it is
provable in $\displayCBI$.
\end{theorem}

We give the proof of Theorem~\ref{thm:DLBI_complete} in
Section~\ref{sec:completeness}.

We remark that, although cut-free proofs in $\displayCBI$ enjoy
the subformula property, they do not enjoy the analogous
``substructure property'', and cut-free proof search in our
system is still highly non-deterministic due to the presence of
the display postulates and structural rules, the usage of which
cannot be straightforwardly constrained in general. In
Figure~\ref{fig:displayBI_proof} we give a sample cut-free
proof of the consecution $\seq{\mneg\neg F}{\neg\mneg F}$,
which illustrates the problems.  The applications of
display-equivalence are required in order to apply the logical
rules, as one would expect, but our derivation also makes %seemingly
essential use of contraction, weakening and a unitary law. It
is plausible that the explicit use of at least some of these
structural rules can be eliminated by suitable reformulations
of the logical rules.  However, the inherent nondeterminism in
proof search cannot be removed by refining $\displayCBI$
without loss of power since, by soundness and completeness,
provability in $\displayCBI$ is equivalent to validity in
$\CBI$, which has been recently shown \emph{undecidable} by the
first author and
Kanovich~\cite{Brotherston-Kanovich:10}.  This is not
fundamentally surprising, since at least some displayable logics are
known to be undecidable; indeed, one of Belnap's original applications of
display logic was in giving a display calculus for the full
relevant logic $\mathbf{R}$, which was famously proven
undecidable by Urquhart~\cite{Urquhart:84}.  (Unfortunately, we
cannot distinguish decidable display calculi from undecidable
ones in general; the decidability of an arbitrary displayable
logic was \emph{itself} shown undecidable by
Kracht~\cite{Kracht:96}.)

Nonetheless, we argue that there are good reasons to prefer our
$\displayCBI$ over arbitrary complete proof systems (e.g.\
Hilbert systems) without cut-elimination. Display calculi
inherit the main virtues of traditional Gentzen systems: they
distinguish structural principles from logical ones, and make
explicit the considerable proof burden that exists at the
meta-level, but nevertheless retain a theoretically very
elegant and symmetric presentation. Furthermore, as a result of
the subformula property one has in display calculi what might
be called a property of ``finite choice'' for proof search: for
any consecution there are only finitely many ways of applying
any rule to it in a backwards fashion\footnote{In fact, this is
not quite true as it stands because for any consecution there
are infinitely many consecutions that are display-equivalent to
it, obtained by ``stacking'' occurrences of $\ainv$ and
$\minv$. However, by identifying structures such as $\ainv\ainv
X$ and $X$, one obtains only finitely many display-equivalent
consecutions. See e.g.~\cite{Restall:99}.}.

\section{Completeness of $\displayCBI$}
\label{sec:completeness}

In this section we prove completeness of our display calculus
$\displayCBI$ with respect to validity in $\CBI$-models.  As in
the case of the analogous result for $\BBI$
in~\cite{Calcagno-Gardner-Zarfaty:07}, our result hinges on a
general completeness theorem for modal logic due to Sahlqvist.
However, we also require an extra layer of translation between
Hilbert-style proofs and proofs in $\displayCBI$.

Our proof is divided into three main parts.  First, in
subsection~\ref{subsec:completeness Sahlqvist}, we reinvent
$\CBI$ as a modal logic by defining a class of standard modal
frames, with associated modalities corresponding to the
standard $\CBI$-model operations, that satisfy a certain
set of modal logic axioms. By appealing to Sahlqvist's
completeness theorem, we obtain a complete Hilbert-style proof
theory for this class of frames. It then remains to connect the
modal presentation of $\CBI$ to our standard presentation. In
subsection~\ref{subsec:completeness_models}, we show that the
aforementioned class of modal frames is exactly the class of
$\CBI$-models given by Definition~\ref{defn:CBI_model}. Then,
in subsection~\ref{subsec:completeness_proofs}, we show how to
translate any modal logic proof into a $\displayCBI$ proof.
Thus we obtain the $\displayCBI$-provability of any valid
consecution.

\subsection{$\CBI$ as a modal logic}
\label{subsec:completeness Sahlqvist}
\renewcommand{\thetheorem}{\thesubsection.\arabic{theorem}}
\setcounter{theorem}{0}

In this subsection we define the semantics of a modal logic
corresponding to $\CBI$, and obtain a complete proof theory
with respect to this semantics, all using standard modal
techniques (see e.g.~\cite{Blackburn-deRijke-Venema:01}).

We first define $\ML_\CBI$ \emph{frames}, which are standard
modal frames with associated modalities corresponding to the
$\CBI$-model operations in Definition~\ref{defn:CBI_model}.

\begin{defn}[Modal logic frames]\label{def:mbiprem}
An \emph{$\ML_\CBI$ frame} is a tuple $\langle
R,\circ,\myerightarrow,e,-,\infty \rangle$, where \mbox{$\circ
: R \times R \rightarrow \pow{R}$}, \mbox{$\myerightarrow : \pow{R}
\times \pow{R} \rightarrow \pow{R}$},
 $e\subseteq R$, \mbox{$\inv : R \rightarrow \pow{R}$},
and $\infty \subseteq R$.  We extend $\circ$ %and $\myerightarrow$
to $\pow{R} \times \pow{R} \rightarrow
\pow{R}$, and $\inv$ to $\pow{R} \rightarrow \pow{R}$, in the
same pointwise manner as in Definition~\ref{defn:CBI_model}.
If $e$ is a singleton set then the frame is said to be {\em
unitary}.
\end{defn}

\begin{defn}[Modal logic formulas]
Modal logic formulas $A$ are defined by: %the grammar:
\[
A ::= P \mid \true \mid \false \mid \neg A \mid A \wedge A \mid A \vee A \mid
A \rightarrow A \mid e \mid \infty \mid \inv A \mid A \circ A \mid A \myerightarrow A
\]
where $P$ ranges over $\vars$.   We remark that we read $e,
\infty,\inv, \circ, \myerightarrow$ as \emph{modalities} (with
the obvious arities). We regard $\rightarrow$ as having weaker
precedence than these modalities, and use parentheses to
disambiguate where necessary.
\end{defn}

The satisfaction relation for modal logic formulas in
$\ML_\CBI$ frames is defined exactly as in
Definition~\ref{defn:CBI_satisfaction} for the additive
connectives, and the modalities are given a ``diamond''
possibility interpretation:
\[\begin{array}{rcl}
\sat{r}{e} & \Leftrightarrow & r\in e\\
\sat{r}{\infty} & \Leftrightarrow & r\in \infty\\
\sat{r}{\inv A} & \Leftrightarrow & \exists r'\in R.\ r \in \inv(r') \mbox{ and } \sat{r'}{A}\\
\sat{r}{A_1 \circ A_2} & \Leftrightarrow & \exists r_1,r_2 \in R.\ r \in r_1 \circ r_2 \mbox{ and } \sat{r_1}{A_1} \mbox{ and } \sat{r_2}{A_2} \\
\sat{r}{A_1 \myerightarrow A_2} & \Leftrightarrow & \exists r_1,r_2 \in R.\ r \in r_1 \myerightarrow r_2 \mbox{ and } \sat{r_1}{A_1} \mbox{ and } \sat{r_2}{A_2} \\
\end{array}\]
We remark that the $\myerightarrow$ modality --- which does not
correspond directly to a $\CBI$-model operation but should be
read informally as $\neg(A_1 \wand \neg A_2)$ --- will be
helpful later in giving a modal axiomatisation of
$\CBI$-models; see Defn.~\ref{defn:ML_axioms}.  We could
alternatively employ a modality corresponding directly to
$\wand$, but it is much more technically convenient to work
exclusively with ``diamond'' modalities.

Given any set $\mathcal{A}$ of modal logic axioms, we define
\emph{$\mathcal{A}$-models} to be those $\ML_\CBI$ frames in
which every axiom in $\mathcal{A}$ holds. The standard modal
logic proof theory corresponding to the class of
$\mathcal{A}$-models is given by
the following definition
(cf.~\cite{Blackburn-deRijke-Venema:01}).

\begin{defn}[Modal logic proof theory]
\label{defn:ML_proof_theory} The modal logic proof theory
generated by a set $\mathcal{A}$ of modal logic axioms, denoted
by $\mathrm{L}\mathcal{A}$, consists of some fixed finite
axiomatisation of % standard Hilbert proof system for
propositional classical logic, extended with the following
axioms and proof rules:

\[\begin{array}{ll}
(\mathcal{A}): & {A}  \hspace{2cm} \mbox{ for each $A \in \mathcal{A}$} \\
\invfalse: & {\inv\false \implies \false} \\
\circfalse: & {P \circ \false \implies \false} \\
\lollyfalse: & {(\false \myerightarrow P) \vee (P \myerightarrow \false) \implies \false} \\
\invdisj: & {\inv(P \vee Q) \leftrightarrow \inv P \vee \inv Q} \\
\circdisj: & {(P \vee Q) \circ R \leftrightarrow (P \circ R) \vee (Q \circ R)} \\
\lollydisjl: & {(P \vee Q) \myerightarrow R \leftrightarrow
(P \myerightarrow R) \vee (Q \myerightarrow R)} \\
\lollydisjr: & {P \myerightarrow (Q \vee R) \leftrightarrow
(P \myerightarrow Q) \vee (P \myerightarrow R)} \\ \\
\end{array}\]

\[\begin{array}{ccc}
\begin{prooftree}
{A \rightarrow B} \quad {A}
\justifies
{B} \using \modusponens
\end{prooftree}
&
\begin{prooftree}
{A}
\justifies A[B/P] \using \subst
\end{prooftree}
&
\begin{prooftree}
{A \implies B}
\justifies {(\inv A) \implies (\inv B)}
\using \mbox{ ($\diamond\inv$)}
\end{prooftree}
\\ & \\
\begin{prooftree}
{A \implies B}
\justifies {(A \circ C) \implies (B \circ C)}
\using \mbox{ ($\diamond\circ$)}
\end{prooftree}
&
\begin{prooftree}
{A \implies B}
\justifies {(C \myerightarrow A) \implies (C \myerightarrow B)}
\using \mbox{ ($\diamond\myerightarrow$L)}
\end{prooftree}
&
\begin{prooftree}
{A \implies B}
\justifies {(A \myerightarrow C) \implies (B \myerightarrow C)}
\using \mbox{ ($\diamond\myerightarrow$L)}
\end{prooftree}
\\ \\
\end{array}\]
where $A,B,C$ range over modal logic formulas, $P,Q,R$ are
propositional variables, and $A \leftrightarrow B$ is as usual
an abbreviation for $(A \rightarrow B) \wedge (B \rightarrow
A)$.
\end{defn}

Note that the axioms and rules for the modalities which are
added to $\mathcal{A}$ by Definition~\ref{defn:ML_proof_theory}
are just the axioms and rules of the standard modal logic
$\mathrm{K}$, instantiated to each of our ``diamond''-type
modalities $e$, $\infty$, $\inv$, $\circ$ and $\myerightarrow$. We emphasise that, by definition, the latter are diamond modalities rather than logical connectives.  In particular, the modality `$\inv$' is not a negation ($\inv A$ should be understood informally as the $\CBI$-formula $\mneg\neg A$), and is monotonic rather than antitonic with respect to entailment, as embodied by the rule \mbox{($\diamond\inv$)}.  Similarly, the $\myerightarrow$ modality is monotonic in its left-hand argument because it is a diamond modality and not an implication.

We now state a sufficient condition, due to Sahlqvist, for
completeness of $\mathrm{L}\mathcal{A}$ to hold with respect to
the class of $\mathcal{A}$-models.

\begin{defn}[Very simple Sahlqvist formulas]
\label{defn:Sahlqvist_formulas} A \emph{very simple Sahlqvist
antecedent} $S$ is a formula given by the grammar:
\[
S ::= \true \mid \false \mid P \mid S \wedge S \mid e \mid \infty \mid \inv S
\mid S \circ S \mid S \myerightarrow S
\]
where $P$ ranges over $\vars$. A \emph{very simple Sahlqvist
formula} is a modal logic formula of the form $S \implies A^+$,
where $S$ is a very simple Sahlqvist antecedent and $A^+$ is a
modal logic formula which is \emph{positive} in that no
propositional variable $P$ in $A^+$ may occur inside the scope
of an odd number of occurrences of $\neg$.
\end{defn}

\begin{theorem}[Sahlqvist~\cite{Blackburn-deRijke-Venema:01}]
\label{thm:Sahlqvist} Let $\mathcal{A}$ be a set of modal
logic axioms consisting only of very simple Sahlqvist formulas. Then
the modal logic proof theory $\mathrm{L}\mathcal{A}$ is
complete with respect to the class of $\mathcal{A}$-models.
That is, if a modal logic formula $F$ is valid with respect to
$\mathcal{A}$-models then it is provable in
$\mathrm{L}\mathcal{A}$.
\end{theorem}

\begin{defn}[Modal logic axioms for $\CBI$]
\label{defn:ML_axioms} The axiom set $\AX_{\CBI}$ consists of
the following modal logic formulas, where $P,Q,R$ are
propositional variables:
\begin{center}
\parbox[t]{7cm}{\begin{enumerate}
\item $e \circ P \rightarrow P$
\item $P \rightarrow e \circ P$
\item $P\circ Q \rightarrow Q \circ P$
\item $(P \circ Q) \circ R  \rightarrow P \circ  (Q \circ
    R)$
\item $P \circ (Q \circ R) \rightarrow (P \circ Q) \circ R$
\item\label{ax:wand1} $Q \wedge (R \circ P) \rightarrow (R
    \wedge (P\myerightarrow Q)) \circ \true$
\end{enumerate}}
\parbox[t]{7cm}{\begin{enumerate} \setcounter{enumi}{6}
\item\label{ax:wand2} $R \wedge (P \myerightarrow Q)
    \rightarrow (\true \myerightarrow (Q \wedge (R \circ P)))$
\item $\inv\inv P \rightarrow P$
\item $P \rightarrow \inv\inv P$
\item $\inv P \rightarrow (P \myerightarrow \infty)$
\item $(P \myerightarrow \infty) \rightarrow \inv P$
\end{enumerate}}
\end{center}
\end{defn}

By inspection we can observe that the $\AX_{\CBI}$ axioms (cf.\
Definition~\ref{defn:ML_axioms}) are all very simple Sahlqvist
formulas, whence we obtain from Theorem~\ref{thm:Sahlqvist}:

\begin{corollary}
\label{cor:Sahlqvist} If a modal logic formula $F$ is valid
with respect to $\AX_{\CBI}$-models then it is
provable in $\LAX_{\CBI}$.
\end{corollary}

We show that the completeness result transfers to unitary
$\AX_{\CBI}$-models.

\begin{lemma}
Let $M=\mbimodel$ be an $\AX_{\CBI}$ model. Then there exist
unitary $\AX_{\CBI}$-models $M_{x}$ for each $x\in e$ such that
the following hold:
\begin{enumerate}
\item $M$ is the disjoint union of the models $M_{x}$ for
    $x\in e$.
\item A formula $A$ is true in $M$ iff it is true in
    $M_x$ for all $x\in e$.
\end{enumerate}
\end{lemma}
\begin{proof}
For each $x\in e$, the model $M_x$ is defined by restricting
$M$ to $R_x \defeq \{r\in R \mid \{r\} \circ \{x\} \neq
\emptyset\}$. Disjointness of models follows directly from the
fact that $\langle R,\circ, e\rangle$ obeys the first five
axioms of $\AX_{\CBI}$, which characterize relational
commutative monoids. Finally, $(1) \Rightarrow (2)$ is a
general result which holds in modal
logic~\cite{Blackburn-deRijke-Venema:01}.
\end{proof}

\begin{corollary}
\label{cor:completeness} If a modal logic formula $F$ is valid
with respect to \emph{unitary} $\AX_{\CBI}$-models then it is
provable in $\LAX_{\CBI}$.
\end{corollary}

\subsection{$\CBI$-models as modal logic models}
\label{subsec:completeness_models} \setcounter{theorem}{0}

\begin{lemma}\label{lem:cbitombi}
If $\cbimodel$ is a $\CBI$-model then, for all $X,Y,Z\in
\pow{R}$, we have:
\begin{enumerate}
\item $X\circ Y = Y\circ X$ and $X\circ (Y\circ Z) =
    (X\circ Y)\circ Z$ and $\{e\}\circ X = X$
\item $\inv X = X \myerightarrow \infty$
\item $\inv\inv X=X$
\end{enumerate}
where $X\myerightarrow Y \defeq \{z\in R \mid \exists x\in
X,y\in Y.\ y\in x \circ z\}$.
\end{lemma}

\begin{proof}
The required properties follow straightforwardly from the
properties of $\CBI$-models given by
Definition~\ref{defn:CBI_model} and
Proposition~\ref{prop:cbimodel_properties}.
\end{proof}

\begin{lemma}\label{lem:mbitocbi}
Let $\mbimodel$ be an unitary $\AX_{\CBI}$-model (so that $e$
is a singleton set). Then $\infty$ is a singleton set, and
$\inv x$ is a singleton set for any $x \in R$. Moreover,
$\cbimodel$ is a $\CBI$-model with the modalities
$e,\inv,\infty$ regarded as having the appropriate types.
\end{lemma}
\begin{proof}
We first show that $-x$ is a singleton by contradiction, using the fact
that $\inv\inv x=\{x\}$ must hold for any set $x$, as a consequence of axioms (8) and (9).
 If $\inv x=\emptyset$ then
    $\inv\inv x= \bigcup_{y \in \inv x} \inv y =
    \emptyset$, which contradicts $\inv\inv x=\{x\}$. If
    $x_1,x_2 \in \inv x$ with $x_1 \neq x_2$, then $\inv
    x_1 \cup \inv x_2 \subseteq \inv\,\inv x$. Also, $\inv
    x_1 \not= \inv x_2$, otherwise we would have
    $\{x_1\}=\inv\inv x_1=\inv\inv x_2=\{x_2\}$ and
    thus $x_1 = x_2$. Since $\inv x_1$ and $\inv x_2$ have
    cardinality $>0$ (see above), $\inv\inv x$ must have
    cardinality $>1$, which contradicts $\inv\inv x=\{x\}$.

We prove that $\infty$ is a singleton by deriving
$\infty = \inv e$. Using the axioms in Definition~\ref{defn:ML_axioms},
we will show that $e \myerightarrow X = X$ must hold for any set $X$.
This fact, together with axioms (10) and (11) instantiated with $P=e$
gives the desired consequence $\infty = \inv e$.

It remains to show $e \myerightarrow X = X$.
Axioms (6) and (7) give the two directions of:
\[q \in r \circ p \textrm{ iff } r \in p \myerightarrow q\]
for any $p,q,r \in R$, and axioms (1), (2) and (3) give, for any $x\in R$:
\[ x \circ e = \{x\}\]
Therefore we have that, for any $x\in R$:
\[ x \in e \myerightarrow X \textrm{ iff } (\exists x'\in X.\,x \in e \myerightarrow x')
\textrm{ iff } (\exists x'\in X.\, x' \in x \circ e) \textrm{ iff } x \circ e \subseteq X
\textrm{ iff } x \in X. \]
\end{proof}

\begin{defn}[Embedding of $\CBI$-models in $\AX_{\CBI}$-models]
Let $M = \cbimodel$ be a $\CBI$-model. The tuple $\embed{M} =
\mbimodel$ is obtained by regarding %the modalities
$e,\inv,\infty$ as having the same types as in
Definition~\ref{def:mbiprem} in the obvious way, and by
defining the modality $\myerightarrow : \pow{R} \times \pow{R} \rightarrow
\pow{R}$ by $X\myerightarrow Y
\defeq \{z\in R \mid \exists x\in X,y\in Y.\, y\in x \circ
z\}$.
\end{defn}

\begin{lemma}\label{lem:emb_models}
If $M$ is a $\CBI$-model then $\embed{M}$ is a unitary
$\AX_{\CBI}$-model. Moreover, the function $\embed{-}$ is a
bijection between $\CBI$-models and unitary
$\AX_{\CBI}$-models.
\end{lemma}

\begin{proof}
First observe that in any $\ML_\CBI$ frame $\mbimodel$, the
$\AX_{\CBI}$ axioms~(\ref{ax:wand1}) and~(\ref{ax:wand2}) hold iff we have, for all $X,Y$ in $\pow{R}$:
\[ X\myerightarrow Y = \{z\in R \mid \exists x\in X,y\in Y.\, y\in x \circ z\}\]
Let $M$ be a $\CBI$-model. Then axioms~(\ref{ax:wand1})
and~(\ref{ax:wand2}) hold in $\embed{M}$ by the above
observation. The remaining $\AX_{\CBI}$ axioms hold in
$\embed{M}$ as a direct consequence of
Lemma~\ref{lem:cbitombi}. Therefore $\embed{M}$ is a unitary
$\AX_{\CBI}$-model.

It remains to show that $\embed{-}$ is a bijection. Injectivity
is immediate by definition. For surjectivity, let
$M'=\mbimodel$ be a unitary $\AX_{\CBI}$ model. By
Lemma~\ref{lem:mbitocbi} we have that $\cbimodel$ is a
$\CBI$-model. Since the interpretation of $\myerightarrow$ is
determined by $\circ$ because of the above observation about
axioms~(\ref{ax:wand1}) and~(\ref{ax:wand2}), it follows that
$\embed{\cbimodel}=M'$, hence $\embed{-}$ is surjective.
\end{proof}

\begin{defn}[Translation of $\CBI$-formulas to modal logic
formulas] \label{defn:embed} We define a function $\embed{-}$
from $\CBI$-formulas to modal logic formulas by induction on
the structure of $\CBI$-formulas, as follows:
\[\begin{array}{r@{\hspace{0.3cm}}c@{\hspace{0.3cm}}l@{\hspace{0.3cm}}l}
\embed{F} & = & F & \mbox{where }F \in \{P,\true,\false\} \\
\embed{\mtrue} & = & e \\
\embed{F_1\,?\,F_2} & = & \embed{F_1}\,?\,\embed{F_2} & \mbox{where }? \in \{\wedge,\vee,\rightarrow\} \\
\embed{F_1\,*\,F_2} & = & \embed{F_1}\,\circ\,\embed{F_2}\\
\embed{F_1\,\wand\,F_2} & = & \neg(\embed{F_1}\,\myerightarrow\,\neg \embed{F_2})\\
\embed{\neg F} & = & \neg\embed{F} \\
\embed{\mfalse} & = & \neg\infty \\
\embed{\mneg F} & = & \neg\inv\embed{F} \\
\embed{F_1 \mor F_2} & = & % \multicolumn{2}{l}{\!\!\!\embed{\mneg(\mneg F_1 * \mneg F_2)}} \\
%& = &
\multicolumn{2}{l}{\!\!\!\neg\inv(\neg\inv\embed{F_1} \circ \neg\inv\embed{F_2})}
\end{array}\]
where $P$ in the first clause ranges over $\vars$.  We extend
the domain of $\embed{-}$ to $\displayCBI$ consecutions by:
\[
\embed{\seq{X}{Y}} = \embed{\dlantform{X}} \rightarrow \embed{\dlconform{Y}}
\]
where $\dlantform{-}$ and $\dlconform{-}$ are the functions
given in Definition~\ref{defn:DLBI_validity}.
\end{defn}

In the following, we write $F[G/P]$ to denote the result of
substituting the formula $G$ for all occurrences of the
propositional variable $P$ in the formula $F$.  This notation
applies both to $\CBI$-formulas and to modal logic formulas.

\begin{lemma}\label{lem:emb_form}
Let $F$ be a $\CBI$-formula, and $M=\cbimodel$ a $\CBI$-model.
Then $F$ is true in $M$ if and only if $\embed{F}$ is true in
$\embed{M}$.
\end{lemma}
\begin{proof}
Let $F$ be a $\CBI$-formula and $A$ a modal logic formula. We
define $F \simeq A$ to hold iff for all environments $\rho$,
and all $r\in R$, the following holds:
\[ r \models_\rho F \mbox{ wrt.\ $M$} \;\Leftrightarrow\; r \models_\rho A \mbox{ wrt.\ $\embed{M}$}\]
The proof is divided into two parts. The first part establishes
the following properties:
\begin{enumerate}
\item $F\simeq A$ and $G \simeq B$ implies  $F[G/P] \simeq
    A[B/P]$
\item $\mtrue \simeq e$
\item $P_1 * P_2 \simeq P_1 \circ P_2$
\item $P_1 \wand P_2 \simeq \neg(P_1\,\myerightarrow\,\neg
    P_2)$
\item $\mfalse \simeq \neg\infty$
\item $\mneg P \simeq \neg\inv P$
\item $P_1 \mor P_2 \simeq \neg\inv(\neg\inv P_1 \circ
    \neg\inv P_2)$
\end{enumerate}
We show one interesting case (7). By Lemma~\ref{lem:cbieq} we
have that $P_1 \mor P_2$ is equivalent to $\mneg(\mneg P_1 *
\mneg P_2)$, therefore it is sufficient to prove $\mneg(\mneg
P_1 * \mneg P_2) \simeq \neg\inv(\neg\inv P_1 \circ \neg\inv
P_2)$. By (6) we have $\mneg P_i \simeq \neg\inv P_i$ for $i
\in \{1,2\}$, hence by (1) and (3) we obtain $(\mneg P_1 *
\mneg P_2) \simeq (\neg\inv P_1 \circ \neg\inv P_2)$. Thus by
(1) and (6) we conclude $\mneg(\mneg P_1 * \mneg P_2) \simeq
\neg\inv(\neg\inv P_1 \circ \neg\inv P_2)$, as required.

The second part establishes $F \simeq \embed{F}$ by induction
on the structure of $F$, using the results from the first part.
\end{proof}

\begin{proposition}
\label{prop:embed_valid2} A consecution $\seq{X}{Y}$ is valid
(wrt.\ $\CBI$-models) iff $\embed{\dlantform{X} \rightarrow
\dlconform{Y}}$ is valid wrt.\ unitary $\AX_{\CBI}$-models.
\end{proposition}
\begin{proof}
By definition, $\seq{X}{Y}$ is valid iff $\dlantform{X}
\rightarrow \dlconform{Y}$ is true in every $\CBI$-model $M$.
By Lemma~\ref{lem:emb_form}, this is equivalent to:
\begin{equation*}
\embed{\dlantform{X} \rightarrow \dlconform{Y}} \mbox{ is true in $\embed{M}$ for every $\CBI$-model $M$}
\end{equation*}
Since $\embed{-}$ is a bijection onto unitary
$\AX_{\CBI}$-models by Lemma~\ref{lem:emb_models}, this is
equivalent to:
\begin{equation*}
\embed{\dlantform{X} \rightarrow \dlconform{Y}} \mbox{ is true in all unitary $\AX_{\CBI}$-models}
\end{equation*}
i.e. $\embed{\dlantform{X} \rightarrow \dlconform{Y}}$ is valid
wrt.\ unitary $\AX_{\CBI}$-models.
\end{proof}

By combining Proposition~\ref{prop:embed_valid2} and
Corollary~\ref{cor:completeness} we obtain the following key
intermediate result towards completeness for $\displayCBI$:

\begin{corollary}
\label{cor:completeness2} If $\seq{X}{Y}$ is a valid
consecution then $\embed{\dlantform{X} \rightarrow
\dlconform{Y}}$ is provable in $\LAX_{\CBI}$.
\end{corollary}

\subsection{From modal logic proofs to $\displayCBI$ proofs}
\label{subsec:completeness_proofs}
\setcounter{theorem}{0}

\begin{defn}[Translation from modal logic formulas to $\CBI$-formulas]
\label{defn:revembed} We define a function $\revembed{-}$ from
modal logic formulas to $\CBI$-formulas by induction on the
structure of $\CBI$-formulas, as follows:
\[\begin{array}{r@{\hspace{0.3cm}}c@{\hspace{0.3cm}}l@{\hspace{0.3cm}}l}
\revembed{A} & = & A & \mbox{where }A \in \{P,\true,\false\} \\
\revembed{\neg A} & = & \neg\revembed{A} \\
\revembed{A_1 \mathrel{?} A_2} & = & \revembed{A_1} \mathrel{?} \revembed{A_2} & \mbox{where }? \in \{\wedge,\vee,\rightarrow\} \\
\revembed{A_1 \circ A_2} & = & \revembed{A_1}\,*\,\revembed{A_2}\\
\revembed{A_1\,\myerightarrow\,A_2} & = & \neg(\revembed{A_1} \wand \neg\revembed{A_2})\\
\revembed{\;e\;} & = & \mtrue \\
\revembed{\inv A} & = & \neg\mneg\revembed{A} \\
\revembed{\;\infty\;} & = & \neg\mfalse \\
\end{array}\]
\end{defn}

\begin{figure}
{\small
\[\begin{prooftree}
\[\[\[\[\[\[\[\[\[
\[
\[\mbox{(Prop.~\ref{prop:DLBI_identity})} \leadsto \seq{R}{R} \]
\[\[\[
\[\mbox{(Prop.~\ref{prop:DLBI_identity})} \leadsto \seq{P}{P} \]
\[\[\[
\[\mbox{(Prop.~\ref{prop:DLBI_identity})} \leadsto \seq{Q}{Q} \]
\justifies
\seq{\ainv Q}{\ainv Q} \using \display \]
\justifies
\seq{\neg Q}{\ainv Q} \using \negl \]
\justifies
\seq{\neg Q}{\ainv Q ; (R \wedge \neg(P \wand \neg Q)) * \true} \using \weakr \]
\justifies
\seq{P \wand \neg Q}{\minv P , (\ainv Q ; (R \wedge \neg(P \wand \neg Q)) * \true)} \using \wandl \]
\justifies
\seq{\ainv(\minv P , (\ainv Q ; (R \wedge \neg(P \wand \neg Q)) * \true))}{\ainv P \wand \neg Q} \using \display \]
\justifies
\seq{\ainv(\minv P , (\ainv Q ; (R \wedge \neg(P \wand \neg Q)) * \true))}{\neg(P \wand \neg Q)} \using \negr \]
\justifies
\seq{R ; \ainv(\minv P , (\ainv Q ; (R \wedge \neg(P \wand \neg Q)) * \true))}{R \wedge \neg(P \wand \neg Q)} \using \andr \]
\[\[\[
\justifies
\seq{\aemp}{\true} \using \truer \]
\justifies
\seq{\aemp ; P}{\true} \using \weakr \]
\justifies
\seq{P}{\true} \using \aunitr \]
\justifies
\seq{(R ; \ainv(\minv P , (\ainv Q ; (R \wedge \neg(P \wand \neg Q)) * \true))) , P}{(R \wedge \neg(P \wand \neg Q)) * \true} \using \starr \]
\justifies
\seq{(R ; \ainv(\minv P , (\ainv Q ; (R \wedge \neg(P \wand \neg Q)) * \true))) , P}{\ainv Q ; (R \wedge \neg(P \wand \neg Q)) * \true} \using \weakr \]
\justifies
\seq{R}{(\minv P , (\ainv Q ; (R \wedge \neg(P \wand \neg Q)) * \true)) ; (\minv P , (\ainv Q ; (R \wedge \neg(P \wand \neg Q)) * \true))} \using \display \]
\justifies
\seq{R}{\minv P , (\ainv Q ; (R \wedge \neg(P \wand \neg Q)) * \true)} \using \contrr \]
\justifies
\seq{R , P}{\ainv Q ; (R \wedge \neg(P \wand \neg Q)) * \true} \using \display \]
\justifies
\seq{R * P}{\ainv Q ; (R \wedge \neg(P \wand \neg Q)) * \true} \using \starl \]
\justifies
\seq{Q ; R * P}{(R \wedge \neg(P \wand \neg Q)) * \true} \using \display \]
\justifies
\seq{Q \wedge (R * P)}{(R \wedge \neg(P \wand \neg Q)) * \true} \using \andl \]
\justifies
\seq{\aemp ; Q \wedge (R * P)}{(R \wedge \neg(P \wand \neg Q)) * \true} \using \aunitl \]
\justifies
\seq{\aemp}{Q \wedge (R * P) \implies (R \wedge \neg(P \wand \neg Q)) * \true} \using \impr
\end{prooftree}\]}
\caption{A $\displayCBI$ derivation of the $\LAX_\CBI$ axiom (\ref{ax:wand1}) under the embedding
$A \mapsto (\seq{\aemp}{\revembed{A}})$, needed for the proof of Proposition~\ref{prop:LAXBIplus_admissible}.
\label{fig:drv_ax_wand1}}
\end{figure}

\begin{proposition}
\label{prop:LAXBIplus_admissible} The axioms and proof rules of
$\LAX_{\CBI}$ (cf.\ Defn.~\ref{defn:ML_proof_theory}) are
admissible in $\displayCBI$ under the embedding $A \mapsto
(\seq{\aemp}{\revembed{A}})$ from modal logic formulas to
consecutions.
\end{proposition}

\begin{proof}
First, we note that all of the proof rules of $\LAX_{\CBI}$,
except $\subst$, are easily derivable in $\displayCBI$ under
the embedding.  The rule $\subst$ is admissible in
$\displayCBI$ (under the embedding) because each of its proof
rules is closed under the substitution of arbitrary formulas
for propositional variables; in the case of the axiom rule
$\ax$ this requires an appeal to
Proposition~\ref{prop:DLBI_identity}.

It remains to show that $\seq{\aemp}{\revembed{A}}$ is
$\displayCBI$-derivable for every axiom ${A}$ of $\LAX_\CBI$.
The $\AX_\CBI$ axioms are mainly straightforward, with the chief
exceptions being axioms (\ref{ax:wand1}) and (\ref{ax:wand2}).
(We remark that axioms (8) and (9) are straightforward once one
has $\displayCBI$ proofs that $\neg$ and $\mneg$ commute; see
Figure~\ref{fig:displayBI_proof} for a proof of $\seq{\mneg\neg
F}{\neg\mneg F}$.) In the case of $\AX_\CBI$
axiom~(\ref{ax:wand1}), we need to show the consecution
$\seq{\aemp}{Q \wedge (R * P) \implies (R \wedge \neg(P \wand
\neg Q)) * \true}$ is provable in $\displayCBI$. We give a
suitable derivation in Figure~\ref{fig:drv_ax_wand1}. The
treatment of $\AX_\CBI$ axiom~(\ref{ax:wand2}) is broadly
similar.  It remains to treat the generic modal logic axioms of
$\LAX_\CBI$, which again are mainly straightforward and involve
showing distribution of the modalities over $\vee$. E.g., in
the case of the axiom $\lollydisjl$ we require to show that
$\seq{\aemp}{\neg((P \vee Q) \wand \neg R) \leftrightarrow
\neg(P \wand \neg R) \vee \neg(Q \wand \neg R)}$ is
$\displayCBI$-derivable.  We give a derivation of one direction
of this bi-implication in Figure~\ref{fig:drv_ax_lollydisjl}.
The other direction of the bi-implication, and the other
axioms, are derived in a similar fashion. \end{proof}

\begin{figure}[t]
{\small
\[\begin{prooftree}
\[\[\[\[\[\[\[\[\[\[\[\[\[
\[\[\[
\[ \justifies \seq{P}{P} \using \ax\]
\[\[\[
\justifies
\seq{R}{R} \using \ax\]
\justifies
\seq{\ainv R}{\ainv R} \using \display \]
\justifies
\seq{\neg R}{\ainv R} \using \negl \]
\justifies
\seq{P \wand \neg R}{\minv P , \ainv R} \using \wandl \]
\justifies
\seq{P \wand \neg R ; Q \wand \neg R}{\minv P , \ainv R} \using \weakl \]
\justifies
\seq{P}{\ainv R , \minv(P \wand \neg R ; Q \wand \neg R)} \using \display \]
\[\[\[
\[ \justifies \seq{Q}{Q} \using \ax\]
\[\[\[
\justifies
\seq{R}{R} \using \ax\]
\justifies
\seq{\ainv R}{\ainv R} \using \display\]
\justifies
\seq{\neg R}{\ainv R} \using \negl \]
\justifies
\seq{Q \wand \neg R}{\minv Q , \ainv R} \using \wandl \]
\justifies
\seq{P \wand \neg R ; Q \wand \neg R}{\minv Q , \ainv R} \using \weakl \]
\justifies
\seq{Q}{\ainv R , \minv(P \wand \neg R ; Q \wand \neg R)} \using \display \]
\justifies
\seq{P \vee Q}{(\ainv R , \minv(P \wand \neg R ; Q \wand \neg R)) ; (\ainv R , \minv(P \wand \neg R ; Q \wand \neg R))} \using \orl \]
\justifies
\seq{P \vee Q}{\ainv R , \minv(P \wand \neg R ; Q \wand \neg R)} \using \contrr \]
\justifies
\seq{(P \wand \neg R ; Q \wand \neg R), P \vee Q}{\ainv R} \using \display \]
\justifies
\seq{(P \wand \neg R ; Q \wand \neg R), P \vee Q}{\neg R} \using \negr \]
\justifies
\seq{P \wand \neg R ; Q \wand \neg R}{(P \vee Q) \wand \neg R} \using \wandr \]
\justifies
\seq{\ainv(P \vee Q) \wand \neg R ; Q \wand \neg R}{\ainv P \wand \neg R} \using \display \]
\justifies
\seq{\ainv(P \vee Q) \wand \neg R ; Q \wand \neg R}{\neg(P \wand \neg R)} \using \negr \]
\justifies
\seq{\ainv(P \vee Q) \wand \neg R ; \ainv\neg(P \wand \neg R)}{\ainv Q \wand \neg R} \using \display \]
\justifies
\seq{\ainv(P \vee Q) \wand \neg R ; \ainv\neg(P \wand \neg R)}{\neg(Q \wand \neg R)} \using \negr \]
\justifies
\seq{\ainv(P \vee Q) \wand \neg R}{\neg(P \wand \neg R) ; \neg(Q \wand \neg R)} \using \display \]
\justifies
\seq{\neg((P \vee Q) \wand \neg R)}{\neg(P \wand \neg R) ; \neg(Q \wand \neg R)} \using \negl \]
\justifies
\seq{\neg((P \vee Q) \wand \neg R)}{\neg(P \wand \neg R) \vee \neg(Q \wand \neg R)} \using \orr \]
\justifies
\seq{\aemp ; \neg((P \vee Q) \wand \neg R)}{\neg(P \wand \neg R) \vee \neg(Q \wand \neg R)} \using \aunitl \]
\justifies
\seq{\aemp}{\neg((P \vee Q) \wand \neg R) \rightarrow \neg(P \wand \neg R) \vee \neg(Q \wand \neg R)} \using \impr
\end{prooftree}\]}
\caption{A $\displayCBI$ derivation of (one direction of) the $\LAX_\CBI$ axiom $\lollydisjl$ under the embedding
$A \mapsto (\seq{\aemp}{\revembed{A}})$, needed for the proof of Proposition~\ref{prop:LAXBIplus_admissible}.
\label{fig:drv_ax_lollydisjl}}
\end{figure}

The following corollary of
Proposition~\ref{prop:LAXBIplus_admissible} is immediate by
induction over the structure of $\LAX_\CBI$ proofs.

\begin{corollary}
\label{cor:ML_to_DLCBI} If ${A}$ is provable in $\LAX_\CBI$
then $\seq{\aemp}{\revembed{A}}$ is provable in $\displayCBI$.
\end{corollary}

We write $\drvequiv{F}{G}$, where $F$ and $G$ are
$\CBI$-formulas, to mean that both $\seq{F}{G}$ and
$\seq{G}{F}$ are provable (in $\displayCBI$), and call
$\drvequiv{F}{G}$ a \emph{derivable equivalence} (of
$\displayCBI$).  We observe that derivable equivalence in
$\displayCBI$ is indeed an equivalence relation: it is
reflexive by Proposition~\ref{prop:DLBI_identity}, symmetric by
definition and transitive by the $\displayCBI$ rule $\cut$.

\begin{figure}
{\small
\[\begin{prooftree}
\[\[\[\[\[\[\[\[\[\[\[\[\[\[\[
\[\mbox{(I.H.)} \leadsto \seq{F_1}{\revembed{\embed{F_1}}} \]
\justifies
\seq{\minv\revembed{\embed{F_1}}}{\minv F_1} \using \display \]
\justifies
\seq{\mneg\revembed{\embed{F_1}}}{\minv F_1} \using \mnegl \]
\justifies
\seq{\ainv\minv F_1}{\ainv\mneg\revembed{\embed{F_1}}} \using \display \]
\justifies
\seq{\ainv\minv F_1}{\neg\mneg\revembed{\embed{F_1}}} \using \negr \]
\justifies
\seq{\ainv\neg\mneg\revembed{\embed{F_1}}}{\minv F_1} \using \display \]
\justifies
\seq{\neg\neg\mneg\revembed{\embed{F_1}}}{\minv F_1} \using \negl \]
\justifies
\seq{F_1}{\minv\neg\neg\mneg\revembed{\embed{F_1}}} \using \display \]
\[\[\[\[\[\[\[
\[\mbox{(I.H.)} \leadsto \seq{F_2}{\revembed{\embed{F_2}}} \]
\justifies
\seq{\minv\revembed{\embed{F_2}}}{\minv F_2} \using \display \]
\justifies
\seq{\mneg\revembed{\embed{F_2}}}{\minv F_2} \using \mnegl \]
\justifies
\seq{\ainv\minv F_2}{\ainv\mneg\revembed{\embed{F_2}}} \using \display \]
\justifies
\seq{\ainv\minv F_2}{\neg\mneg\revembed{\embed{F_2}}} \using \negr \]
\justifies
\seq{\ainv\neg\mneg\revembed{\embed{F_2}}}{\minv F_2} \using \display \]
\justifies
\seq{\neg\neg\mneg\revembed{\embed{F_2}}}{\minv F_2} \using \negl \]
\justifies
\seq{F_2}{\minv\neg\neg\mneg\revembed{\embed{F_2}}} \using \display \]
\justifies
\seq{F_1 \mor F_2}{\minv\neg\neg\mneg\revembed{\embed{F_1}} , \minv\neg\neg\mneg\revembed{\embed{F_2}}} \using \morl \]
\justifies
\seq{\neg\neg\mneg\revembed{\embed{F_1}} , \neg\neg\mneg\revembed{\embed{F_2}}}{\minv F_1 \mor F_2} \using \display \]
\justifies
\seq{\neg\neg\mneg\revembed{\embed{F_1}} * \neg\neg\mneg\revembed{\embed{F_2}}}{\minv F_1 \mor F_2} \using \starl \]
\justifies
\seq{F_1 \mor F_2}{\minv \neg\neg\mneg\revembed{\embed{F_1}} * \neg\neg\mneg\revembed{\embed{F_2}}} \using \display \]
\justifies
\seq{F_1 \mor F_2}{\mneg(\neg\neg\mneg\revembed{\embed{F_1}} * \neg\neg\mneg\revembed{\embed{F_2}})} \using \mnegr \]
\justifies
\seq{\ainv\mneg(\neg\neg\mneg\revembed{\embed{F_1}} * \neg\neg\mneg\revembed{\embed{F_2}})}{\ainv F_1 \mor F_2} \using \display \]
\justifies
\seq{\neg\mneg(\neg\neg\mneg\revembed{\embed{F_1}} * \neg\neg\mneg\revembed{\embed{F_2}})}{\ainv F_1 \mor F_2} \using \negl \]
\justifies
\seq{F_1 \mor F_2}{\ainv\neg\mneg(\neg\neg\mneg\revembed{\embed{F_1}} * \neg\neg\mneg\revembed{\embed{F_2}})} \using \display \]
\justifies
\seq{F_1 \mor F_2}{\neg\neg\mneg(\neg\neg\mneg\revembed{\embed{F_1}} * \neg\neg\mneg\revembed{\embed{F_2}})} \using \negr
\end{prooftree}\]}
\caption{A $\displayCBI$ proof for the non-trivial case of
Lemma~\ref{lem:embed_cancel_formula}.\label{fig:embed_cancel_formula}}
\end{figure}

\begin{lemma}
\label{lem:embed_cancel_formula}
$\drvequiv{F}{\revembed{\embed{F}}}$ is a derivable equivalence
of $\displayCBI$ for any CBI-formula $F$.
\end{lemma}

\begin{proof}
By combining the definitions of $\embed{-}$ and
$\revembed{\;-\;}$ (cf.~Defns.~\ref{defn:embed}
and~\ref{defn:revembed}) we obtain the following definition of
$\revembed{\embed{-}}$, given by structural induction on
$\CBI$-formulas:
\[\begin{array}{r@{\hspace{0.3cm}}c@{\hspace{0.3cm}}l@{\hspace{0.3cm}}l}
\revembed{\embed{F}} & = & F & \mbox{where }F \in \{P,\true,\false,\mtrue\} \\
\revembed{\embed{\neg F}} & = & \neg\revembed{\embed{F}} \\
\revembed{\embed{F_1 \mathrel{?} F_2}} & = & \revembed{\embed{F_1}} \mathrel{?} \revembed{\embed{F_2}}
& \mbox{where }? \in \{\wedge,\vee,\rightarrow,*\} \\
\revembed{\embed{\mfalse}} & = & \neg\neg\mfalse \\
\revembed{\embed{\mneg F}} & = & \neg\neg\mneg\revembed{\embed{F}} \\
\revembed{\embed{F_1 \wand F_2}} & = & \neg\neg(\revembed{\embed{F_1}} \wand \neg\neg\revembed{\embed{F_2}}) \\
\revembed{\embed{F_1 \mor F_2}} & = & \neg\neg\mneg(\neg\neg\mneg\revembed{\embed{F_1}} * \neg\neg\mneg\revembed{\embed{F_2}})
\end{array}\]
With this in mind, we now proceed by structural induction on
$F$.  The base cases, in which $\revembed{\embed{F}} = F$, are
immediate since $\drvequiv{F}{F}$ is a derivable equivalence of
$\displayCBI$ by Proposition~\ref{prop:DLBI_identity}.  Most of
the other cases are straightforward using the induction
hypothesis and the fact that $\drvequiv{\neg\neg F}{F}$ is
easily seen to be a derivable equivalence of $\displayCBI$.  We
show one direction of the only non-trivial case, $F = F_1 \mor
F_2$, in Figure~\ref{fig:embed_cancel_formula}.  The reverse
direction is similar.
\end{proof}

The following two lemmas, which show how to construct proofs of
arbitrary valid consecutions given proofs of arbitrary valid
formulas, are standard in showing completeness of display
calculi relative to Hilbert-style proof systems, and were first
employed by Gor\'e~\cite{Gore:96-2}.

\begin{lemma}
\label{lem:drv_structure_formula} For any structure $X$ the
consecutions $\seq{X}{\dlantform{X}}$ and
$\seq{\dlconform{X}}{X}$ are both $\displayCBI$-provable.
\end{lemma}

\begin{proof}
By structural induction on $X$.  The case where $X$ is a
formula $F$ follows directly from
Proposition~\ref{prop:DLBI_identity}. The other cases all
follow straightforwardly from the induction hypothesis and the
logical rules of $\displayCBI$. E.g., when $X = \minv Y$ we
have $\dlantform{X} = \mneg\dlconform{Y}$ and $\dlconform{X} =
\mneg\dlantform{Y}$, and proceed as follows:
{\small\[\begin{array}{c@{\hspace{1cm}}c}
\begin{prooftree}
\[
\[\mbox{(I.H.)} \leadsto \seq{\dlconform{Y}}{Y}\]
\justifies \seq{\minv Y}{\minv\dlconform{Y}} \using \display \]
\justifies \seq{\minv Y}{\mneg\dlconform{Y}} \using \mnegr
\end{prooftree}
&
\begin{prooftree}
\[
\[\mbox{(I.H.)} \leadsto \seq{Y}{\dlantform{Y}} \]
\justifies \seq{\minv\dlantform{Y}}{\minv Y} \using \display \]
\justifies \seq{\mneg\dlantform{Y}}{\minv Y} \using \mnegl
\end{prooftree}
\end{array}\]}
The remaining cases are similar.
\end{proof}

\begin{lemma}
\label{lem:embed_cancel} If
$\seq{\aemp}{\revembed{\embed{\dlantform{X} \rightarrow
\dlconform{Y}}}}$ is $\displayCBI$-provable then so is
$\seq{X}{Y}$.
\end{lemma}

\begin{proof}
We first note that $\revembed{\embed{\dlantform{X} \rightarrow
\dlconform{Y}}} = \revembed{\embed{\dlantform{X}}} \rightarrow
\revembed{\embed{\dlconform{Y}}}$, and then build a
$\displayCBI$ proof of $\seq{X}{Y}$ as follows:
{\small\[\begin{prooftree}
\[\[
\[\mbox{(assumption)} \leadsto \seq{\aemp}{\revembed{\embed{\dlantform{X}}} \rightarrow \revembed{\embed{\dlconform{Y}}}} \]
\hspace{-0.7cm}
\[\[
\[\mbox{(Lemma~\ref{lem:drv_structure_formula})} \leadsto \seq{X}{\dlantform{X}} \]
\[\mbox{(Lemma~\ref{lem:embed_cancel_formula})} \leadsto \seq{\dlantform{X}}{\revembed{\embed{\dlantform{X}}}} \]
\justifies
\seq{X}{\revembed{\embed{\dlantform{X}}}} \using \cut \]
\[
\[\mbox{(Lemma~\ref{lem:embed_cancel_formula})} \leadsto \seq{\revembed{\embed{\dlconform{Y}}}}{\dlconform{Y}} \]
\[\mbox{(Lemma~\ref{lem:drv_structure_formula})} \leadsto \seq{\dlconform{Y}}{Y} \]
\justifies
\seq{\revembed{\embed{\dlconform{Y}}}}{Y} \using \cut \]
\justifies
\seq{\revembed{\embed{\dlantform{X}}} \rightarrow \revembed{\embed{\dlconform{Y}}}}{\ainv X ; Y} \using \impl \]
\justifies
\seq{\aemp}{\ainv X ; Y} \using \cut \]
\justifies
\seq{\aemp ; X}{Y} \using \display \]
\justifies
\seq{X}{Y} \using \aunitl
\end{prooftree}\]}
\end{proof}

We can now prove the completeness of $\displayCBI$ with respect
to $\CBI$-validity.

\paragraph{\em Proof of Theorem~\ref{thm:DLBI_complete}.} Let
$\seq{X}{Y}$ be a valid consecution.
Then $\embed{\dlantform{X} \rightarrow \dlconform{Y}}$ is
$\LAX_{\CBI}$-provable by Corollary~\ref{cor:completeness2}. By
Corollary~\ref{cor:ML_to_DLCBI},
$\seq{\aemp}{\revembed{\embed{\dlantform{X} \rightarrow
\dlconform{Y}}}}$ is then provable in $\displayCBI$ and thus,
by Lemma~\ref{lem:embed_cancel}, $\seq{X}{Y}$ is
$\displayCBI$-provable as required.

\section{Examples of $\CBI$-models}
\label{sec:examples}
\renewcommand{\thetheorem}{\thesection.\arabic{theorem}}
\setcounter{theorem}{0}

In this section we give some concrete examples of
$\CBI$-models, and some general constructions for forming new
models.  In most of our examples the relational monoid
operation $\circ$ is actually a partial function, and in these
cases we treat it as such (e.g., by writing $x \circ y = z$
rather than $x \circ y = \{z\}$).

\begin{proposition}[Abelian groups as $\CBI$-models]
\label{prop:abelian_group_models} Any Abelian group $\abgroup$
can be understood as a $\CBI$-model $\langle
R,\circ,e,\inv,e\rangle$.  Conversely, if $\cbimodel$ is a
$\CBI$-model with $\circ$ a partial function, then imposing the
condition $\infty = e$ forces $\circ$ to be total, whence
$\abgroup$ is an Abelian group.
\end{proposition}

\begin{proof}
($\Rightarrow$) Let $\abgroup$ be an Abelian group.  To see that $\monoid$ is a $\BBI$-model, we just note that
$\circ$ is associative and commutative and that $e$ is the unit
of $\circ$ by the group axioms.  By the uniqueness of group
inverses, we then have that $\inv x$ is the unique $y$ such
that $e \in x \circ y$.  Thus $\langle R, \circ, e, \inv,
e\rangle$
is a $\CBI$-model, as required. \\

\noindent($\Leftarrow$) Let $\cbimodel$ be a $\CBI$-model with
$\infty = e$ and $\circ$ a partial function.  First note that,
by the latter two facts, we have $\inv x \circ x = \infty = e$
for all $x \in R$. Now for any $x,y \in R$ we observe that
$\inv x \circ (x \circ y) = (\inv x \circ x) \circ y = e \circ
y = y$. Thus $\inv x \circ (x \circ y)$ is defined, which can
only be the case if $x \circ y$ is defined. Thus $\circ$ is in
fact a total function.

To see that $\abgroup$ is an Abelian group, we first observe
that, since $\circ$ is a total function by the above, $\monoid$
is a total commutative monoid by the conditions imposed on
$\BBI$-models.  The uniqueness of group inverses then follows
immediately from the $\CBI$-model conditions and the fact that
$\infty = e$.
\end{proof}

The following example, which looks at some typical resource
interpretations of $\CBI$-formulas inside an Abelian group
model, builds on the ``vending machine'' model for $\BI$ given
by Pym, O'Hearn and Yang~\cite{Pym-OHearn-Yang:04}, which
itself was inspired by Girard's well-known ``Marlboro and
Camel'' illustration of linear logic~\cite{Girard:95}.

\begin{example}[Personal finance]
\label{ex:money} Let $\langle \integers, +, 0, \inv \rangle$ be
the Abelian group of integers under addition with identity $0$,
where $\inv$ is the usual unary minus.  This group can be
understood as a $\CBI$-model $\langle \integers, +, 0, \inv, 0
\rangle$ by Proposition~\ref{prop:abelian_group_models}. We
view the elements of this model as financial resources, i.e\
money (which we shall measure in pounds sterling, \pounds),
with positive and negative integers representing respectively
\emph{credit} and \emph{debt}.  We read the $\CBI$-satisfaction
relation $\sat{\pounds m}{F}$ informally as ``$\pounds m$ is
enough to make $F$ true'', and show how to read some example
$\CBI$-formulas according to this interpretation.

Let $C$ and $W$ be atomic formulas denoting respectively the
ability to buy cigarettes costing \pounds 5 and whisky costing
\pounds 20, so that we have $\sat{\pounds m}{C} \Leftrightarrow
m \geq 5$ and $\sat{\pounds m}{W} \Leftrightarrow m \geq 20$.
Then the formula $C \wedge W$ denotes the ability to buy
cigarettes and the ability to buy whisky (but not necessarily
to buy both together):
\[\begin{array}{rcl}
\sat{\pounds m}{C \wedge W} & \Leftrightarrow & \sat{\pounds m}{C} \mbox{ and } \sat{\pounds m}{W} \\
& \Leftrightarrow & m \geq 5 \mbox{ and } m \geq 20 \\
& \Leftrightarrow & m \geq 20
\end{array}\]
In contrast, the formula $C * W$ denotes the ability to buy
both cigarettes and whisky together:
\[\begin{array}{rcl}
\sat{\pounds m}{C * W} & \Leftrightarrow & \exists m_1, m_2 \in \integers.
\;\pounds m = \pounds m_1 + \pounds m_2
\mbox{ and } \sat{\pounds m_1}{C} \mbox { and } \sat{\pounds m_2}{W} \\
& \Leftrightarrow & \exists m_1, m_2 \in \integers.
\;m = m_1 + m_2
\mbox{ and } m_1 \geq 5 \mbox { and } m_2 \geq 20 \\
& \Leftrightarrow & m \geq 25
\end{array}\]
The multiplicative implication $C \wand W$ denotes the fact
that if one acquires enough money to buy cigarettes then the
resulting balance of funds is sufficient to buy whisky:
\[\begin{array}{rcl}
\sat{\pounds m}{C \wand W} & \Leftrightarrow & \forall m' \in \integers.\;
\sat{\pounds m'}{C} \mbox{ implies } \sat{\pounds m + \pounds m'}{W} \\
& \Leftrightarrow & \forall m' \in \integers.\;
m' \geq 5 \mbox{ implies }  m + m' \geq 20 \\
& \Leftrightarrow & m \geq 15
\end{array}\]
We remark that all of the above formulas are $\BBI$-formulas,
and so would be interpreted in exactly the same way in the
$\BBI$-model $\langle\integers,+,0\rangle$.  Let us examine the
 multiplicative connectives that are particular to $\CBI$.  We have
$\sat{\pounds m}{\mfalse} \Leftrightarrow m \neq 0$, so that
$\mfalse$ simply denotes the fact that one has either some
credit or some debt.  (This is exactly the interpretation of
the formula $\neg\mtrue$, a collapse induced by the fact that
$e$ and $\infty$ coincide in the Abelian group model.)  Now
consider the formula $\mneg C$. We have:
\[
\sat{\pounds m}{\mneg C} \;\Leftrightarrow\; \notsat{\inv\pounds m}{C} \;\Leftrightarrow\; -m < 5 \;\Leftrightarrow\; m > -5
\]
So $\mneg C$ denotes the fact that one's debt, if any, is
strictly less than the price of a pack of cigarettes. As for
the multiplicative disjunction, $C \mor W$, we have:
\[\begin{array}{rcl}
&& \sat{\pounds m}{C \mor W} \\
& \Leftrightarrow & \forall m_1,m_2 \in \integers.\;
\inv\pounds m = \pounds m_1 + \pounds m_2 \mbox{ implies } \sat{\inv\pounds m_1}{C} \mbox{ or } \sat{\inv\pounds m_2}{W} \\
 & \Leftrightarrow & \forall m_1,m_2 \in \integers.\;
\inv m = m_1 + m_2 \mbox{ implies } \inv m_1 \geq 5 \mbox{ or } \inv m_2 \geq 20 \\
 & \Leftrightarrow & \forall m_1,m_2 \in \integers.\;
m + m_1 + m_2 = 0 \mbox{ implies } m_1 \leq \inv 5 \mbox{ or } m_2 \leq \inv 20 \\
 & \Leftrightarrow & \forall m_1,m_2 \in \integers.\;
(m + m_1 + m_2 = 0 \mbox{ and } m_1 > \inv 5) \mbox{ implies } m_2 \leq \inv 20 \\
& \Leftrightarrow & m \geq 24
\end{array}\]
It is not immediately obvious how to read this formula
informally. However, observing that $C \mor W$ is semantically
equivalent to $\mneg C \wand W$ and to $\mneg W \wand C$, the
meaning becomes perfectly clear: if one spends strictly less than the
price of a pack of cigarettes, then one will still have enough
money to buy whisky, and vice versa.
\end{example}

We remark that, in fact, there is a logic in the relevantist
mould, called \emph{Abelian logic}, whose models are exactly the lattice-ordered Abelian
groups~\cite{Meyer-Slaney:89}.

\begin{proposition}[Effect algebras as $\CBI$-models]
\label{prop:effect_algebra_models} \emph{Effect algebras},
which arise in the mathematical foundations of
quantum-mechanical systems~\cite{Foulis-Bennett:94}, are
exactly $\CBI$-models $\cbimodel$ such that $\circ$ is a
partial function and $\infty$ is nonextensible (i.e.\ $x \circ
\infty$ is undefined for all $x \neq e$).
\end{proposition}

The $\CBI$-models constructed in the next examples are all
effect algebras.

\begin{example}[Languages]
\label{ex:reg_langs} Let $\Sigma$ be an alphabet and let
$\langs{\Sigma}$ be any set of languages over $\Sigma$ that
is closed under union and complement and contains the empty
language $\epsilon$ (e.g., the set of regular languages over $\Sigma$). Write $\Sigma^*$ for the set of all words over $\Sigma$, and note that $\Sigma^* \in \langs{\Sigma}$. Let
$L_1 + L_2$ be the union of disjoint languages $L_1$ and $L_2$,
with $L_1 + L_2$ undefined if $L_1$ and $L_2$ are not disjoint.
Clearly $\langle\langs{\Sigma} , + , \epsilon \rangle$ is a
partial commutative monoid.  Furthermore, for any language $L$,
its complement $\overline{L} = \Sigma^* \setminus L$ is the
unique language such that $L + \overline{L} = \Sigma^*$. Thus
$\langle \langs{\Sigma} , + , \epsilon, \overline{\;\cdot\;},
\Sigma^* \rangle$ is a $\CBI$-model. To see that it is also
an effect algebra, just notice that $+$ is a partial function
and $\Sigma^*$ is nonextensible because $\Sigma^* + L$ is undefined
for any $L \neq \epsilon$.
\end{example}

\begin{example}[Action communication]
Let $A$ be any set of objects (to be understood as CCS-style
``actions''~\cite{Milner:89}), define the set $\overline{A} = \{\overline{a}
\mid a \in A\}$ to be disjoint from $A$, and let elements
$0,\tau \not\in A \cup \overline{A}$, whence we write $B
\defeq A \cup \overline{A} \cup \{0,\tau\}$.  We extend the operation
$\overline{\,\cdot\,}$ to $B \rightarrow B$ by $\overline{0}
\defeq \tau$ and $\overline{\overline{a}}
\defeq a$, and define a partial binary operation \mbox{$\cdot\mid\cdot$} with
type $B \times B \partialfn B$ as follows:
\[\begin{array}{ccl}
b \mid c & \defeq & \left\{\begin{array}{ll}
b & \mbox{if $c = 0$} \\
\tau & \mbox{if $c = \overline{b}$} \\
\mbox{undefined} & \mbox{otherwise}
\end{array}\right.
\end{array}\]
The operation \mbox{$\cdot\mid\cdot$} models a very simplistic
version of communication between actions: communication with
the empty action $0$ has no effect, communication between a
pair of dual actions $b$ and $\overline{b}$ (which may be read,
e.g., as ``send $b$'' and ``receive $b$'') results in the
``successful communication'' action $\tau$, and all other
communications are disallowed. It is easy to check that
$\langle B , \cdot | \cdot, 0\rangle$ is a partial commutative
monoid. Furthermore, for any $b \in B$, we clearly have
$\overline{b}$ the unique element with $b\mid \overline{b} =
\tau$. Thus \mbox{$\langle B , \cdot\mid\cdot\; , 0 ,
\overline{\,\cdot\,} , \tau \rangle$} is a $\CBI$-model.
Furthermore, it is clearly an effect algebra, because
\mbox{$\cdot\mid\cdot$} is a partial function and $\tau$ is
nonextensible.
\end{example}

\begin{example}[Generalised heaps]
\label{ex:gen_heaps} A natural question is whether the
\emph{heap models} of $\BBI$ employed in separation logic (see
e.g.~\cite{Calcagno-OHearn-Yang:07}) are also $\CBI$-models.
Consider the basic heap model given by the partial commutative
monoid $\langle H,\circ,e\rangle$, where $H\defeq \nat
\partialfn_{\mathrm{fin}} \mathbb{Z}$ is the set of \emph{heaps} (i.e.\
partial functions mapping finitely many natural numbers to
integers), $h_1 \circ h_2$ is the union of partial functions
$h_1$ and $h_2$ when their domains are disjoint (and undefined
otherwise), and $e$ is the function with empty domain.
Unfortunately, no choice of $\infty$ for $\langle
H,\circ,e\rangle$ gives rise to a $\CBI$-model.

However, it is possible to embed the set of heaps $H$ above into a
more general structure $\langle H',\circ',e'\rangle$, where $H'
\defeq \nat \rightarrow \pow{\integers}$ is the set of
(total) functions from natural numbers to \emph{sets} of
integers (we may additionally require that $h(n) \neq
\emptyset$ for finitely or cofinitely many $n \in \nat$). Then
$\circ': H' \times H'
\partialfn H'$ is defined by: if $h_1(n)$ and $h_2(n)$ are
disjoint for all $n$, then $(h_1 \circ' h_2)(n) = h_1(n) \cup
h_2(n)$, otherwise $h_1 \circ' h_2$ is undefined. The unit $e'$
is defined by $e'(n) = \emptyset$ for all $n \in \nat$.  A
$\CBI$-model $\langle H',\circ',e',\inv,\infty\rangle$ is then obtained by defining $\infty(n) =
\integers$ and $(\inv h)(n) = \integers \setminus h(n)$ for all
$n \in \nat$.

We note that this model behaves quite differently than $\langle
H, \circ, e\rangle$: generalised heaps with overlapping domains
can be composed providing that their \emph{contents} do not
overlap for any point in the domain.  We consider the
interpretation of some ``separation logic-like'' formulas
inside this model. Let $X \subseteq \integers$ be some fixed
set of integers and define the atomic formula $4 \mapsto X$ by
the following:
\[
\sat{h}{4 \mapsto X} \;\;\Leftrightarrow\;\; h(4) = X
\]
i.e., the formula $4 \mapsto X$ denotes those generalised heaps
with contents exactly $X$ at location $4$.  This can be seen as
the set-based analogue of the $\mapsto$ predicate in standard
separation logic~\cite{Reynolds:02} (with fixed arguments for
simplicity, as we are working in a propositional setting). Then
we have, for example:
\[\begin{array}{rcl}
\sat{h}{(4 \mapsto X) * \true} & \Leftrightarrow & h = h_1 \circ' h_2 \mbox{ and } \sat{h_1}{4 \mapsto X} \mbox{ and } \sat{h_2}{\true} \\
& \Leftrightarrow & (\forall n \in \nat.\ h(n) = h_1(n) \cup h_2(n)) \mbox{ and } h_1(4) = X \\
& \Leftrightarrow & X \subseteq h(4)
\end{array}\]
so that the formula $(4 \mapsto X) * \true$ denotes the general
heaps which contain every element of $X$ at location $4$.  If
we then take the multiplicative negation of this formula, we
have, using the above:
\[\begin{array}{rcl}
\sat{h}{\mneg((4 \mapsto X) * \true)} & \Leftrightarrow & \notsat{\inv h}{(4 \mapsto X) * \true} \\
& \Leftrightarrow & X \not\subseteq (\inv h)(4) \\
& \Leftrightarrow & X \not\subseteq (\integers \setminus h(4)) \\
& \Leftrightarrow & \exists x \in X.\ x \in h(4)
\end{array}\]
i.e., this formula denotes the general heaps containing
\emph{some} element from $X$ at location $4$.  So, in this
case, the multiplicative negation has the effect of changing a
universal quantifier to an existential one.  The meaning of
multiplicative disjunctions, however, is typically very
complicated.  For example, picking a second set $Y \subseteq
\integers$ and defining the atomic formula $4 \mapsto Y$ in the
same way as $4 \mapsto X$, we have, using previous derivations:
\[\begin{array}{cl}
& \sat{h}{((4 \mapsto X) * \true) \mor ((4 \mapsto Y) * \true)} \\
\Leftrightarrow & \sat{h}{\mneg((4 \mapsto X) * \true) \wand ((4 \mapsto Y) * \true)} \;\;\mbox{ (by Lemma~\ref{lem:cbieq})}\\
\Leftrightarrow & \forall h'.\ (h \circ' h' \mbox{ defined and } \sat{h'}{\mneg((4 \mapsto X) * \true)})
\mbox{ implies } \sat{h \circ h'}{(4 \mapsto Y) * \true} \\
\Leftrightarrow & \forall h'.\ (h \circ' h' \mbox{ defined and } \exists x \in X.\ x \in h'(4))
\mbox{ implies } Y \subseteq h(4) \cup h'(4) \\
\Leftrightarrow & X \subseteq h(4) \mbox{ or } Y \subseteq h(4) \mbox{ or }
\exists z \in \integers.\  (X \setminus h(4)) = (Y \setminus h(4)) = \{z\}
\end{array}\]
so that this disjunction denotes those general heaps that
either contain one of $X$ and $Y$, or are missing a single
common element from $X$ and $Y$, at location $4$. We give a
short proof of the last equivalence above, since it is not
especially obvious.

\paragraph{($\Leftarrow$)}
If $X \subseteq h(4)$ then the required implication holds
vacuously because $x \in X \cap h'(4)$ implies $h \circ' h'$ is
undefined.  If $Y \subseteq h(4)$ then the implication also
holds trivially because the consequent is immediately true.
Lastly, suppose $X \setminus h(4) = Y \setminus h(4) = \{z\}$.
Let $h'$ be any heap with $h \circ' h'$ defined and $x \in
h'(4)$ for some $x \in X$.  We must then have $x = z$ because
$h(4)$ and $h'(4)$ must be disjoint and $X \setminus h(4) =
\{z\}$. Then, since also $Y \setminus h(4) = \{z\}$,we have $Y
\subseteq h(4) \cup \{z\} \subseteq h(4) \cup h'(4)$ as
required.

\paragraph{($\Rightarrow$)} If $X \subseteq h(4)$ or $Y
\subseteq h(4)$ we are trivially done. Now suppose that $X
\not\subseteq h(4)$ and $Y \not\subseteq h(4)$, so that there
are $x \in X \setminus h(4)$ and $y \in Y \setminus h(4)$. Let
$h'$ be given by $h'(4) = \{x\}$ and $h'(n) = \emptyset$ for
all other $n$, and note that $h \circ' h'$ is defined.  By
assumption, we have $Y \subseteq h(4) \cup h'(4) = h(4) \cup
\{x\}$, and thus $Y \setminus h(4) = \{x\}$ because $Y
\setminus h(4)$ is nonempty by assumption. It follows that $Y
\setminus h(4) = \{x\}$ for \emph{any} $x \in X \setminus
h(4)$, and so also $X \setminus h(4) = \{x\}$, as required.
\end{example}

We note that a number of general categorical constructions for
effect algebras have recently appeared in~\cite{Jacobs:09}.

Our next examples differ both from Abelian groups in that $e$
and $\infty$ are non-identical, and from effect algebras in
that $\infty$ is extensible.  Indeed, as shown by our
Example~\ref{ex:int_modulo} below, fixing the monoidal
structure of a $\CBI$-model does not in general determine the
choice of $\infty$.

\begin{example}[Bit arithmetic]
\label{ex:bit_arith} Let $n \in \nat$ and observe that an
$n$-bit binary number can be represented as an element of the
set $\{0,1\}^n$.  Let XOR and NOT be the usual logical
operations on binary numbers.  Then the following is a
$\CBI$-model:
\[\langle\{0,1\}^n, \mathrm{XOR}, \{0\}^n, \mathrm{NOT}, \{1\}^n\rangle\]
In this model, the resources $e$ and $\infty$ are the $n$-bit
representations of $0$ and $2^n-1$ respectively.
\end{example}

\begin{example}[Integer modulo arithmetic]
\label{ex:int_modulo} Consider the monoid $\langle
\mathbb{Z}_n,+_n,0 \rangle$, where $\mathbb{Z}_n$ is the set of
integers modulo $n$, and $+_n$ is addition modulo $n$. We can
form a $\CBI$-model from this monoid by choosing, for any
$m\in\mathbb{Z}_n$, $\infty \defeq m$ and $\inv k \defeq m -_n
k$ (where $-_n$ is subtraction modulo $n$).
\end{example}

\begin{example}[Syntactic models]
\label{ex:syntactic_model} Given an arbitrary monoid $\monoid$,
we give a syntactic construction to generate a $\CBI$-model
$\langle R',\circ',e',-',\infty'\rangle$. Consider the set $T$
of terms given by the grammar:
\[t\in T ::= r\in R \mid \infty \mid t \cdot t \mid - t\]
and let $\approx$ be the least congruence such that:
\[\begin{array}{l}
r_1 \circ r_2 = r \mbox{ implies } r_1 \cdot r_2 \approx r; \\
t_1 \cdot t_2 \approx t_2 \cdot t_1; \\
t_1 \cdot (t_2 \cdot t_3) \approx (t_1\cdot t_2) \cdot t_3; \\
\inv\inv t \approx t;  \\
t \cdot (\inv t) \approx \infty; \\
t_1\circ t_2 \approx \infty \mbox{ implies } t_1\approx \inv t_2.
\end{array}\]
We write $T{/}{\approx}$ for the quotient of
$T$ by the relation $\approx$, and $[t]$ for the equivalence
class of $t$. The required $\CBI$-model $\langle
R',\circ',e',\inv',\infty'\rangle$ is obtained by defining
$R'\defeq T{/}{\approx}$, $\circ'([t_1],[t_2])\defeq [t_1 \circ
t_2]$, $e'\defeq [e]$, $\inv'(t)\defeq [\inv t]$, and $\infty'\defeq
[\infty]$.
\end{example}

We now consider some general ways of composing $\CBI$-models.

\begin{lemma}[Disjoint union of $\CBI$-models]
\label{lem:disjoint_union_model} Let $\langle R_1, \circ_1,
e_1, \inv_1, \infty_1\rangle$ and $\langle R_2, \circ_2, e_2,$
$\inv_2, \infty_2\rangle$ be $\CBI$-models such that $R_1$ and
$R_2$ are disjoint and either $\infty_1 = e_1$ and $\infty_2 =
e_2$ both hold or $\infty_1,\infty_2$ are both nonextensible, i.e.\ $\infty_1
\circ_1 x = \emptyset$ for all $x \neq e_1$ and $\infty_2
\circ_2 x = \emptyset$ for all $x \neq e_2$.

Now let $R$ be the set obtained by identifying $e_1$ with $e_2$ and
$\infty_1$ with $\infty_2$ in $R_1 \cup R_2$, and write \mbox{$e =
e_1 = e_2$} and \mbox{$\infty = \infty_1 = \infty_2$} for the elements
obtained by this identification.  Define $\inv = \inv_1 \cup
\inv_2$ and $\circ = \circ_1 \cup \circ_2$.  Then $\cbimodel$
is a $\CBI$-model.
\end{lemma}

\begin{proof}
We start by observing that $\inv$ is indeed a function from $R$
to $R$ because $R_1$ and $R_2$ are assumed disjoint and, using
Proposition~\ref{prop:cbimodel_properties}, $\inv_1 e_1 =
\infty_1 = \infty_2 = \inv_2 e_2$, and similarly $\inv_1 \infty_1 =
\inv_2 \infty_2$. Thus $\inv e$ and $\inv \infty$ are well-defined.

We need to check that $\monoid$ is a $\BBI$-model.  The
commutativity of $\circ$ is immediate by the commutativity of
$\circ_1$ and $\circ_2$.  Similarly, $x \circ e =
\{x\}$ for all $x \in R$ because $e = e_1$ is a unit of
$\circ_1$ and $e = e_2$ is a unit of $\circ_2$.  To see that
$\circ$ is associative, we let $x,y,z \in R$ and show that $x
\circ (y \circ z) = (x \circ y) \circ z$ by case analysis.

\paragraph{\em Case: at least one of $x,y,z$ is $e$.} We are immediately done by the fact that
$e$ is a unit for $\circ$.

\paragraph{\em Case: at least one of $x,y,z$ is $\infty$.} We may assume
that none of $x,y,z$ is $e$, since these possibilities are
covered by the previous case, and so it follows by assumption
that $\infty_1$ and $\infty_2$ are nonextensible.  Consequently
$\infty \circ x = \emptyset$ for all $x \neq e$, so $x \circ (y
\circ z) = \emptyset = (x \circ y) \circ z$.

\paragraph{\em Case: all of $x,y,z \in R_1$.} We may assume by the previous cases that
none of $x,y,z$ is either $e$ or $\infty$, so we have $x \circ
(y \circ z) = x \circ_1 (y \circ_1 z)$ and $(x \circ y) \circ z
= (x \circ_1 y) \circ_1 z$, whence we are done by the
associativity of $\circ_1$.

\paragraph{\em Case: all of $x,y,z \in R_2$.}  Similar to the
case above.

\paragraph{\em Case: none of the above.}  We have $x \circ (y
\circ z) = \emptyset = (x \circ y) \circ z$ since $x \circ y =
\emptyset$ whenever $x \in R_1$, $y \in R_2$ and neither $x$
nor $y$ is $e$ or $\infty$.  This covers all the cases, so
$\circ$ is indeed associative.

Now to see that $\cbimodel$ is a $\CBI$-model, given $x \in R$
we need to show that $\inv x$ is the unique $y \in R$ such that
$\infty \in x \circ y$.  It is easily verified that $\infty \in
x \circ \inv x$ for all $x \in R$.  Now suppose that $\infty
\in x \circ y = x \circ_1 y \cup x \circ_2 y$ for some $y \in
R$. If $\infty \in x \circ_1 y$ then $y = \inv_1 x = \inv x$ as
required. Similarly, if $\infty \in x \circ_2 y$ then $y =
\inv_2 x = \inv x$.
\end{proof}

We remark that the restrictions on $\infty_1$ and $\infty_2$ in
Lemma~\ref{lem:disjoint_union_model} are needed in order to
ensure the associativity of $\circ$.  For example, if $x \in
R_1$ and $y \in R_2$ and $x,y \neq e$ then $(x \circ y) \circ
\inv y = \emptyset \circ \inv y = \emptyset$ while $x \circ (y
\circ \inv y) \supseteq x \circ \infty = x \circ_1 \infty_1$,
which is not empty in general.

\newcommand{\setprod}[1]{\mathord{\textstyle\bigotimes_{a \in A}}{#1}}
\newcommand{\elemprod}[1]{\mathord{\otimes_{a \in A}}\,{#1}}

\begin{lemma}[Generalised Cartesian product of $\CBI$-models]
\label{lem:product_model} Let $A$ be an ordered set and write
$\elemprod{x_a}$ for an ordered tuple indexed by the elements
of $A$.  Suppose that $M_a = \langle R_a, \circ_a, e_a, \inv_a,
\infty_a\rangle$ is a $\CBI$-model for each $a \in A$.  Then
$\langle R, \circ, \elemprod{e_a}, \inv,
\elemprod{\infty_a}\rangle$ is a $\CBI$-model, where $R$
denotes the $A$-ordered Cartesian product of the sets $R_a$,
and the operations $\circ: R \times R \rightarrow \pow{R}$ and
$\inv: R \rightarrow R$ are defined as follows:
\[\begin{array}{rcl}
\inv(\elemprod{x_a}) & = & \elemprod{(\inv_a x_a)} \\
\elemprod{x_a} \circ \elemprod{y_a} & = & \bigcup_{a \in A, w_a \in x_a \circ_a y_a} \{\elemprod{w_a}\}
\end{array}\]
\end{lemma}

\begin{proof}
In the following, all uses of $\bigcup$ notation should be understood as ranging over all $a \in A$ (we suppress the explicit subscript for legibility).
First, we need to check that $\langle R, \circ, \elemprod{e_a}
\rangle$ is a $\BBI$-model.  The commutativity of $\circ$
follows immediately from its definition and the commutativity
of each $\circ_a$. To see that $\elemprod{e_a}$ is a unit for
$\circ$ we observe:
\[\begin{array}{l}
\elemprod{x_a} \circ \elemprod{e_a} = \bigcup_{w_a \in x_a \circ_a e_a} \{\elemprod{w_a}\}
= \bigcup_{w_a \in \{x_a\}} \{\elemprod{w_a}\} = \{\elemprod{x_a}\}
\end{array}\]
Next we need to check that $\circ$ is associative.  Using the
the standard extension of $\circ$ to $\pow{R} \times \pow{R}
\rightarrow \pow{R}$ we have:
\[\begin{array}{rcl}
(\elemprod{x_a} \circ \elemprod{y_a}) \circ \elemprod{z_a} & = &
(\bigcup_{w_a \in x_a \circ_a y_a} \{\elemprod{w_a}\}) \circ \elemprod{z_a} \\
& = & \bigcup_{w_a \in x_a \circ_a y_a} (\elemprod{w_a} \circ \elemprod{z_a}) \\
& = & \bigcup_{w_a \in x_a \circ_a y_a} (\bigcup_{v_a \in w_a \circ_a z_a} \{\elemprod{v_a}\}) \\
& = & \bigcup_{v_a \in (x_a \circ_a y_a) \circ_a z_a} \{\elemprod{v_a}\}
\end{array}\]
Similarly, we have:
\[
\elemprod{x_a} \circ (\elemprod{y_a} \circ \elemprod{z_a}) =
\textstyle\bigcup_{v_a \in x_a \circ_a (y_a \circ_a z_a)} \{\elemprod{v_a}\}
\]
whence $(\elemprod{x_a} \circ \elemprod{y_a}) \circ
\elemprod{z_a} = \elemprod{x_a} \circ (\elemprod{y_a} \circ
\elemprod{z_a})$ as required by the associativity of each
$\circ_a$.

Now, to see that $\langle R, \circ, \elemprod{e_a}, \inv,
\elemprod{\infty_a}\rangle$ is a $\CBI$-model, it just remains
to check that the required conditions on $\inv$ and
$\elemprod{\infty_a}$ hold.  We have by definition:
\[\begin{array}{rcl}
\elemprod{x_a} \circ \inv(\elemprod{x_a}) & = & \elemprod{x_a} \circ \elemprod{(\inv_a x_a)} \\
& = & \bigcup_{w_a \in x_a \circ_a (\inv_a x_a)} \{\elemprod{w_a}\}
\end{array}\]
Then, since $\infty_a \in x_a \circ_a (\inv_a x_a)$ for all $a
\in A$ we have $\elemprod{\infty_a} \in \elemprod{x_a} \circ
\inv(\elemprod{x_a})$ as required.  To see that
$\inv(\elemprod{x_a})$ is the unique element of $R$ satisfying
this condition, suppose $\elemprod{\infty_a} \in \elemprod{x_a}
\circ \elemprod{y_a}$.  Then for each $a \in A$ we would have
$\infty_a \in x_a \circ_a y_a$, which implies $y_a =
\inv_a x_a$ for each $a \in A$ and thus $\elemprod{y_a} =
\elemprod{\infty_a}$ as required.  This completes the
verification.
\end{proof}

We remark that, as well as standard Cartesian product
constructions, Lemma~\ref{lem:product_model} gives a canonical
way of extending $\CBI$-models to heap-like structures mapping
elements of an ordered set $A$ into model values by taking
$M_a$ to be the same $\CBI$-model for each $a \in A$. For
example, our ``money'' model of Example~\ref{ex:money} extends
via Lemma~\ref{lem:product_model} to a model of maps from a set of identifiers
to the integers, which can be understood as financial ``asset
portfolios'' mapping identifiers (commodities) to integers
(assets or liabilities).  Such a model might potentially form
the basis of a Hoare logic for financial transactions in the
same way that the heap model of $\BBI$ underpins separation
logic.  The following example shows another application.

\begin{example}[Deny-guarantee model]
\label{ex:deny_guarantee} The \emph{deny-guarantee} permissions
employed by Dodds et al.~\cite{Dodds-etal:09} are elements of
$\mathrm{PermDG} = \mathrm{Actions} \rightarrow
\mathrm{FractionDG}$, where $\mathrm{Actions}$ is a set of
``actions'' and: \[\mathrm{FractionDG} = \{(deny,\pi) \mid \pi
\in (0,1)\} \cup \{(guar,\pi) \mid \pi \in (0,1)\} \cup
\{0,1\}\]
A partial binary function $\oplus$ is defined on
$\mathrm{FractionDG}$ by:
\[\begin{array}{rcl}
0 \oplus x = x \oplus 0 & = & x \\
(deny,\pi_1) \oplus (deny,\pi_2)& = & \left\{
\begin{array}{ll}
(deny,\pi_1 + \pi_2) & \mbox{if $\pi_1 + \pi_2 < 1$} \\
1 & \mbox{if $\pi_1 + \pi_2 = 1$} \\
\mbox{undefined} & \mbox{otherwise}
\end{array}\right.         \\
(guar,\pi_1) \oplus (guar,\pi_2) & = & \left\{
\begin{array}{ll}
(guar,\pi_1 + \pi_2) & \mbox{if $\pi_1 + \pi_2 < 1$} \\
1 & \mbox{if $\pi_1 + \pi_2 = 1$} \\
\mbox{undefined} & \mbox{otherwise}
\end{array}\right.        \\
1 \oplus x = x \oplus 1 & = & \mbox{undefined for $x \neq 0$}
\end{array}\]
The operation $\oplus$ is lifted to $\mathrm{PermDG}$ by $(p_1
\oplus p_2)(a) = p_1(a) \oplus p_2(a)$.  Next, define the
involution $\inv$ on $\mathrm{FractionDG}$ by:
\[
\inv 0 = 1 \quad \inv(deny,\pi) = (deny,1-\pi) \quad \inv(guar,\pi) = (guar,1-\pi) \quad \inv 1 = 0
\]
and lift $\inv$ to $\mathrm{PermDG}$ by $(\inv p)(a) = \inv
p(a)$.  Finally, we lift $0$ and $1$ to $\mathrm{PermDG}$ by
$0(a) = 0$ and $1(a) = 1$.

Then $\langle \mathrm{PermDG},\oplus,0,\inv,1\rangle$ is a
$\CBI$-model.  One can check this directly, but we can also
reconstruct the model using our general constructions.  First,
one verifies easily that both the ``deny fragment'' and the
``guarantee fragment'' of $\mathrm{FractionDG}$ given by the tuples:
\[\begin{array}{c}
\langle \{(deny,\pi) \mid \pi \in (0,1)\} \cup
\{0,1\},\oplus,0,\inv,1\rangle \\ \langle \{(guar,\pi) \mid
\pi \in (0,1)\} \cup \{0,1\},\oplus,0,\inv,1\rangle
\end{array}\] are
$\CBI$-models. Noting that $1$ is nonextensible in both models,
we can apply Lemma~\ref{lem:disjoint_union_model} to obtain the
disjoint union of these models, which is exactly
$\langle\mathrm{FractionDG},\oplus,0,\inv,1\rangle$.  By
applying Lemma~\ref{lem:product_model} (taking $A$ to be $\mathrm{Actions}$ and $M_a$ to be
$\langle\mathrm{FractionDG},\oplus,0,\inv,1\rangle$ for all $a
\in \mathrm{Actions}$) we then obtain the $\CBI$-model
$\langle\mathrm{PermDG},\oplus,0,\inv,1\rangle$.
\end{example}

We end this section by addressing the general question of
whether there are embeddings of arbitrary $\BBI$-models into
$\CBI$-models.  This is not trivial for the following reason.
Consider a $\BBI$-model $\monoid$ with $\circ$ a function and
\mbox{$z = x_1 \circ y = x_2 \circ y$} with $x_1 \neq x_2$.  In
any simple extension of this model into a $\CBI$-model $\langle
R',\circ',e',\inv,\infty\rangle$ with $R \subseteq R'$ and $\circ \subseteq \circ'$, we are forced to have both
\mbox{$\inv x_1 \in y \circ' \inv z$} and \mbox{$\inv x_2 \in y
\circ' \inv z$} by the $\CBI$-model conditions (see
Proposition~\ref{prop:cbimodel_properties},
part~\ref{propitem:cancel}), while $\inv x_1 \neq \inv x_2$.
Thus any such extension of the \emph{functional} $\BBI$-model
$\monoid$ into a $\CBI$-model is forced to be
\emph{relational}.  Our construction below shows how a general
embedding  from $\BBI$-models to $\CBI$-models may be obtained,
which can be viewed as being weakly canonical in the sense that
it is an injection.

\begin{proposition}[$\CBI$-extension of $\BBI$-models]
\label{prop:model_extension} Let $\monoid$ be a $\BBI$-model
and define a second, disjoint copy $\overline{R}$ of $R$ by
$\overline{R} \defeq \{\overline{r} \mid r \in R\}$.  Define
$\inv x = \overline{x}$ for all $x \in R$ and $\inv\overline{x}
= x$ for all $x \in \overline{R}$.  Finally, define the binary
relation $\oplus$ over $R \cup \overline{R}$ by the following:
\[\begin{array}{c@{\hspace{1cm}}r@{\hspace{0.4cm}}c@{\hspace{0.4cm}}l}
(\oplus 1) & z \in x \circ y & \Rightarrow & z \in x \oplus y \\
(\oplus 2) & z \in x \circ y & \Rightarrow & \overline{y} \in (x \oplus \overline{z})
\cap (\overline{z} \oplus x)
\end{array}\]
Then $\langle R \cup \overline{R}, \oplus, e , \inv ,
\overline{e} \rangle$ is a $\CBI$-model.  Moreover, the construction of
$\langle R \cup \overline{R}, \oplus, e , \inv ,
\overline{e} \rangle$ from $\monoid$ is injective.
\end{proposition}

\begin{proof}
We start by stating the following \emph{elimination principle}
for $\oplus$ which follows directly from its introduction rules
$(\oplus 1)$ and $(\oplus 2)$.

\paragraph{\em Elimination principle.}  If $z \in x \oplus y$ then the
following hold:
\begin{enumerate}
\item $z \in R$ iff $x,y \in R$, and if $x,y,z \in R$ then
    $z \in x \circ y$.
\item $z \in \overline{R}$ iff either $x \in R$ and $y \in
    \overline{R}$, or $x \in \overline{R}$ and $y \in R$.
    Furthermore:
    \begin{itemize}
    \item if $x \in R$ and $y,z \in \overline{R}$ then
        $y' \in x \circ z'$, where $\overline{y'} = y$
        and $\overline{z'} = z$;
    \item if $y \in R$ and $x,z \in \overline{R}$ then
        $x' \in z' \circ y$, where $\overline{x'} = x$
        and $\overline{z'} = z$.
    \end{itemize}
\end{enumerate}

With this principle in place we carry out the main proof.
First, we need to check that $\langle R \cup \overline{R},
\oplus, e\rangle$ is a $\BBI$-model, i.e., that $\oplus$ is
commutative and associative, and satisfies $x \circ e = \{x\}$
for all $x \in R \cup \overline{R}$.

We tackle the last of these requirements first.  Since
$\monoid$ is a $\BBI$-model we have $x \in x \circ e = e \circ
x = \{x\}$ for all $x \in R$. Thus, for all $x \in R$, we have
$x \in x \oplus e$ by ($\oplus$1) and $\overline{x} \in
\overline{x} \oplus e$ by ($\oplus$2). That is, $x \in x \oplus
e$ for all $x \in R \cup \overline{R}$.  Now suppose $y \in x
\oplus e$.  Since $e \in R$, there are two cases to consider by
the elimination principle. If both $x,y \in R$ then we have $y
\in x \circ e = \{x\}$, thus $y = x$. Otherwise, both $x,y \in
\overline{R}$ and $x' \in y' \circ e = \{y'\}$, where
$\overline{x'} = x$ and $\overline{y'} = y$. Thus $x' = y'$
and, since $\overline{\,\cdot\,}$ is injective, $x = y$.  So $x
\oplus e = \{x\}$ for all $x \in R \cup \overline{R}$ as
required.

To see that $\oplus$ is commutative, let $z \in x \oplus y$,
and consider the cases given by the elimination principle.
First, suppose that all of $x,y,z \in R$ and $z \in x \circ y$.
Since $\monoid$ is a $\BBI$-model, $\circ$ is commutative, so
$z \in y \circ x$.  Thus by ($\oplus$1) we have $z \in y \oplus
x$. Next, suppose that $x \in R$, $y,z \in \overline{R}$ and
$y' \in x \circ z'$, where $\overline{y'} = y$ and
$\overline{z'} = z$. By ($\oplus$2) we then have $z \in y
\oplus x$. The case where $y \in R$ and $x,z \in \overline{R}$
is symmetric. Thus $z \in x \oplus y$ implies $z \in y \oplus
x$, so $x \oplus y = y \oplus x$ for any $x,y \in R \cup
\overline{R}$, i.e.\ $\oplus$ is commutative.

It remains to show that $\oplus$ is associative, i.e.\ that $(x
\oplus y) \oplus z = x \oplus (y \oplus z)$ for any $x,y,z \in
R \cup \overline{R}$.  We divide into cases as follows:

\paragraph{\em Case: at least two of $x,y,z$ are in
$\overline{R}$.}  The elimination principle implies that $x
\oplus y = \emptyset$ whenever both $x,y \in \overline{R}$ and,
furthermore, $z \in \overline{R}$ whenever $z \in x \oplus y$
and either $x \in \overline{R}$ or $y \in \overline{R}$.
Combined with the pointwise extension of $\oplus$ to sets of
elements, this implies that $(x \oplus y) \oplus z = x \oplus
(y \oplus z) =\emptyset$.

\paragraph{\em Case: none of $x,y,z$ are in $\overline{R}$.}  The
elimination principle implies that $(x \oplus y) \oplus z = (x
\circ y) \circ z$ and $x \oplus (y \oplus z) = x \circ (y \circ
z)$.  We are then done since $\circ$ is associative by assumption.

\paragraph{\em Case: exactly one of $x,y,z$ is in $\overline{R}$.}
We show how to treat the case where $x \in \overline{R}$; the
other cases are similar.  We write $x = \overline{x'}$.  Let $w
\in (\overline{x'} \oplus y) \oplus z = \bigcup_{v \in
\overline{x'} \oplus y} v \oplus z$.  Thus $w \in v \oplus z$
for some $v \in \overline{x'} \oplus y$.  By part 2 of the
elimination principle, $v \in \overline{R}$ and $x' \in v'
\circ y$, where $v = \overline{v'}$.  Applying the same
elimination principle to $w \in \overline{v'} \oplus z$, we
obtain that $w \in \overline{R}$ and $v' \in w' \circ z$, where
$w = \overline{w'}$.  Thus $x' \in \bigcup_{v' \in w' \circ z}
v' \circ y = (w' \circ z) \circ y$.  Since $\circ$ is
associative and commutative, $x' \in w' \circ (y \circ z)$.  By
($\oplus$1), it is certainly the case that $y \circ z \subseteq y
\oplus z$, whence we obtain $x' \in w' \circ (y \oplus z) = (y
\oplus z) \circ w'$. Thus, by ($\oplus$2), we obtain $w \in x
\oplus (y \oplus z)$.

As we have shown $w \in (x \oplus y) \oplus z)$ implies $w \in
x \oplus (y \oplus z)$, we conclude $(x \oplus y) \oplus z = x
\oplus (y \oplus z)$, i.e.\ $\oplus$ is associative as
required. Thus $\langle R \cup \overline{R}, \oplus, e\rangle$
is indeed a $\BBI$-model.

To see that $\langle R \cup \overline{R}, \oplus, e , \inv ,
\infty\rangle$ is a $\CBI$-model, we just need to check that
for any $x \in R \cup \overline{R}$, $\overline{x}$ is the
unique element such that $\infty = \overline{e} \in x \oplus
\overline{x}$. Suppose first that $x \in R$. Since $x \circ e =
\{x\}$, we have
$\overline{e} \in x \oplus \overline{x}$ by ($\oplus$2).  To see that $\overline{x}$ is
unique, suppose that $\overline{e} \in x \oplus y$. By part 2
of the elimination principle, we must have $y = \overline{y'}$
and $y' \in x \circ e = \{x\}$. Thus $y' = x$ so $y =
\overline{x}$ as required. When $x \in \overline{R}$, we have
$x = \overline{y}$ for some $y \in R$ and the reasoning is
exactly dual to the case above, since $\circ$ is commutative.
This completes the proof.
\end{proof}

Another interesting possibility for obtaining $\CBI$-models
from arbitrary $\BBI$-models would be to extend the well-known
Grothendieck completion --- which constructs the canonical Abelian
group corresponding to a total commutative monoid --- to the relational setting.  From a
category-theoretic perspective, it would be interesting to see
whether the obvious forgetful functor from $\CBI$-models to
$\BBI$-models has a left-adjoint, which would give the truly
canonical $\CBI$-model corresponding to any $\BBI$-model.

\section{Related and future work}
\label{sec:conclusion}

We consider related work, and directions for future work, from
several perspectives.

\paragraph{\em Bunched logics:}
In his monograph on $\BI$~\cite{Pym:02}, Pym observed that it
made sense to think not of one bunched logic but rather a
family of bunched logics, characterised by the strengths of
their additive and multiplicative components. We reprise his
diagram of the bunched logic family, suitably updated, in
Figure~\ref{fig:bunched_logics}.  $\CBI$ is the strongest
member of this family, boasting two classical negations and
being characterised by an underlying Boolean algebra in its
additive component and a de Morgan algebra in the
multiplicative component.  Indeed, Pym anticipated the
formulation of $\CBI$ as presented here in at least two
important respects: he observed that a relevantist
approach to multiplicative negation (which we take by using the
involution operation `$\inv$' in our models in place of the
Routley star) is classically compatible with the other
multiplicative connectives; and he noted the
problems with cut-elimination seemingly inherent in a two-sided
sequent calculus for bunched logic.  In this paper,
we provide two key missing links.  First, our display calculus
$\displayCBI$ and its cut-elimination theorem, obtained by
following Belnap's original methodology for display logic~\cite{Belnap:82}, provides a well-behaved proof theory for
$\CBI$.  (Subsequently, the first author has given in~\cite{Brotherston:10} a unified
display calculus proof theory for all four bunched logics in
Figure~\ref{fig:bunched_logics}.)
Second, and perhaps more importantly, we also provide the
connection to Kripke-style resource models with
precisely the structure necessary to interpret $\CBI$. Our
soundness and completeness results establishing the
correspondence between validity and provability, plus cut-elimination for
$\displayCBI$, can be taken as strong evidence that the
formulation of $\CBI$ we present here may be considered
canonical.  We also establish nonconservativity of $\CBI$ over $\BBI$, and its incompleteness with respect to partial functional models.

We remark that the bunched logic $\mathrm{dMBI}$ (standing for
``de Morgan $\BI$'') in the diagram, which combines
intuitionistic additives with classical multiplicatives, has
not been investigated in any great detail, to our knowledge,
but it is closely related to the relevant logic $\mathrm{RW}$.
See the section on relevant logics below for a comparison.

\begin{figure}
\begin{center}
\begin{picture}(100,80)(0,0)

  \gasset{Nframe=n,Nadjust=w,Nh=15,Nmr=2,AHLength=3,ATLength=3}

  \node(BI)(50,5){\begin{tabular}{c}$\mathbf{BI}$\\(Heyting, Lambek) \\ \emph{decidable}~\cite{Galmiche-Mery-Pym:05} \end{tabular}}
  \node(BBI)(90,40){\begin{tabular}{c}$\mathbf{BBI}$ \\ (Boolean, Lambek) \\ \emph{undecidable}~\cite{Brotherston-Kanovich:10,Larchey-Wendling-Galmiche:10} \end{tabular}}
  \node(CBI)(50,75){\begin{tabular}{c}$\mathbf{CBI}$ \\ (Boolean, de Morgan) \\ \emph{undecidable}~\cite{Brotherston-Kanovich:10} \end{tabular}}
  \node(IBI)(10,40){\begin{tabular}{c}{$\mathbf{dMBI}$} \\ (Heyting, de Morgan) \\ \end{tabular}}

  \drawedge[ELside=r,dash={0.7}0,exo=5](BI,BBI){$\neg$}
  \drawedge[dash={0.7}0,exo=-5](BI,IBI){$\mneg$}
  \drawedge[ELside=r,dash={0.7}0,sxo=5](BBI,CBI){$\mneg$}
  \drawedge[dash={0.7}0,sxo=-5](IBI,CBI){$\neg$}
\end{picture}
\end{center}

\caption{The bunched logic family. The (additive,
multiplicative) subtitles denote the strength of the underlying
additive and multiplicative algebras.  The arrows denote the
addition of either additive ($\neg)$ or multiplicative
($\mneg$) classical negation. \label{fig:bunched_logics}}
\end{figure}

\paragraph{\em Relevant logics:}
$\CBI$, like its bunched logic predecessors, owes a historical
debt to the extensive work on relevant logics and takes many of
its mathematical cues from the development of these logics, as
described in the case of $\BI$ by O'Hearn and
Pym~\cite{OHearn-Pym:99}.  Indeed, as they point out, if one
understands by ``relevant logics'' nothing but logics whose
logical connectives are understood primarily in terms of the
structural rules which they must respect (cf.~\cite{Read:87}),
then bunched logics \emph{are} relevant logics.  However, in
bunched logics, the philosophical ideal of
relevance has been entirely sacrificed in favour of full-strength
additives as equal partners alongside the multiplicatives.
The justification for doing so is semantic; in the Kripke models of bunched logics, one has
a simple truth reading of formulas in terms of resources, in which the additives have their standard meanings.  In other words, while relevant logic
seeks to exclude the paradoxes of material implication, in the setting of bunched logic we regard these paradoxes as being perfectly justifiable in terms of our resource models.

Retrospectively, $\CBI$ can be obtained in terms of relevant logics
by a series of surgeries on the axiomatisation of the full
system $\mathbf{R}$ and its corresponding class of Kripke
models (see e.g.~\cite{Restall:00}) in the following way.
First, drop the axiom of multiplicative contraction from
$\mathbf{R}$ (so that the corresponding condition $Rxxx$ on the
ternary relation $R$ in the Kripke models of $\mathbf{R}$ does
not necessarily hold) to obtain the well-known relevant logic
$\mathbf{RW}$ (a.k.a.\ $\mathbf{C}$). Then one can add both the
additive \emph{intuitionistic} implication $\rightarrow$ and
falsum $\false$, which are barred from relevant logics in order
to exclude various logical principles which contravene the
philosophical notion of relevance (e.g.\ the classical
tautology $A \wedge B \rightarrow A$).  This addition is
\emph{conservative} over the language of $\mathbf{RW}$ because
$\rightarrow$ and $\false$ can already be interpreted in its
Kripke models using the ordering $\leq$ on points in the model
in the usual intuitionistic way (a fact exploited by Restall in
order to formulate display calculi for $\mathbf{RW}$ and other
relevant logics~\cite{Restall:98}). At this point we have
obtained a characterisation of the bunched logic
$\mathrm{dMBI}$, whence to obtain $\CBI$ we strengthen the
implication $\rightarrow$ into the (additive) \emph{classical}
implication, which corresponds to taking $\leq$ in the
corresponding Kripke models to be the identity ordering.  The
situation is also similar to that for the \emph{classical
relevant logics} introduced by Meyer and
Routley~\cite{Meyer-Routley:73,Meyer-Routley:74}, which feature
traditional Boolean negation alongside the relevantist negation
employing the Routley star --- though, again, multiplicative
contraction must be removed and the additives given their full
classical strength.

Similarly, it would not surprise a relevantist that $\CBI$ can be given a display calculus presentation, as display logic historically served as one of the main
proof-theoretic tools in formulating sensible proof systems for
relevant and other substructural logics.  Indeed, one might
deduce that this was the main intention behind Belnap's
original formulation of display logic~\cite{Belnap:82}, in
which the choice of structural rules for a particular logic are
identified as the principal factor affecting cut-elimination.
We note that Gor\'e has shown how to automatically generate
display calculi for a general class of substructural logics
based on Dunn's \emph{gaggle theory}~\cite{Gore:98-2}, and it
seems more than likely that his techniques could equally well
be used to obtain $\displayCBI$.  Similarly, the correct
formulation of $\displayCBI$ could have been deduced from
Restall's display calculi for the relevant logic $\mathbf{DW}$
and its various extensions including
$\mathbf{RW}$~\cite{Restall:98}.  In both the aforementioned
cases, however, the modelling power is in considerable excess
of what is needed to obtain our display calculus for $\CBI$,
which falls directly under Belnap's original description of
displayable logics in~\cite{Belnap:82} because it features
classical negation in both its additive and multiplicative
connective families.

\paragraph{\em Linear logic:} Readers may wonder
about the relationship between $\CBI$ and classical linear
logic ($\CLL$), which also features a full set of propositional
multiplicative connectives, and is a nonconservative extension
of intuitionistic linear logic ($\ILL$)~\cite{Schellinx:91}.
The differences between the two are intuitively obvious when
comparing our money model of $\CBI$ (Example~\ref{ex:money})
alongside Girard's corresponding Marlboro / Camel
example~\cite{Girard:95}.  In particular, formulas in our model
are read as declarative statements about resources (i.e.\
money), whereas linear logic formulas in Girard's model are
typically read as procedural statements about actions. Compared
to $\CLL$, $\CBI$ has the advantage of a simple, declarative
notion of truth relative to resource, but this advantage
appears to come at the expense of $\CLL$'s constructive
interpretation of proofs.

Of course, the typical reading of $\BI$ departs from that of
$\ILL$ in a similar way (see~\cite{OHearn-Pym:99} for a
discussion), and indeed it seems that the main differences
between $\CBI$ and $\CLL$ are inherited from the wider
differences between bunched logic and linear logic in general.
These differences are not merely conceptual, but are also
manifested at the technical level of logical consequence. For
example, $\seq{P \linimp Q}{P \rightarrow Q}$ is a theorem of
linear logic for any propositions $P$ and $Q$, via the encoding
of additive implication $P \rightarrow Q$ as $!P \linimp Q$,
but $\seq{P \wand Q}{P \rightarrow Q}$ is not a theorem of
bunched logic.  Similarly, distributivity of additive
conjunction $\wedge$ over additive disjunction $\vee$ holds in
bunched logics, but fails in linear logics.  Further differences
are highlighted in~\cite{Brotherston-Kanovich:10}.

Interestingly, however, there is an intersection between
$\CBI$-models and the $\CLL$-models obtained from the
\emph{phase semantics} of classical linear
logic~\cite{Girard:95}. A $\CBI$-model $\cbimodel$ in which the
monoid operation $\circ$ is a total function, rather than a
relation, is a special instance of a phase space, used to
provide a phase model of $\CLL$.  This can be seen by taking
the linear logic ``perp'' $\false$ to be the set $R \setminus
\{\infty\}$, whence the linear negation $X^\perp$ on sets $X
\subseteq R$ becomes $\inv X$.  In the linear logic
terminology, every subset $X$ of $R$ is then a ``fact''  in the
sense that $(X^\false)^\false = \inv\inv X = X$. It seems
somewhat curious that there is a subclass of models where
$\CBI$ and $\CLL$ agree, since known interesting phase models
of linear logic are relatively few whereas there appear to be
many interesting $\CBI$-models (cf.\
Section~\ref{sec:examples}).  However, one can argue that this
subclass is faithful to the spirit of neither logic. On the one
hand, the restriction to a total monoid operation in
$\CBI$-models rules out many natural examples where resource
combination is partial (or indeed relational).  On the other
hand, it seems certain that the induced subclass of $\CLL$
phase models will be at odds with the coherence semantics of
$\CLL$ proofs.

\paragraph{\em Applications:}
The main application of $\BBI$ so far has been the use of
separation logic in program analysis. There are now several
program analysis
tools~\cite{Calcagno-Parkinson-Vafeiadis:07,Chang-Rival:08,Distefano-Parkinson:08,Yang-etal:08,Nguyen-Chin:08}
which use logical and semantic properties of the heap model of
$\BBI$ at their core. These tools typically define a suitable
fragment of separation logic with convenient algebraic
properties, and use it in custom lightweight theorem provers
and abstract domains. We suggest that our work on $\CBI$ could
be relevant in this area as a foundation for richer resource
models.  In this paper we have already given several new models
and model constructions which, though relatively simple in
their present form, are suggestive of the applicability of
$\CBI$ to more complex domains (cf.\
Section~\ref{sec:examples}).  In particular, we have observed
that several models introduced recently for reasoning about
concurrent access to resources are $\CBI$ models, e.g.\
fractional permissions as used in deny-guarantee reasoning
(cf.\ Example~\ref{ex:deny_guarantee}).

More speculatively, our display calculus $\displayCBI$ might
form a basis for the design of new theorem provers, which could
easily employ the powerful (and historically difficult to use)
implication $\wand$ since, in $\CBI$, it can be reexpressed
using more primitive connectives. Moreover, the notion of
dual or negative resource might be employed in extended theorem proving
questions, such as the frame inference problem $F \vdash G * X$
where the frame $X$ is computed essentially by subtracting $G$
from $F$.  A similar problem is the bi-abduction question,
which forms the basis of the compositional shape analysis
in~\cite{Calcagno-etal:09} and has the form $\seq{F * X}{G
* Y}$, interpreted as an obligation to find formulae to
instantiate $X$ and $Y$ such that the implication holds.  This
question arises at program procedure call sites, where $F$ is
the procedure's precondition, $G$ is the current precondition
at the call point, $X$ is the resource missing, and $Y$ is the
leftover resource. We speculate that such inferences could be
explained in terms of an ordinary proof theory, providing that
multiplicative negation is supported, as in $\CBI$.

Finally, $\CBI$ could be applied to the study of other logics.
For example, Kleene's 3-valued logic~\cite{Kleene:87} can be
modelled using a subset of $\CBI$'s connectives. Consider the
two-element $\CBI$ model given by
$\langle\{e,\infty\},\circ,e,\inv,\infty\rangle$, where $\infty
\circ \infty = \emptyset$ (note that $\circ$ and $\inv$ are
then determined by the $\CBI$-model axioms). There are
$\CBI$-formulas denoting each of the subsets of $\{e,\infty\}$:
$\true$, $\false$, $\mtrue$, $\infty$ (where $\infty$ is used
as an abbreviation for $\neg \mfalse$). To model 3-valued logic
we focus on $\true$, $\false$, $\infty$, with $\infty$ playing
the role of the third logical value, ``unknown''. A direct
calculation shows that the connectives $\wedge$, $\vee$, and
$\mneg$ indeed generate the truth tables required by 3-valued
logic. For example, we have $\infty \vee \mneg \infty = \infty
\vee \infty = \infty$. We speculate that $\CBI$ could be
applied to other situations in logic where a non-standard
notion of negation is used.

We believe that, aside from its intrinsic technical interest,
our development of $\CBI$ contributes to the picture of bunched
logic and its connections to computer science as a whole, as
well as to the broader area of substructural logics in general.
Although our suggestions regarding specific applications of
$\CBI$ are necessarily still somewhat speculative at this early
stage in its existence, we hope that the foundations
established in this paper will provide a solid platform upon
which such applications can, in time, be constructed.

\subsubsection*{Acknowledgements}
We extend special thanks to Peter O'Hearn and David Pym for
many interesting and enlightening discussions which informed
the present paper. We also thank Byron Cook, Ross Duncan,
Philippa Gardner, Greg Restall, Sam Staton, Alex Simpson, and
the members of the East London Massive for useful discussions
and feedback.  Finally, thanks to the three anonymous referees
for their useful suggestions on improving the paper.

\appendix

\section{Cut-elimination for $\displayCBI$ \texorpdfstring{(Theorem~\ref{thm:DLBI_cut_elim})}{(Theorem 3.8)}}
\label{app:DLBI_cut_elim}
\renewcommand{\thetheorem}{\thesection.\arabic{theorem}}
\setcounter{theorem}{0}

The following definition is taken from Belnap~\cite{Belnap:82}. By a
\emph{constituent} of a structure or consecution we mean an
occurrence of one of its substructures.

\begin{defn}[Parameters / congruence]
\label{defn:display_congr} Let $I$ be an instance of a rule $R$ of
$\displayCBI$.  Note that $I$ is obtained by assigning structures to
the structure variables occurring in $R$ and formulas to the formula
variables occurring in $R$.

Any constituent of the consecutions in $I$ occurring as part of
structures assigned to structure variables in $I$ are defined to be
\emph{parameters} of $I$.  All other constituents are defined to be
\emph{non-parametric} in $I$, including those assigned to formula
variables.

Constituents occupying similar positions in occurrences of
structures assigned to the same structure variable are defined to be
\emph{congruent} in $I$.
\end{defn}

We remark that congruence as defined above is an equivalence
relation.

Belnap's analysis guarantees cut-elimination for $\displayCBI$
(Theorem~\ref{thm:DLBI_cut_elim}) provided its proof rules (cf.\ Figure~\ref{fig:logical_rules}) satisfy the
following conditions, which are stated with reference to an
instance $I$ of a $\displayCBI$ rule $R$. (Here, following
Kracht~\cite{Kracht:96}, we state a stronger, combined version
of Belnap's original conditions C6 and C7, since our rules
satisfy this stronger condition.)  In each case, we indicate
how to verify that the condition holds for our rules.
\begin{description}
\item[C1] \emph{Preservation of formulas.}  Each formula
    which is a constituent of some premise of $I$ is a
    subformula of some formula in the conclusion of $I$. \\

    \emph{Verification.} One observes that, in each rule,
    no formula variable or structure variable is lost when
    passing from the premises to the conclusions. \\

\item[C2] \emph{Shape-alikeness of parameters.}  Congruent
    parameters are occurrences of the same structure.\\

    \emph{Verification.} Immediate from the definition of
    congruence. \\

\item[C3] \emph{Non-proliferation of parameters.}  No two
    constituents in the conclusion of $I$ are congruent to
    each other.\\

    \emph{Verification.} One just observes that, for each
    rule, each structure variable occurs exactly once in
    the conclusion. \\

\item[C4] \emph{Position-alikeness of parameters.}
    Congruent parameters are either all antecedent or all
    consequent parts of their respective consecutions.\\

    \emph{Verification.} One observes that, in each rule,
    no structure variable occurs both as an antecedent part
    and a consequent part. \\

\item[C5] \emph{Display of principal constituents.}  If a
    formula is nonparametric in the conclusion of $I$, it
    is either the entire antecedent or the entire
    consequent of that conclusion.  Such a formula is said
    to be \emph{principal} in $I$.\\

    \emph{Verification.} It is easy to verify that the only
    non-parametric formulas in the conclusions of our rules
    are the two occurrences of P in $\ax$ and those
    occurring in the introduction rules for the logical
    connectives, which
    obviously satisfy the condition. \\

\item[C6/7] \emph{Closure under substitution for
    parameters.}  Each rule is closed under simultaneous
    substitution of arbitrary structures for congruent
    formulas which are parameters.\\

    \emph{Verification.} This condition is satisfied
    because no restrictions are placed on the structural
    variables used in our rules. \\

\item[C8] \emph{Eliminability of matching principal
    formulas.}  If there are inferences $I_1$ and $I_2$
    with respective conclusions $\seq{X}{F}$ and
    $\seq{F}{Y}$ and with $F$ principal in both inferences,
    then either $\seq{X}{Y}$ is equal to one of
    $\seq{X}{F}$ and $\seq{F}{Y}$, or there is a derivation
    of $\seq{X}{Y}$ from the premises of $I_1$ and $I_2$ in
    which every instance of cut has a cut-formula which is
    a proper subformula of $F$.\\

    \emph{Verification.} There are only two cases to
    consider.  If $F$ is atomic then $\seq{X}{F}$ and
    $\seq{F}{Y}$ are both instances of $\ax$.  Thus we must
    have $\seq{X}{F} = \seq{F}{Y} = \seq{X}{Y}$, and are
    done. Otherwise $F$ is non-atomic and introduced in
    $I_1$ and $I_2$ respectively by the right and left
    introduction rule for the main connective of $F$.  In
    this case, a derivation of the desired form can be
    obtained using only the display rule $\display$ and
    cuts on subformulas of $F$.  For example, if the
    considered cut is of the form:

\[\begin{prooftree}
\[
\[\leadsto \seq{X}{F , G} \]
\justifies \seq{X}{F \mor G} \using \morr \]
\[
\[\leadsto \seq{F}{Y} \]
\phantom{w}
\[\leadsto \seq{G}{Z} \]
\justifies
\seq{F \mor G}{Y,Z} \using \morl \]
\justifies
\seq{X}{Y,Z} \using \cut
\end{prooftree}\]\vspace{0.2cm}

\noindent then we can reduce this cut to cuts on $F$ and
$G$ in the following manner:

\[\begin{prooftree}
\[\[\[\[
\[\leadsto \seq{X}{F,G} \]
\justifies
\seq{X,\minv G}{F} \using \display \]
\[\leadsto \seq{F}{Y} \]
\justifies \seq{X,\minv G}{Y}
\using \cut \]
\justifies
\seq{X,\minv Y}{G} \using \display \]
\[\leadsto \seq{G}{Z}\]
\justifies
\seq{X,\minv Y}{Z} \using \cut \]
\justifies \seq{X}{Y,Z} \using \display
\end{prooftree}\] \vspace{0.2cm}

The cases for the other connectives are similarly
straightforward.  This completes the verification of the
conditions, and thus the proof. \qed
\end{description}

\end{document}